\documentclass[12pt,onecolumn,a4paper]{article}
\usepackage{amsmath,amssymb}
\usepackage{booktabs}
\usepackage[OT2,T1]{fontenc} 
\usepackage{colortbl} 
\usepackage{tabularx}
\usepackage{setspace}
\usepackage{tikz-cd}
\usepackage{verbatim}
\usepackage{subcaption}
\usetikzlibrary{external}

\usepackage{algorithm}
\usepackage{algpseudocode}
\usepackage{amsmath}
\usepackage{authblk}

\usepackage[braket, qm]{qcircuit}
\usepackage{enumerate}
\usepackage{hyperref}
\usepackage{ulem}
\usepackage{mathtools}
\mathtoolsset{showonlyrefs}
\usepackage{amsmath}
\usepackage{graphicx}
\usepackage{enumerate}
\usepackage{amssymb}
\usepackage{color}
\usepackage{booktabs}
\usepackage{xcolor}
\usepackage{graphicx}
\usepackage{diagbox}

\usepackage{amsthm,amsmath}
\usepackage[left=0.75in,right=0.75in,top=0.75in,bottom=0.85in]{geometry}
\newcommand{\appref}[1]{\hyperref[#1]{{Appendix~\ref*{#1}}}}
\newcommand{\be}{\begin{eqnarray} \begin{aligned}}
\newcommand{\ee}{\end{aligned} \end{eqnarray} }
\newcommand{\benn}{\begin{eqnarray*} \begin{aligned}}
\newcommand{\eenn}{\end{aligned} \end{eqnarray*}}
\newcommand*{\textfrac}[2]{{{#1}/{#2}}}

\newcommand*{\bbR}{\mathbb{R}}

\newcommand*{\cB}{\mathcal{B}}

\newcommand*{\cF}{\mathcal{F}}
\newcommand*{\cG}{\mathcal{G}}
\newcommand*{\cH}{\mathcal{H}}

\newcommand*{\cJ}{\mathcal{J}}

\newcommand*{\cL}{\mathcal{L}}

\newcommand*{\cS}{\mathcal{S}}
\newcommand*{\cU}{\mathcal{U}}
\newcommand*{\cT}{\mathcal{T}}

\newcommand*{\cW}{\mathcal{W}}

\newcommand*{\supp}{\mathrm{supp}}

\newcommand{\bc}{\begin{center}}
\newcommand{\ec}{\end{center}}

\newtheorem{theorem}{Theorem}[section]
\newtheorem{lemma}[theorem]{Lemma}

\newtheorem{definition}[theorem]{Definition}

\newtheorem{corollary}[theorem]{Corollary}

\usepackage{amsfonts}

\def\01{\{0,1\}}

\newcommand{\proj}[1]{|#1\rangle\langle#1|}

\newcommand*{\ExpE}{\mathbb{E}}

\newcommand*{\CZ}{\mathsf{CZ}}

\newcommand{\base}[2]{[{#2}]}

\newcommand*{\Uelem}[2]{\cal{U}_{\mathsf{elem}}^{#1,#2}} 
\newcommand*{\Uelemunbounded}[2]{\cal{U}^{#1,#2}} 
\DeclareSymbolFont{cyrletters}{OT2}{wncyr}{m}{n}
\DeclareMathSymbol{\Sha}{\mathalpha}{cyrletters}{"58}

\newcommand*{\size}{\mathsf{size}}
\newcommand{\gateerror}[1]{\mathsf{err}_{#1}}
\newcommand*{\post}{\mathsf{post}}

\newcommand*{\Var}{\sigma^2}
\newcommand*{\Varroot}{\sigma}

\newcommand*{\symradius}{\mathsf{symradius}}

\newcommand*{\energy}{\mathsf{energy}}
\newcommand*{\energytotal}{\mathsf{energy}_{\mathsf{tot}}}
\newcommand*{\tr}{\mathsf{tr}}

\newcommand*{\encmapgkp}{\mathsf{Enc}}

\renewcommand{\cal}[1]{\mathcal{#1}}
\newcommand{\bb}[1]{\mathbb{#1}}

\newcommand*{\bittransfer}[2]{\mathsf{Transf}^{#1}_{#2}}

\newcommand*{\vac}{\mathsf{vac}}
\newcommand*{\discretize}{\mathsf{discretize}}
\newcommand*{\sampscheme}{\mathsf{SAMP}}
\newcommand*{\poly}{\mathsf{poly}}
\newcommand*{\logicalgates}{\cG_n}

\newcommand{\gkpcoderect}[3]{\Sha\cal{GKP}_{#1}^{#2}[#3]}

\renewcommand{\ctrl}{\mathsf{ctrl}} 
\newcommand*{\subsset}{\mathsf{Subs}}

\newcommand*{\squeezingparam}{\overline{g}}

\newcommand*{\logicalQ}{\overline{Q}}
\newcommand*{\diam}{\mathsf{diam}}
\usepackage{authblk}

\usepackage{hyperref}
\hypersetup{
   colorlinks,
   citecolor=blue,
   filecolor=blue,
   linkcolor=blue,
   urlcolor=blue
}

\begin{document}

\title{Trading modes against energy}
\author{Lukas Brenner}
\author{Beatriz Dias}
\author{Robert Koenig}
\affil{Department of Mathematics, School of Computation, Information and Technology, \\ Technical University of Munich, 85748 Garching, Germany}
\affil{Munich Center for Quantum Science and Technology, 80799 Munich, Germany}
\date{\today}
\maketitle

\begin{abstract}
We ask how much energy is required to 
weakly simulate an~$n$-qubit quantum circuit (i.e., produce samples from its output distribution) by a unitary circuit 
in a hybrid qubit-oscillator model. The latter consists of a certain number of bosonic modes coupled to a constant  number of qubits by a Jaynes-Cummings Hamiltonian.  
We find that efficient approximate weak simulation of  an~$n$-qubit quantum circuit  of polynomial size with inverse polynomial error is possible  with
 (1) a linear number of bosonic modes and  a polynomial amount of energy,  or
 (2) a sublinear (polynomial) number of modes and a subexponential amount of energy, or
 (3) a constant number of modes and an exponential amount of energy. Our construction encodes qubits  into high-dimensional approximate Gottesman-Kitaev-Preskill (GKP) codes. It provides new insight into the trade-off between system size (i.e., number of modes) and the amount of energy required to perform quantum computation in the continuous-variable setting.
\end{abstract}

\tableofcontents

\section{Introduction}
Ever since the discovery of the first quantum algorithms, the question of which physical systems are 
most suited for realizing universal quantum computation has been under intense debate. There are now a handful of  competitive contenders which are being pursued experimentally, but the jury is still out on which approach is most promising in the long run. Furthermore, even when focusing on a single concrete physical platform, it often remains challenging to figure out how to best use the available resources.

A major reason for the difficulties arising when trying to use naturally occurring quantum systems for computation is the fact that these are typically associated with infinite-dimensional Hilbert spaces. In contrast, fundamental quantum algorithms are typically phrased in terms of qubits as basic building blocks.  The idealization of qubits as information carriers, and the fact that  multi-qubit operations can be approximated by finite universal gate sets (as shown by Solovay and Kitaev~\cite{Kitaev_1997}), 
is highly convenient  from several perspectives. For example, it brings significant simplifications to the problem of realizing fault-tolerant, i.e., noise-resilient computations by allowing to focus on a number of basic primitives such as magic state distillation~\cite{knill2004faulttolerantpostselectedquantumcomputation,BravyiKitaevMagic}. It facilitates the design of quantum algorithms, e.g., for computational problems arising in discrete mathematics. In addition, it is also of key importance when  trying to assess the power of quantum computing, i.e., when  studying questions related to computational complexity.  By the discrete nature of qubit-based computations 
(manifested, e.g., in efficient circuit descriptions), it can naturally be related and compared to basic (classical) computational models appearing in theoretical (classical) computer science.

Proposals for how to emulate the behavior of a qubit-based quantum computer by using infinite-dimensional systems (also referred to as oscillators or bosonic modes in the following) have been studied early on. A central goal here is to identify a set of elementary operations which 
\begin{enumerate}[(i)]
\item\label{it:computationallyuniversal}
 allow for universal quantum computation and
\item\label{it:experimentallyfeasible}
 which are experimentally feasible, i.e.,  realizable by basic physical components.
\end{enumerate}
 For example, bosonic linear optics operations (i.e., Gaussian unitaries and measurements applied to Gaussian states) clearly satisfy property~\eqref{it:experimentallyfeasible} with generators given by basic linear optics elements such as half-way mirrors. Unfortunately, however, bosonic linear optics operations do not satisfy property~\eqref{it:computationallyuniversal}: As shown by Knill, Laflamme and Milburn~\cite{KnillLaflammeMilburn2001}, 
corresponding computations can efficiently be simulated classically. Under standard complexity-theoretic assumptions, this means that these operations are not (quantum) computationally universal.

To meet both requirements~\eqref{it:computationallyuniversal} and~\eqref{it:experimentallyfeasible}, various models for CV quantum computation have been proposed which extend linear optics by different non-Gaussian operations. These schemes differ in the set of experimentally allowed operations, as well as the associated resource requirements. For example,
Cerf, Adami, and Kwiat~\cite{cerfadamikwiat98} gave a protocol for realizing an~$n$-qubit quantum computation by Gaussian operations, single-photon states and photon number counting (i.e., number-state measurements). This scheme requires an exponential number of modes, and thus suffers from a lack of scalability.

Knill, Laflamme and Milburn~\cite{KnillLaflammeMilburn2001} (KLM) subsequently gave a  protocol for universal computation based on
Gaussian operations and adaptive photon counting (i.e., photon number measurements). We refer to~\cite{KokMunroNemotoRaphDowling07} for a review of this and related protocols. The KLM protocol brings the required number of modes to simulate~$n$~qubits down to a polynomial in~$n$. It also motivated the complexity-theoretic result~\cite{AaronsonArkhipov}
 (now known under the term boson sampling), where evidence for the computational power of a computational model based on analogous circuits but with non-adaptive measurements was provided (see e.g., \cite{madsen2022quantum} for experimental work in this direction). 

In a different direction, Lloyd and Braunstein~\cite{LloydBraunstein99} argued that combining Gaussian operations (beam splitters, phase shifters and squeezers) with  evolution under a non-linear Kerr-type (or, in fact, any higher-order) Hamiltonian provides computational universality. Their arguments center on the Lie algebra generated by such evolutions and are thus primarily a proof-of-principle demonstration. In particular, no explicit procedure for translating a multi-qubit computation into the CV setting is provided. A rigorous analysis of this approach was recently given by Chabaud et al.~\cite{chabaud2024bcomplexity}. In particular, the authors of~\cite{chabaud2024bcomplexity} show that a model of CV quantum computation (more precisely, a certain complexity class) based on Gaussian unitary operations, evolution under a certain cubic Hamiltonian and number state measurements contains the class BQP of bounded-error polynomial-time quantum computation. In other words, these operations  provide quantum computational universality.

Given the extensive body of prior work showing how to exploit CV quantum systems for quantum computation, why should one try to propose and study new schemes? There are at least two main reasons for doing so: 
\begin{enumerate}[(I)]
\item\label{it:firstnewscheme}
First, the proposed schemes still make use of  several different types  of non-Gaussian operations which  may be experimentally challenging to implement. In particular, 
all the schemes mentioned above make use of photon number measurements. Such measurements are typically significantly more challenging to realize than homodyne or heterodyne detection (i.e., Gaussian measurements). The proposals~\cite{LloydBraunstein99,chabaud2024bcomplexity} additionally use non-linear unitary gates which  are also highly non-trivial to realize in experiments. (We note that the cubic phase gate considered in~\cite{chabaud2024bcomplexity} has also been proposed as a way of obtaining a universal gate set in quantum fault-tolerance based on Gottesman-Kitaev-Preskill (GKP) codes~\cite{gkp}, but its use in that context has also been questioned~\cite{hastrupmikkelmen21}.) New schemes can try to reduce the use of such sources of non-Gaussianity, or at least eliminate the use of different types of non-Gaussian operations such as unitaries and measurements.

\item\label{it:secondnewscheme}
Second, and more importantly, significant resource-theoretic aspects of quantum computation using CV schemes remain largely unexplored and require further study. Most significantly, unlike for qubits, the number of bosonic modes involved in a computation is not the only relevant measure determining scalability. Instead, it is necessary to consider the amount of energy required in a computation. We note that  -- while the importance of this aspect has been recognized in earlier work -- a detailed analysis for the considered schemes has mostly been missing.
Ref.~\cite{chabaud2024bcomplexity}
emphasizes the need to further study  
computational complexity under energy limitations. (The model CVBQP studied therein involves Gaussian operations, the cubic phase gate, and number state measurements, but does not incorporate energy considerations.) To our knowledge there are no prior results on the trade-off between system size (such as the number of modes) and the amount of energy expended when realizing a quantum computation.
\end{enumerate}
\subsection*{Our contribution}
We contribute to both points~\eqref{it:firstnewscheme} and~\eqref{it:secondnewscheme} above:

\begin{description}
\item[A new scheme for quantum computation in hybrid qubit-oscillator setups:] 
We introduce a new efficient scheme for realizing an~$n$-qubit quantum computation using oscillators coupled to a constant number of qubits.

Our scheme is distinguished by the fact that the set of operations supplementing (Gaussian) linear optics is different from those considered earlier, and quite minimal. In addition to bosonic modes equipped with linear optics operations, we use a constant number of auxiliary qubits, qubit operations and -- most importantly -- qubit-oscillator couplings of Jaynes-Cummings type. This simple set of operations is natively available 
in several setups such as superconducting circuits~\cite{Eickbusch_2022, CampagneEikbushetal20}. 
We refer to~\cite{liu2024hybridoscillatorqubitquantumprocessors} for an up-to-date and detailed review of the state of the art, and an extensive discussion of physical realizations of the operations we use here.  We stress that  unlike prior work (relying on photon number measurements), our constructions only use homodyne (position-) measurements on the oscillators. That is, all operations on the oscillators are Gaussian, and the only source of non-Gaussianity is the coupling to the qubits. 

While providing a significant simplification from a practical point of view
in suitable experimental setups, the use of these alternative operations to implement circuits  also means that the construction relies on a quite different encoding of qubit states in oscillators. Instead of using e.g., number states (or certain linear combinations thereof), our scheme is based on so-called comb states. These  can be viewed as code states of higher-dimensional approximate GKP codes~\cite{gkp}. Although our construction does not incorporate fault-tolerance considerations at present, the choice of this kind of encoding should facilitate the design of corresponding error correction procedures.

\item[Energy-versus-system size tradeoff analysis:]
We provide a detailed analysis showing how the number of modes can be traded off against the energy required: We introduce a family of protocols
for weakly simulating an~$n$-qubit computation, with each protocol covering a different range of system parameters~$(m,\energy,\varepsilon)$. Here~$m$~denotes the 
number of modes involved,~$\energy$ is a parameter determining an upper bound on the maximal amount of energy created in the execution of the protocol (as defined below), and~$\varepsilon$~determines the level of accuracy of simulation (in~$L^1$-distance). By covering different choices of system parameters, our construction becomes accessible to a larger range of experimental systems.

In addition to this practical aspect, this rigorous achievability result provides insights into the fundamental trade-off between the system size (quantified by the number~$m$ of modes) and the amount of energy required (quantified by~$\energy$). 
We also establish new lower bounds on the amount of energy required to effectively encode an~$n$-qubit Hilbert space into a number of oscillators. These provide evidence that at least in some limiting cases, our construction is optimal (i.e., requires the minimal amount of energy possible).

We note that these results make first steps in the direction  of formalizing  computational complexity of CV quantum computation under energy constraints, a fundamental question put forward in~\cite{chabaud2024bcomplexity}. Indeed, it is straightforward to define  complexity classes analogous to CVBQP (introduced in~\cite{chabaud2024bcomplexity}) 
which capture the power of a hybrid qubit-oscillator model given a tuple~$(m,\energy)$ of system parameters (scaling with the problem size). Our results 
on simulating~$n$-qubit circuits can then be specialized to the statement that the corresponding complexity classes contain~BQP. 

\end{description}

\subsubsection*{Outline}
In Section~\ref{sec:problemstatement} we introduce the physical setup and the computational problem we consider. In Section~\ref{sec:mainresult} we give our main result. In Section~\ref{sec:proofmain} we give the proof of this result. Finally, we conclude in Section~\ref{sec:discussion}.

\section{Problem statement\label{sec:problemstatement}}
Let us state the problem we consider in detail. We first define the qubit-oscillator model. We then formally introduce the computational problem of sampling from the output distribution of an~$n$-qubit circuit. 

\paragraph{The qubit-oscillator model.}The hybrid qubit-oscillator model  (see e.g.,~\cite{liu2024hybridoscillatorqubitquantumprocessors} for a recent review) describes a setting with~$m$ oscillators and~$r$~qubits, i.e., with Hilbert space~$L^2(\mathbb{R})^{\otimes m}\otimes(\mathbb{C}^2)^{\otimes r}$.  It assumes that in addition to
\begin{enumerate}[(A)]
\item \label{it:elem1}
preparation of the vacuum state~$\ket{\vac}$ on any mode and of the computational basis state~$\ket{0}$ on any qubit, as well as
\item \label{it:elem2}
computational basis measurement of any qubit, and homodyne (position) measurement of any mode,
\end{enumerate}
 the following unitary operations are available:
\begin{enumerate}[(a)]
\item\label{it:elementaryhybridfirst}

Single-qubit unitaries on any qubit or two-qubit unitaries on any pair of qubits. Without loss of generality, we may restrict to single-qubit Cliffords and the $T$-gate, as well as two-qubit controlled-phase gates~$\CZ$.
\item
Single-mode displacements of (any) constant strength on any mode, i.e., any operator of the form~$e^{i(\alpha Q_k+\beta P_k)}$  with~$|\alpha|,|\beta|$ constants (independent of e.g., the problem size considered). Here~$Q_k,P_k$ denote the canonical position- and momentum operators on the~$k$-th mode for~$k\in \{1,\ldots,m\}$.
\item
Single-mode squeezing operators of constant strength applied to any mode. For a constant~$\alpha>0$, this is defined as~$M_\alpha=e^{-i(\log \alpha) (Q_kP_k+P_kQ_k)/2}$  when acting on mode~$k\in \{1,\ldots,m\}$. 
\item\label{it:elementaryhybridlast}
Qubit-controlled single-mode phase space displacements of bounded strength. This takes the form
\begin{align}
\ctrl_je^{i(\alpha Q_k-\beta P_k)}&= I\otimes \proj{0}_j+ e^{i(\alpha Q_k-\beta P_k)} \otimes \proj{1}_j
\end{align}
when controlled on the~$j$-th qubit and acting on the~$k$-th mode, where~$|\alpha|,|\beta|$ are constants,~$j\in \{1,\ldots,r\}$  and~$k\in \{1,\ldots,m\}$. 
\end{enumerate}
We call the set of unitary operations~\eqref{it:elementaryhybridfirst}--\eqref{it:elementaryhybridlast} acting on the space~$L^2(\mathbb{R})^{\otimes m} \otimes (\mathbb{C}^2)^{\otimes r}$ the set of elementary (unitary) operations in the hybrid (qubit-oscillator) model and denote it by~$\Uelem{m}{r}$. 
(We note that this set of unitary operations is more restricted than the one considered e.g., in~\cite{brenner2024complexity}, where additionally (controlled) single-mode phase space rotations are included.)

Here we study the computational power of this model, or, more precisely, of non-adaptive quantum circuits composed of these operations. (Non-adaptivity here refers to the fact that measurement results are not used to (classically) control further  operations.) 
 Specifically, we ask if circuits in the qubit-oscillator model
 with parameters~$(m,r)$ can (approximately) sample from the output distribution of an~$n$-qubit quantum circuit. Clearly, this is trivial if the number~$r$ of available qubits satisfies~$r\geq n$. We will thus focus on the case where~$r$ is constant (in fact,~$r=3$ in our construction).

\paragraph{The output distribution of an~$n$-qubit quantum circuit.}
Let us define the sampling problem considered in more detail. 
Consider an~$n$-qubit system and the set~$\logicalgates$ consisting of all single-qubit~$T$-gates, single-qubit Clifford gates, and two-qubit controlled-phase gates~$\CZ$ on any pair of qubits. Let~$U=U_s\cdots U_1$ be a circuit  consisting of~$s$~gates, where~$U_t\in \logicalgates$ for every~$t\in \{1,\ldots,s\}$. We  write~$s=\size(U)$ for the (circuit) size of~$U$. 
We are interested in the output distribution 
\begin{align}
p(x)&=|\langle x,U0^n\rangle |^2\qquad\textrm{ for }\qquad x\in\{0,1\}^n\ \label{eq:samplingdistributionmain}
\end{align}
of measurement outcomes when applying such a circuit to the initial state~$\ket{0^n}=\ket{0}^{\otimes n}$, and subsequently measuring all qubits in the computational basis.

\paragraph{Sampling in the qubit-oscillator model.}We ask if the distribution~\eqref{eq:samplingdistributionmain} on~$n$~bits can approximately be  sampled from by a circuit~$V$ in a hybrid qubit-oscillator setup with~$m$ oscillators and~$r$~qubits.   Concretely, we consider efficient  procedures of the following form, where efficiency is expressed in terms of the number of qubits~$n$ and the circuit size~$s=\size(U)$ of the circuit~$U$:
\begin{enumerate}[(1)]
\item\label{it:firststepprep}
The state~$\ket{\Phi^{(0)}}=\ket{\vac}^{\otimes m}\otimes \ket{0}^{\otimes r}$ is prepared initially.
\item
An efficient (non-adaptive) 
 qubit-oscillator circuit~$V=V_T\cdots V_1$, where~$V_t\in \Uelem{m}{r}$ for every~$t\in \{1,\ldots,T\}$,
  is applied to~$\ket{\Phi^{(0)}}$.
  Efficiency here means that the size~$T=\size(V)$ of the qubit-oscillator circuit is at most~$T=\poly(n,s)$, i.e., polynomial in the number~$n$ of qubits and the size~$s=\size(U)$ of the circuit~$U$ considered.
  \item
  In the resulting state~$V\ket{\Phi^{(0)}}$,   every bosonic mode is measured using a homodyne (position) measurement, and 
 every qubit is measured in the computational basis. 
  This results in a measurement outcome
\begin{align}
(y,z)&=\left((y_1,\ldots,y_m),(z_1,\ldots,z_r)\right)\in\mathbb{R}^m\times\{0,1\}^r\ .
\end{align}
\item\label{it:laststepsampling}
An (efficiently computable) ``post-processing'' function~$\post:\mathbb{R}^m\times\{0,1\}^r\rightarrow \{0,1\}^n$ is applied to the measurement outcomes, yielding an output~$x=\post(y,z)$. 
\end{enumerate}
 We note that for ``efficient computability'' to make sense, we have to restrict to machine-precision arithmetic in principle. However, it will be clear from our results that these are robust to rounding errors. 
  For brevity, we therefore omit a more detailed discussion of this aspect. We note, however, that it has important physical implications: For example,  the homodyne measurements do not need to be sharp (in the sense of resolving any arbitrarily small length-scale).

We call a pair~$\left(V,\post\right)$ 
defining a procedure specified by   
Steps~\eqref{it:firststepprep}--\eqref{it:laststepsampling}
a sampling scheme on~$m$ oscillators and~$r$~qubits. In more detail, let us fully specify the distribution~$q$ over outputs~$x\in \{0,1\}^n$ produced by such a sampling scheme. The 
state produced by the procedure (before the measurements) can be written as 
\begin{align}
\ket{\Phi^{(t)}}&=V\ket{\Phi^{(0)}}=\sum_{z\in \{0,1\}^r}\sqrt{\lambda(z)} \ket{\Psi^{(t)}_z}\otimes\ket{z}\in L^2(\mathbb{R}^m)\otimes (\mathbb{C}^2)^{\otimes r}
\end{align}
where~$\lambda(z)\geq 0$,~$\sum_{z\in\{0,1\}^r}\lambda(z)=1$ and~$\{|\Psi^{(t)}_z\rangle\}_{z\in \{0,1\}^r}\subset L^2(\mathbb{R}^m)$ are normalized (but not necessarily pairwise orthogonal) states of the~$m$~oscillators. This implies that the  measurement outcome~$z\in \{0,1\}^r$ on the qubits is obtained with probability~$\lambda(z)$.
Conditioned on getting an outcome~$z\in \{0,1\}^r$ when measuring the qubits, the conditional probability density function for obtaining
an outcome~$y\in \mathbb{R}^m$ from the 
homodyne measurements is given by
\begin{align}
f_{Y|Z=z}(y)&= |\Psi^{(t)}_z(y)|^2\qquad\textrm{ for }\qquad y\in\mathbb{R}^m\ .
\end{align}
After post-processing, the output~$x\in\{0,1\}^n$ is therefore produced with probability
\begin{align}
q(x)&=\sum_{z\in \{0,1\}^r} \lambda(z) \int_{\post_z^{-1}(\{x\})} f_{Y|Z=z}(y)dy\qquad\textrm{ where }\qquad \post_z(y)=\post(y,z)\ .\label{eq:qdistribution}
\end{align}
In summary, our procedure produces a sample~$x\in \{0,1\}^n$ from the distribution~$q$.

Let us write~$\sampscheme_{m,r}$ for the set of all distributions of the form~\eqref{eq:qdistribution} produced by   sampling schemes on~$m$~oscillators and~$r$~qubits. Our 
goal is to approximate the distribution~$p$ defined by Eq.~\eqref{eq:samplingdistributionmain} by a sampling scheme in~$L^1$-norm. Given a fault-tolerance/approximation parameter~$\varepsilon\in (0,1)$, our question therefore is the following:  
\begin{quote}
{\bf Question 1}: Is there is a distribution~$q\in \sampscheme_{m,r}$ such that~$\|q-p\|_1 \leq \varepsilon$? That is, can the distribution~$p$ be sampled from with error~$\varepsilon$ 
in the qubit-oscillator model using~$m$ oscillators and~$r$~qubits?
\end{quote}

In fact, we are interested in a more refined question: We would  like to know if there is a sampling scheme with limited energy.  We define the (amount of) energy  of a qubit-oscillator state~$\rho\in \cB(L^2(\mathbb{R})^{\otimes m}\otimes (\mathbb{C}^2)^{\otimes r})$ as 
\begin{align}
\energy\left(\rho
\right)&=\max_{\alpha\in \{1,\ldots,m\}}
\tr\left((Q_\alpha^2+P_\alpha^2)\rho\right)\ , \label{eq:defsqueezemain}
\end{align}
where $Q_\alpha$ and $P_\alpha$ are the canonical position and momentum operators on the mode~$B_\alpha$. 
In other words, $\energy(\rho)$ is the maximum amount of energy contained in any single mode of~$\rho$. 
 Finally, we define the (maximal) energy of a qubit-oscillator circuit~$V=V_T\cdots V_1$ on input~$\ket{\Phi^{(0)}}$ as 
\begin{align}
\energy(V)&=\max_{0\leq t\leq T}  \energy\left(
V_t\cdots V_1\ket{\Phi^{(0)}}\right)\, ,
\end{align}
In other words, the energy of the circuit~$V$ is the maximal energy of any state encountered in the execution of~$V$. For a given amount~$\energy\in (0,\infty)$ of energy, let us write~$\sampscheme_{m,r}(\energy)$ for the set of distributions obtained by sampling schemes~$(V,\post)$ whose energy is limited by~$\energy(V)\leq \energy$. 
The refined question we address then is the following:
\begin{quote}
{\bf Question 2}: Is there is a distribution~$q\in \sampscheme_{m,r}(\energy)$ such that~$\|q-p\|_1 \leq \varepsilon$? In other words, is it possible to use~$m$~oscillators and~$r$~qubits to sample from~$p$ with error~$\varepsilon$  without  generating more energy than specified by~$\energy$?
\end{quote}

\section{Main result: Energy versus space tradeoff\label{sec:mainresult}}
Our main result is the following:
\begin{theorem}\label{thm:main}
There are constants~$C,\alpha,\beta,\gamma>0$ such that the following holds.
Let~$p$ be the output distribution 
of an~$n$-qubit circuit
$U=U_s\cdots U_1$  of size~$s$, see Eq.~\eqref{eq:samplingdistributionmain}. 
Let~$m\in\mathbb{N}$ be a certain number of modes and~$\energy>0$ an upper bound on the amount of available energy. Then there is a distribution~$q\in \sampscheme_{m+1,3}(\energy)$
such that
\begin{align}
\|q-p\|_1 \leq C\cdot (s+m)^{\alpha} \cdot 2^{\beta n/m} \cdot \energy^{-\gamma}=:\varepsilon\ . \label{eq:epsimnmrelationship}
\end{align}
In other words, the distribution~$p$ can be sampled from with an error~$\varepsilon$ in~$L^1$-distance using~$m+1$~oscillators,~$3$ qubits and energy bounded by~$\energy$.
\end{theorem}

We obtain explicit constants in the proof of Theorem~\ref{thm:main} (see Eq.~\eqref{eq:mainbound} in Section~\ref{sec:proofmain}).  
As an immediate consequence of Theorem~\ref{thm:main} we obtain the following upper bound on the amount of energy 
required, given a number~$m=m(n)$ of available modes  and a desired error~$\varepsilon=\varepsilon(n)$ (both possibly given as a function of the number~$n$~of qubits in a circuit family). 
\begin{corollary} \label{cor:squeezingtradeoff}
There are constants~$C, \delta, \mu>0$ such that the following holds. 
  Consider the problem of sampling from the output distribution(s)~$p_n$ of a (family of)~$s(n)$-size circuit(s) on~$n$~qubits. Let~$ m(n)$ be the number of modes used.
  Let~$\varepsilon(n)$ 
 be a desired $L^1$-distance error. Here~$s,m:\mathbb{N}\rightarrow\mathbb{N}$ 
  and~$\varepsilon:\mathbb{N}\rightarrow (0,2]$ are functions of the number of qubits~$n$ considered. Define the function
\begin{align}
  \energy(n):=C\cdot (2^{n/m(n)})^\delta \cdot (s(n)/\varepsilon(n))^\mu\ .
\end{align} 
Then the following holds. There is a distribution
$q_n\in \sampscheme_{m(n)+1,3}(\energy(n))$
such that~$\|q_n-p_n\|_1 \leq \varepsilon(n)$.

In particular, for a circuit of polynomial size~$s(n)=O(\mathsf{poly}(n))$, 
an inverse polynomial sampling error~$\varepsilon(n)=O(1/\mathsf{poly}(n))$ can be achieved using~$3$~qubits and
\begin{enumerate}[(i)]
\item
a linear number~$m=\Theta(n)$ of modes with polynomial energy~$\energy(n)=O(\mathsf{poly}(n))$, or
\item
a sublinear, polynomial number~$m=\Theta(n^\alpha)$,~$\alpha \in (0,1)$ of modes with subexponential energy~$\energy(n)=2^{O(n^{1-\alpha})}$, or 
\item
a constant number~$m=\Theta(1)$ of modes and exponential energy~$\energy(n)=2^{O(n)}$. 
\end{enumerate}
\end{corollary}

\section{Proof of the main result} \label{sec:proofmain}
In this section we prove the main result stated in Theorem~\ref{thm:main}. More precisely, we give a concrete scheme that follows the general outline described in Steps~\eqref{it:firststepprep}--\eqref{it:laststepsampling}. It leverages approximate GKP states as the fundamental physical information carriers together with three physical qubits to realize computation.

\begin{figure}[t]
\centering
\includegraphics{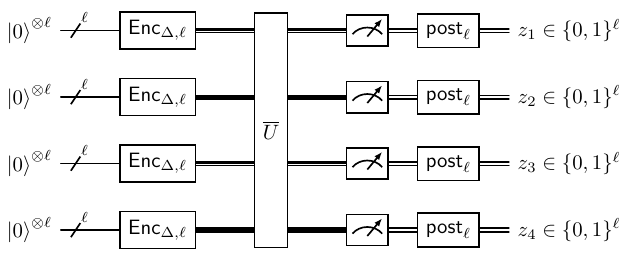}
\caption{Realization of a (logical) circuit given by a unitary~$U$ on~$n=\ell m$ qubits  
 using~$m$~bosonic modes. The illustration is for~$m=4$. 
Blocks of~$\ell$ (logical) qubits are encoded into a GKP-code~$\gkpcoderect{\Delta}{\star}{2^\ell}\subset L^2(\mathbb{R})$ encoding a~$2^\ell$-dimensional qudit. We use a gadget~$G$ realizing each (logical) single-qubit gate~$\overline{G}$, as well as a gadget~$\CZ$ implementing a (logical) two-qubit~$\overline{\CZ}$-gate. Substituting each logical gate by such a gaget, we obtain a physical implementation~$U$ of the (logical) unitary~$\overline{U}$.
Each measurement is a (homodyne) position-measurement. The measurement result is post-processed using a function~$\post_\ell$ in order to emulate a measurement in the computational (multiqubit) basis.
\label{fig:implementationbasicidea}}
\end{figure}

\paragraph{\bf Encoding subsets of qubits into oscillators. } Fig.~\ref{fig:implementationbasicidea} shows the basic idea of our construction.  It makes essential use of a one-parameter family~$\{\gkpcoderect{\Delta}{\star}{d}\}_{\Delta>0}$ of  approximate Gottesman-Kitaev-Preskill (GKP) codes which encode a qudit of dimension~$d\geq 2$ into a subspace of~$L^2(\mathbb{R})$.  For any integer~$d\geq 2$ and real parameter~$\Delta>0$, the code~$\gkpcoderect{\Delta}{\star}{d}$ is spanned by an orthonormal basis~$\{\ket{\Sha_{\Delta}^\star(j)_d}\}_{j=0}^{d-1}$ of ``comb-like'' GKP states (see Fig.~\ref{fig:comblikestate} and Appendix~\ref{sec:appendixshastates} for detailed definitions). Throughout, we choose the code space dimension~$d=2^\ell$ as a power of~$2$.

We are interested in encoding~$n$ (logical) qubits~$\logicalQ_1\cdots \logicalQ_n$ into~$m\leq n$ oscillators. To this end, we set 
\begin{align}
K&:= m-(n\pmod m)\ .
\end{align}
Then~$K\in \{0,\ldots,m-1\}$
and~$n':=n+K$ is divisible by~$m$. We introduce~$K$~``dummy'' (logical) qubits~$\overline{D}_1\cdots \overline{D}_K$ in the state~$\ket{0}$, and extend the given $n$-qubit (logical) circuit~$U$ to act trivially on these qubits. 
The resulting circuit~$U'=U_{\logicalQ_1\cdots \logicalQ_n}\otimes I_{\overline{D}_1\cdots \overline{D}_K}$ acts on~$n'$ (logical) qubits. 

Let us organize the~$n'$ qubits~$\logicalQ_1\cdots \logicalQ_{n'}:=\logicalQ_1\cdots \logicalQ_n\overline{D}_1\cdots \overline{D}_K$ into 
$m$ blocks~$S_1\cdots S_m$ of size
\begin{align}
\ell&=n'/m\ 
\end{align}
each. That is, we group the~$n'=m\ell$ qubits as
\begin{align}
S_1&=\logicalQ_1\cdots \logicalQ_{\ell} \ ,\\
S_2&=\logicalQ_{\ell+1}\cdots \logicalQ_{2\ell}\ ,\\
\vdots\\
S_m&=\logicalQ_{(m-1)\ell+1}\cdots \logicalQ_{m\ell} \ .
\end{align}
In total, we obtain~$m$ blocks~$S_1,\ldots,S_m$ where each~$S_j\cong (\mathbb{C}^2)^{\otimes \ell}$ consists of~$\ell$~qubits, for~$j\in \{1,\ldots,m\}$. 

We now encode each block~$S_j\cong (\mathbb{C}^2)^{\otimes \ell}$,~$j\in \{1,\ldots,m\}$ of qubits into a~$2^\ell$-dimensional subspace of a single oscillator, namely the approximate 
 GKP code~$\gkpcoderect{\Delta}{\star}{2^\ell}\subset L^2(\mathbb{R})$.
 The squeezing parameter~$\Delta>0$ will be chosen below.  
  In more detail, we identify the sets~$\{0,1\}^\ell$ and~$\{0,\ldots,2^\ell-1\}$ using the bijection (i.e., binary representation)
\begin{align}
\begin{matrix}
\iota_\ell : & \{0,1\}^\ell & \rightarrow & \{0,\ldots,2^\ell-1\}\\
&(x_{\ell-1},\ldots,x_0)&\mapsto &\base{2}{x_{\ell-1},\ldots,x_0}:=\sum_{j=0}^{\ell-1}x_j 2^j
\end{matrix}\  .\label{eq:iotaelldefinition}
\end{align}
We then use the isometric encoding map
\begin{align}
\begin{matrix}
\encmapgkp_{\Delta,\ell}:& (\mathbb{C}^2)^{\otimes \ell } & \rightarrow & L^2(\mathbb{R})\\
       &\ket{x_{\ell-1},\ldots,x_0} & \mapsto & \ket{\Sha_{\Delta}^\star(\base{2}{x_{\ell-1},\ldots,x_0})_{2^\ell}}\ 
\end{matrix}
\end{align}
to encode~$\ell$~qubits into the approximate GKP-code~$\gkpcoderect{\Delta}{\star}{2^\ell}$.  
This encoding map is used for each of the~$m$ blocks, i.e., 
we use the map
\begin{align}
\begin{matrix}
\encmapgkp_{\Delta,\ell}^{\otimes m}: &(\mathbb{C}^2)^{\otimes n'}  &\rightarrow  &L^2(\mathbb{R})^{\otimes m}\\
& \ket{x} & \mapsto & (\encmapgkp_{\Delta,\ell}|x^{(1)}\rangle)\otimes\cdots\otimes (\encmapgkp_{\Delta,\ell}|x^{(m)}\rangle)
\end{matrix} \label{eq:multi-modeencoding}
\end{align}
to encode~$n'=m\ell$-qubit states into the~$m$~oscillators. Here we
identified~$\left((\mathbb{C}^2)^{\otimes \ell }\right)^{\otimes m}\cong (\mathbb{C}^2)^{\otimes n'}$
by linearly extending the bijection 
\begin{align}
\begin{matrix}
\{0,1\}^{n'}&\rightarrow & (\{0,1\}^{\ell})^m\\
x&\mapsto &(x^{(1)},\ldots,x^{(m)})
\end{matrix}
\end{align}
where~$x=(x_{2^{n'-1}},\ldots,x_0)$ and~$x^{(\alpha)}_j=x_{(\alpha-1)\ell+j}$ for~$(\alpha,j)\in \{1,\ldots,m\}\times \{0,\ldots,\ell-1\}$. 

\begin{figure}[H]
  \centering
  \begin{subfigure}{\textwidth}
    \includegraphics[height = 5.6cm]{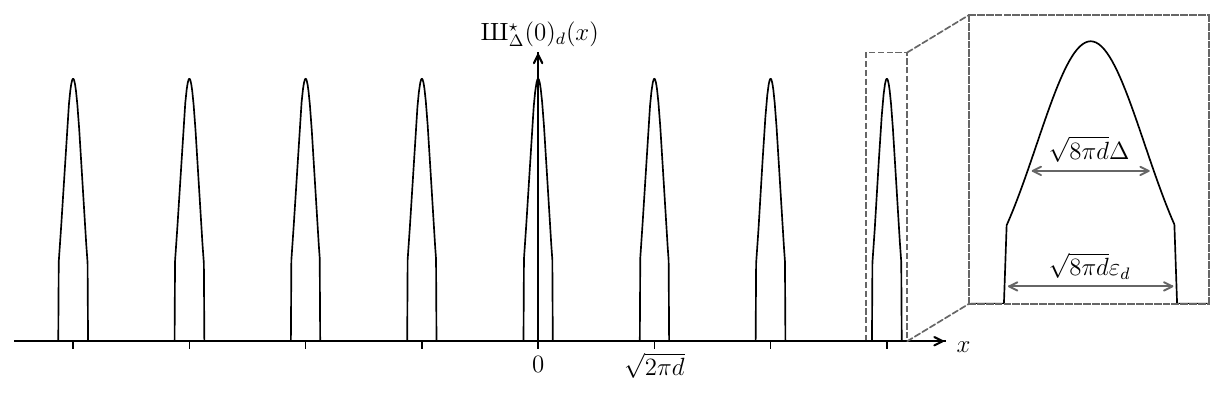}
    \caption{The state~$\Sha_{\Delta}^\star(0)_d\in L^2(\mathbb{R})$.\label{fig:insetvd}}
  \end{subfigure}\\
  \vspace{5ex}
  \begin{subfigure}{\textwidth}
    \includegraphics[height = 5.6cm]{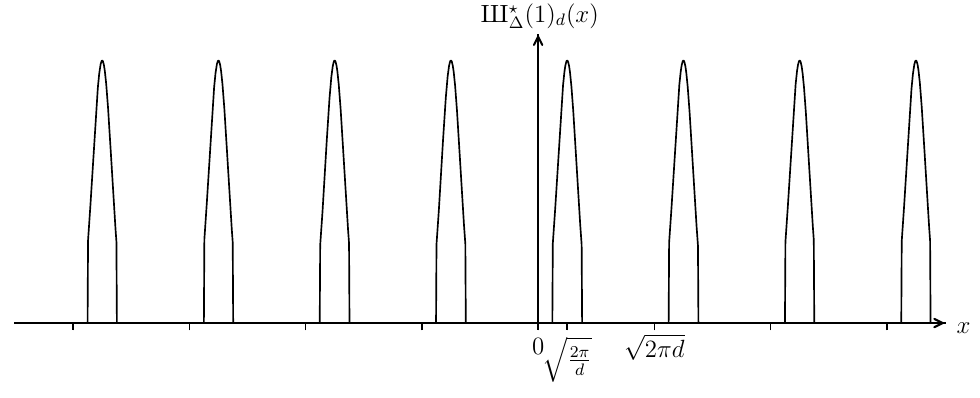}
    \caption{The state~$\Sha_{\Delta}^\star(1)_d\in L^2(\mathbb{R})$. It is obtained from~$\Sha_{\Delta}^\star(0)_d$ by a shift of~$\sqrt{2\pi/d}$.}
  \end{subfigure}   
 
\caption{
The ``comb'' states~$\Sha_{\Delta}^\star(j)_d$ for~$d=4$,~$j =0$ and~$j=1$ (see Appendix~\ref{sec:appendixshastates} for their definition). 
The parameter~$\Delta$ determines the width of Gaussians defining the individual peaks, while the width of their support (obtained by truncation, i.e., restriction) is proportional to~$\varepsilon_d :=1/(2d)$. The truncation of the individual peaks (best visible in the inset in Fig.~\ref{fig:insetvd})
 ensures that the family~$\{|\Sha_\Delta^\star(j)_d\rangle\}_{j=0}^{d-1}$ is orthogonal
 and defines a $d$-dimensional code space~$\gkpcoderect{\Delta}{\star}{d}\subset L^2(\mathbb{R})$. 
The depicted states each have~$L_{\Delta,d} = 8$ peaks. 
Unlabeled tick marks are shown at integer multiples of~$\sqrt{2\pi d}$.}
\label{fig:comblikestate}
\end{figure}

We divide the presentation of our scheme into three steps: initial state preparation, logical unitaries, and logical measurement. We then present an error analysis for our scheme and bound the amount of energy generated. This results in the bound~\eqref{eq:epsimnmrelationship}.

\paragraph{\bf Initial state preparation.}
In addition to the~$m$ bosonic modes~$B_1\cdots B_m\cong L^2(\mathbb{R})^{\otimes m}$ encoding the~$n'=m\ell$ qubits, our scheme uses an auxiliary mode we denote by~$B_{\mathsf{aux}}\cong L^2(\mathbb{R})$, and three auxiliary qubits denoted~$Q_1Q_2Q_3\cong (\mathbb{C}^2)^{\otimes 3}$. It starts by preparing an initial state which is approximately of the form
\begin{align}
\ket{\Phi^{\mathsf{ideal}}_{\mathsf{init}}(m,\ell,\Delta)}:=\ket{\Sha_{\Delta}^\star(0)_{2^\ell}}^{\otimes m}_{B_1\cdots B_m}\otimes 
\ket{\Sha_{\Delta,\ell}^{\mathsf{aux}}(0)_2}_{B_{\mathsf{aux}}}\otimes \ket{0}^{\otimes 3}_{Q_1Q_2Q_3}\ .\label{eq:idealinitialstatedef}
\end{align}
We note that  the state
$\ket{\Sha_{\Delta}^\star(0)_{2^\ell}}^{\otimes m}$ on the modes~$B_1\cdots B_m$ 
is an encoding 
of the logical state~$\ket{0^{n'}}:=\ket{0}^{\otimes n'}$ of the~$n'$~qubits, and can be prepared by creating~$m$~copies of the state
\begin{align}
\encmapgkp_{\Delta,\ell}\ket{0^{\ell}}&=\ket{\Sha_{\Delta}^\star(0)_{2^\ell}} \ .
\label{eq:msqeezing}
\end{align}
The state~$\ket{\Sha_{\Delta,\ell}^{\mathsf{aux}}(0)_2}$
on the auxiliary mode~$B_{\mathsf{aux}}$ is a code state of a certain  approximate GKP code encoding a logical qubit with parameters depending on~$\Delta$ and~$\ell$ (see Appendix~\ref{sec:rectGKP} for a rigorous definition).

Importantly, (approximations to) the  states~$\ket{\Sha_{\Delta}^\star(0)_{2^\ell}}$ and~$\ket{\Sha_{\Delta,\ell}^{\mathsf{aux}}(0)_2}$ can be created by  efficient protocols in the qubit-oscillator model: For both states, there  is a preparation circuit  (denoted $U^{\mathsf{prep}}$ and $V^{\mathsf{prep}}$, respectively) using only one oscillator and one qubit (i.e., with~$m=r=1$) and a logarithmic number of elementary operations which achieves a polynomial error in~$\Delta$ (as measured by the trace distance). This was shown in~\cite{brenner2024complexity}
 (see Theorem~\ref{thm:preparationproceduregmvd} in the appendix for details).
 Here we give a derived construction, see Fig.~\ref{fig:unitarycircuitWprepapp}, with parameter choices adapted for our purposes. It generates an approximation to the state $\ket{\Phi^{\mathsf{ideal}}_{\mathsf{init}}(m,\ell,\Delta)}$ (see Eq.~\eqref{eq:idealinitialstatedef}). We additionally establish an upper bound on the amount of energy generated in this protocol, see Lemma~\ref{lem:squeezingWprep} in Appendix~\ref{sec: momlimitsimplementation}.

\begin{theorem}(Initial state preparation)\label{thm:initialstateprep}
Let~$\ell\in\mathbb{N}$ and~$\Delta\in (0,1/4)$ be such that~$\Delta \leq 2^{-(\ell+1)}$.  Let~$m\in\mathbb{N}$. Consider the state $\ket{\Phi^{\mathsf{ideal}}_{\mathsf{init}}}:=\ket{\Phi^{\mathsf{ideal}}_{\mathsf{init}}(m,\ell,\Delta)}\in L^2(\mathbb{R})^{\otimes (m+1)}\otimes (\mathbb{C}^2)^{\otimes 3}$ defined by Eq.~\eqref{eq:idealinitialstatedef}. There is a circuit
\begin{align}
W^{\mathsf{prep}}=W_{\size(W^{\mathsf{prep}})}\cdots W_{1}
\end{align}
on~$L^2(\mathbb{R})^{\otimes (m+1)}\otimes \mathbb{C}^2$ 
composed of
\begin{align}
\size(W^{\mathsf{prep}})\leq 42m \log 1/\Delta\label{eq:sizeupperboundUprep}
\end{align}
elementary operations belonging to~$\Uelem{m}{1}$ such that the
output state
\begin{align}
\ket{\Phi_{\mathsf{init}}}_{B_1\cdots B_mB_{\mathsf{aux}}Q_1Q_2Q_3}:=\left(W^{\mathsf{prep}}_{B_1\cdots B_mB_{\mathsf{aux}}Q_1}\otimes I_{Q_2Q_3}\right)\left(
\ket{\vac}^{\otimes m+1}_{B_1\cdots B_mB_{\mathsf{aux}}}\otimes \ket{0}^{\otimes 3}_{Q_1Q_2Q_3}\right)
\end{align}
when applying~$W^{\mathsf{prep}}$ to  $m+1$~bosonic modes prepared in the vacuum state and three qubits in the state~$\ket{0}$ satisfies
\begin{align}
\varepsilon_{prep}&:=\left\|
\proj{\Phi_{\mathsf{init}}}-
\proj{\Phi^{\mathsf{ideal}}_{\mathsf{init}}}
\right\|_1\leq 50m \left(\sqrt{\Delta}+2^{2\ell}\Delta^2\right)\ .\label{eq:claimdistancemfold}
\end{align}
\end{theorem}

\begin{figure}[H]
  \centering
  \includegraphics[width = 1\textwidth]{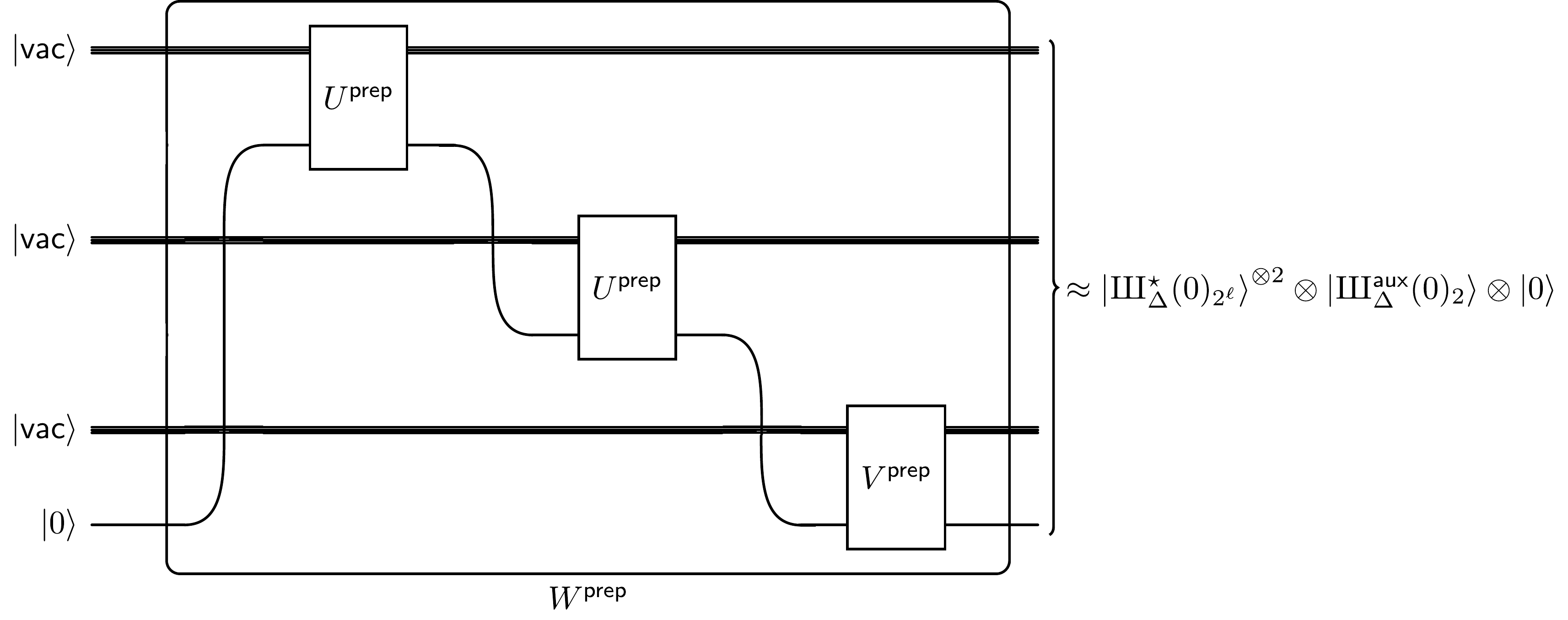}
   \caption{
   The unitary circuit~$W^{\mathsf{prep}}$ preparing an approximation of the state~$\ket{\Sha_\Delta^\star(0)_{2^\ell}}^{\otimes m} \otimes \ket{\Sha_{\Delta,\ell}^{\mathsf{aux}}(0)_2}\otimes\ket{0}$ (for~$m=2$) starting from~$m+1$ copies of the vacuum state~$\ket{\vac}$ and a single qubit initialized in the state~$\ket{0}$.   Each unitary~$U^\mathsf{prep}$ approximately stabilizes the qubit state~$\ket{0}$ while approximately generating the state~$\ket{\Sha_\Delta^\star(0)_{2^\ell}}$, i.e.,  it acts as~$U^\mathsf{prep}(\ket{\vac} \otimes \ket{0}) \approx \ket{\Sha_\Delta^\star(0)_{2^\ell}} \otimes \ket{0}$. Similarly, the unitary~$V^{\mathsf{prep}}$  acts as~$V^{\mathsf{prep}}(\ket{\vac}\otimes \ket{0}) \approx \ket{\Sha_{\Delta,\ell}^{\mathsf{aux}}(0)_2} \otimes \ket{0}$.}\label{fig:unitarycircuitWprepapp}
\end{figure}

\paragraph{\bf Logical unitaries. \label{par:logicalgates}}
We need to argue that we can perform (encoded) computations using this encoding, i.e., in the code space~$\gkpcoderect{\Delta}{\star}{2^\ell}^{\otimes m}$ encoding our~$n'=m\ell$ logical qubits. We have previously given a corresponding construction and an analysis of the associated error in~\cite{cliffordshybrid2025}. Here we only give a high-level sketch and state the relevant parameters, see Theorem~\ref{thm:implementationlogicalqubit} below.

The recompilation procedure of Ref.~\cite{cliffordshybrid2025}  takes as input an $n'=m\ell$-qubit (logical) circuit~$U=U_s\cdots U_1$ consisting of $s$ two-qubit gates, where $U_t$ acts on any pair of qubits for each $t\in \{1,\ldots,s\}$.
It produces a unitary circuit~$W_U=W_T\cdots W_1$ consisting of  $T=O(s\ell^2)$  elementary operations~$W_1,\ldots,W_T\in \Uelem^{m+1,3}$ in the hybrid qubit-oscillator model with $m+1$~modes and $3$~qubits. 
(The construction of $W_U$ proceeds in a gate-by-gate-fashion: Each logical two-qubit unitary~$U_t$ for $t\in \{1,\ldots,s\}$ is implemented by 
a circuit as illustrated in Fig.~\ref{fig:W_Uimplement}.) 

 \begin{figure}[h]
  \centering
  \includegraphics[width=\textwidth]{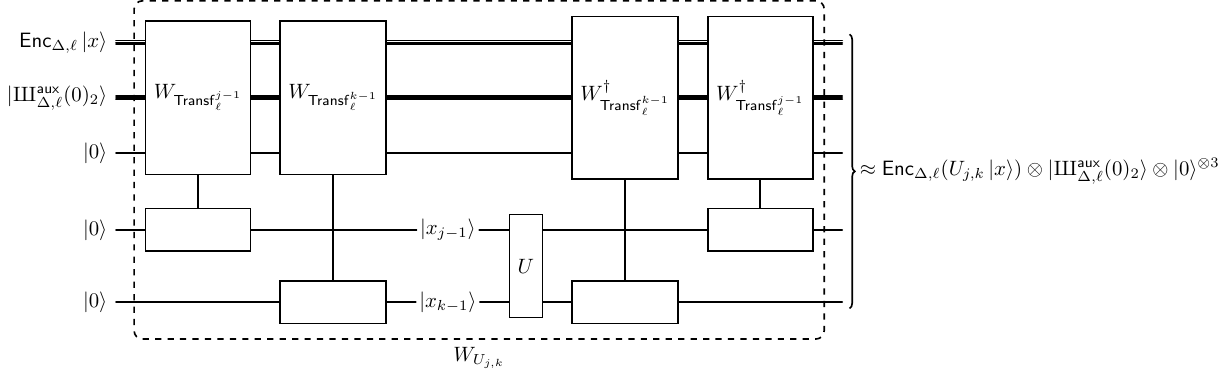}
  \caption{
  Illustration of construction
from  Ref.~\cite{cliffordshybrid2025}: An    implementation~$W_{U_{j,k}}$ of 
a logical two-qubit unitary~$U=U_{j,k}$ acting on the~$j$-th and~$k$-th qubits for~$j,k \in \{1,\dots, m \ell\}$. The illustration is for the special case  $m=1$. 
When acting on an encoded computational basis state~$\ket{x}\in (\mathbb{C}^2)^{\otimes m\ell}$, $x=(x_{m\ell-1},\ldots,x_0)$ 
 in the code space~$\cL_\Delta(m,\ell)$ (see Eq.~\eqref{eq:eqldeltamell}), two bit-transfer unitaries
$W_{\bittransfer{j-1}{\ell}},W_{\bittransfer{k-1}{\ell}}$ are applied to transfer the $j$-th and the $k$-th bits $x_{j-1},x_{k-1}$ onto the two auxiliary qubits. 
 (Here the bit-transfer unitary~$W_{\bittransfer{r}{\ell}}$ approximately acts as~$W_{\bittransfer{r}{\ell}}(\ket{\Sha_{\Delta}^\star([x_{\ell-1}, \dots, x_0])_{2^\ell}} \otimes \ket{\Sha_{\Delta,\ell}^{\mathsf{aux}}(0)_2} \otimes \ket{0}^{\otimes 2}) \approx \ket{\Sha_{\Delta}^{\star}([x_{\ell-1}, \dots,x_{r+1}, 0, x_{r-1},\dots, x_0])_{2^\ell}} \otimes \ket{\Sha_{\Delta,\ell}^{\mathsf{aux}}(0)_2} \otimes \ket{0} \otimes \ket{x_r}$ for~$r \in \{0,\dots,\ell-1\}$.)
  Subsequently, the two-qubit unitary~$U$ can be applied at the physical level to the qubits. Finally,
  the bits are transferred back to the $m$ bosonic modes (the first bold line)
  by acting with the adjoint bit-transfer unitaries.  The state of the second mode (second bold line) acts as a catalyst in this process. Importantly, the bit-transfer unitaries can be realized by $O(\ell^2)$ elementary unitary operations belonging to~$\Uelem^{2,2}$ (respectively $\Uelem^{m+1,2}$), with an error scaling as $O(2^{2\ell}\Delta)$, see Theorem~\ref{thm:implementationlogicalqubit}. We refer to~\cite[Corollary D.6]{cliffordshybrid2025} for a detailed analysis. \label{fig:W_Uimplement}}
  \end{figure}

The unitary $W_U$ constitutes an approximate physical implementation of $U$ when the $\ell\cdot m$~qubits 
are encoded in the code space
\begin{align}
\cL_\Delta(m,\ell)&:=(\gkpcoderect{\Delta}{\star}{2^\ell})^{\otimes m}\otimes \mathbb{C}(\ket{\ket{\Sha_{\Delta,\ell}^{\mathsf{aux}}(0)_2}}\otimes \ket{0}^{\otimes 3})\ \label{eq:eqldeltamell}
\end{align}
(We note that $\gkpcoderect{\Delta}{\star}{2^\ell})^{\otimes m}\cong (\mathbb{C}^2)^{\otimes \ell m}$, hence the first~$m$ bosonic modes contain the logical qubits in this encoding.)
That is, we have 
\begin{align}
W_{U} \left((\encmapgkp_{\Delta,\ell}^{\otimes m}\ket{\Psi})\otimes \ket{\Sha_{\Delta,\ell}^{\mathsf{aux}}(0)_2} \otimes\ket{0}^{\otimes 3}\right)\approx (\encmapgkp_{\Delta,\ell}^{\otimes m}(U\ket{\Psi}))\otimes \ket{\Sha_{\Delta,\ell}^{\mathsf{aux}}(0)_2}\otimes\ket{0}^{\otimes 3}\label{eq:approximateequalitygate}
\end{align}
for any $\ket{\Psi}\in (\mathbb{C}^2)^{\otimes m\ell}$.   The error in the approximation~\eqref{eq:approximateequalitygate} can be bounded as follows:
For a code subspace~$\cL\subset\cH$ of a physical Hilbert space~$\cH$ and a unitary implementation~$W_U:\cH\rightarrow\cH$ of an ideal logical unitary~$U^{\mathsf{ideal}}:\cL\rightarrow\cL$, we call
\begin{align}
\gateerror{\cL}(W_{U},U^{\mathsf{ideal}}):=\left\|(\cW_U-\cU^{\mathsf{ideal}})\circ \Pi_\cL\right\|_\diamond
\end{align}
the logical gate error of the implementation~$W_U$, see~\cite{cliffordslinearoptics2025} for a detailed discussion of this quantity. Here~$\Pi_\cL(\rho)=\pi_\cL\rho \pi_\cL^\dagger$ is defined in terms of the orthogonal projection~$\pi_{\cL}:\cH\rightarrow\cL$ onto~$\cL$, whereas
\begin{align}
\begin{matrix}
\cW_U&:\cB(\cH) & \rightarrow &\cB(\cH)\\
&\rho & \mapsto & \cW_U(\rho):=W_U\rho W_U^\dagger
\end{matrix}
\end{align}
and
\begin{align}
\begin{matrix}
\cU^{\mathsf{ideal}}&:\cB(\cL) & \rightarrow &\cB(\cL)\\
&\rho & \mapsto & \cU^{\mathsf{ideal}}(\rho):=U^{\mathsf{ideal}}\rho (U^{\mathsf{ideal}})^\dagger
\end{matrix}
\end{align}
are the completely positive trace-preserving (CPTP) maps corresponding to the implementation~$W_U$ and an ideal implementation of the gate~$U^{\mathsf{ideal}}$, respectively. 
It was shown in~\cite[Corollary D.6]{cliffordshybrid2025} that the constructed implementation~$W_{U_{j,k}}$  of any logical two-qubit unitary~$U_{j,k}$  has gate error  bounded as 
\begin{align}
\gateerror{\cL_\Delta(1,\ell)}(W_{U_{j,k}}, U^{\mathsf{ideal}}_{j,k})&\leq 600\cdot 2^{2\ell }\Delta\ .\label{eq:gaterrorboundcomposablex}
\end{align}
The following was shown in~\cite[Theorem E.3]{cliffordshybrid2025}, see also the remark thereafter.
 \begin{theorem}(Implementation of logical  qubit circuits~\cite[Theorem E.3]{cliffordshybrid2025})\label{thm:implementationlogicalqubit}
 Consider~$m\ell$ qubits encoded in the space~$\cL_\Delta(m,\ell)$ (cf. Eq.~\eqref{eq:eqldeltamell}), i.e., into
~$m$ copies of the code~$\gkpcoderect{\Delta}{\star}{2^\ell}$ using an auxiliary mode in the state~$\ket{\Sha_{\Delta,\ell}^{\mathsf{aux}}(0)_2}$ and three auxiliary qubits in the state~$\ket{0}$.
 Let~$U=U_s\cdots U_1$ be a unitary circuit on~$n'=m\ell$~qubits of size~$s$, i.e., composed of~$s$ one- and two-qubit gates~$U_1,\ldots,U_s$. Then there is a unitary circuit~$W_U=W_{T}\cdots W_1$
 on~$L^2(\mathbb{R})^{\otimes (m+1)}\otimes (\mathbb{C}^2)^{\otimes 3}$ composed of  
 \begin{align}
 T&\le 340 s\ell^2
 \end{align}
 elementary operations~$W_1,\ldots,W_T\in \Uelem{m+1}{3}$ such that  the logical gate error of the implementation satisfies
 \begin{align}
 \gateerror{\cL_{\Delta}(m,\ell)}(W_{U},U)\leq 600s\cdot 2^{2\ell} \Delta\ .
 \end{align}
 \end{theorem}

\paragraph{\bf Logical measurement.}
We need to argue that an encoded measurement can approximately be realized  by suitably  post-processing the result of a homodyne measurement on the oscillators. Here we introduce the corresponding post-processing procedure. It is derived from the fact that logical information encoded in the code~$\gkpcoderect{\Delta}{\star}{2^\ell}$ can be read out by homodyne (position) measurement and suitable post-processing of the measurement result. This is expressed by the following lemma.

\begin{lemma}[Logical measurement for the code~$\gkpcoderect{\Delta}{\star}{2^\ell}$]\label{lem:logicalmeasurementgkpcoderectdeltastar}
There is an efficiently computable function~$\post_\ell :\mathbb{R}\rightarrow\{0,1\}^\ell$
such that 
\begin{align}
\mathsf{supp}(\encmapgkp_{\Delta,\ell}\ket{x})\subset 
\post_\ell^{-1}(\{x\})\qquad \textrm{ for all }\qquad x\in \{0,1\}^\ell\ .\label{eq:encmappostcontainmone}
\end{align}
as well as
\begin{align}
\post_\ell^{-1}(\{x\})\cap \mathsf{supp}(\encmapgkp_{\Delta,\ell}\ket{x'}))&=\emptyset\qquad  \textrm{ for  any pair }\qquad x\neq x'\in \{0,1\}^\ell\  .\label{eq:encmappostcontainmtwo}
\end{align}
In particular, if~$z\in \mathbb{R}$ is the measurement result when applying a homodyne position-measurement to an encoding~$\ket{\overline{\Psi}}\in \gkpcoderect{\Delta}{\star}{2^\ell}$
of a state~$\ket{\Psi}\in (\mathbb{C}^2)^{\otimes \ell}$, 
then the post-processed output~$x:=\post_\ell(z)$
is distributed according to the distribution~$p(x)=|\langle x,\Psi\rangle|^2$,~$x\in \{0,1\}^\ell$ of measurement outcomes when applying a computational basis measurement to~$\ket{\Psi}$.
\end{lemma}
\begin{proof}
Clearly, it suffices to establish Eqs.~\eqref{eq:encmappostcontainmone}
and~\eqref{eq:encmappostcontainmtwo} for a suitably chosen function~$\post_\ell$.

The following immediately follows from the definition (see Section~\ref{sec:rectGKP} in the appendix) of the state~$\ket{\Sha^\star_\Delta(j)_{2^\ell}}$. For~$j\in \{0,\ldots,2^{\ell}-1\}$,  the function~$\Sha^\star_\Delta(j)_{2^\ell}$  has individual peaks (local maxima)
located at points belonging to the set
\begin{align}
\cS(j) &= \left\{ \sqrt{2\pi \cdot 2^{-\ell}} j+\sqrt{2\pi\cdot  2^{\ell}}z \mid z \in\{-L_{\Delta,2^{\ell}}/2, \ldots, L_{\Delta,2^{\ell}}/2-1\} \right\}\ ,\label{eq:sjsupport}
\end{align}
where~$L_{\Delta,2^{\ell}}= 2^{2(\lceil \log_2 1/\Delta \rceil - \ell)}$ (cf.  Fig.~\ref{fig:comblikestate}).
Furthermore, the support of the function is 
\begin{align}
\mathsf{supp} (\Sha^\star_\Delta(j)_{2^\ell})&=\cS(j)+[-\sqrt{\pi \cdot 2^{-(\ell+1)}},\sqrt{\pi \cdot 2^{-(\ell+1)}}]\label{eq:shafctsup}
\end{align}
where we write~$A+B:=\{a+b\ |\ a\in A, b\in B\}$ for the Minkowski sum of two subsets~$A,B\subset \mathbb{R}$. 
In particular, two states~$\ket{\Sha^\star_\Delta(j)_{2^\ell}}$ and~$\ket{\Sha^\star_\Delta(k)_{2^\ell}}$ have disjoint support for~$j\neq k$, and  a homodyne measurement of the position-operator applied to the state~$\ket{\Sha^\star_\Delta(j)_{2^\ell}}$ and subsequent application of  the function 
\begin{align}
\begin{matrix}
\discretize_{2^\ell}:& \mathbb{R}&\rightarrow & \mathbb{Z}_{2^\ell}\\
&x& \mapsto & \mathsf{round}\left(x/\sqrt{2\pi\cdot 2^\ell}\right) \pmod {2^\ell}\ ,
\end{matrix}\label{eq:postprocessingdcode}
\end{align}
to the measurement outcome returns~$j$ with certainty. Here~$\mathsf{round}:\mathbb{R}\rightarrow\mathbb{Z}$ rounds to the nearest integer (breaking ties arbitrarily). That is, homodyne detection followed by classical post-processing given by~\eqref{eq:postprocessingdcode}  realizes a logical  computational qudit basis measurement for the code~$\gkpcoderect{\Delta}{\star}{2^\ell}$.
We can therefore simulate a measurement in the computational qubit tensor product basis on~$\gkpcoderect{\Delta}{\star}{2^\ell}\cong(\mathbb{C}^2)^{\otimes \ell}$ by using the post-processing map
\begin{align}
\begin{matrix}
\post_\ell :& \mathbb{R} & \rightarrow & \{0,1\}^\ell\\
& x & \mapsto & \post_\ell(x)&=\iota_\ell^{-1}(\discretize_{2^\ell}(x))
\end{matrix}\ ,\label{eq:postprocessingell}
\end{align}
where~$\iota_\ell:\{0,1\}^\ell\rightarrow \{0,\ldots,2^\ell-1\}$ is the bijection defined by~\eqref{eq:iotaelldefinition}. Eqs.~\eqref{eq:shafctsup}
and~\eqref{eq:postprocessingdcode}
and the definition of~$\post_\ell$ (Eq.~\eqref{eq:postprocessingell}) imply our claim, i.e., Eqs.~\eqref{eq:encmappostcontainmone}
and~\eqref{eq:encmappostcontainmtwo}.
\end{proof}

In our construction, the state~$V\ket{\Psi_{in}}$ before the measurement is (approximately) supported on the code space~$\cL_\Delta(m,\ell)$ (see Eq.~\eqref{eq:eqldeltamell}).
In particular, the 
state of the auxiliary mode~$B_{\mathsf{aux}}$ is~$\ket{\Sha_{\Delta,\ell}^{\mathsf{aux}}}$ whereas the three qubits~$Q_1Q_2Q_3$ are in the state~$\ket{0}^{\otimes 3}$. The logical information is encoded in the subspace~$(\gkpcoderect{\Delta}{\star}{2^\ell})^{\otimes m}$ of  the~$m$~modes~$B_1\cdots B_m$. 
(In fact, by construction, the information is
in the subspace 
$\encmapgkp_{\Delta,\ell}^{\otimes m}
\left((\mathbb{C}^2)^{\otimes n}\otimes \mathbb{C}\ket{0}^{\otimes (n'-m)}\right)\subset (\gkpcoderect{\Delta}{\star}{2^\ell})^{\otimes m}$
where the~$n'-m$ logical dummy qubits  are in the state~$\ket{0}$.) 
Correspondingly, our readout procedure only applies homodyne detection to the~$m$~modes (and either traces out the remaining systems and/or measures these and discards the measurement results). 

Now  consider a measurement result~$(y_1,\ldots,y_m)\in\mathbb{R}^m$ obtained when applying a homodyne position-measurement to each of the modes~$B_1\cdots B_m$. Our post-processing map applies the post-processing map~$\post_\ell$
from Lemma~\ref{lem:logicalmeasurementgkpcoderectdeltastar}
 to each value~$y_j$,~$j\in \{1,\ldots,m\}$. 
This results in an~$m$-tuple of~$\ell$-bit strings~$\left(\post_\ell(y_1),\ldots,\post_\ell(y_m)\right)\in (\{0,1\}^\ell)^m$.
We can identify~$(\{0,1\}^\ell)^m$ with~$\{0,1\}^{n'}$ by concatenating strings, and interpret this as an~$n'=\ell m$-bit string. Discarding the last~$n'-n$~bits using the map
\begin{align}
\begin{matrix}
\mathsf{discard}:&\{0,1\}^{n'}& \rightarrow &\{0,1\}^n\\
& (z_1,\ldots,z_{n},z_{n+1},\ldots,z_{n'}) & \mapsto & (z_1,\ldots,z_n)\ ,
\end{matrix}
\end{align}
finally gives an~$n$-bit string. That is, our overall post-processing map is 
\begin{align}
\begin{matrix}
\post:&\mathbb{R}^m & \rightarrow & \{0,1\}^n\\
& (y_1,\ldots,y_m) & \mapsto & \mathsf{discard}\left(\post_\ell(y_1),\ldots,\post_\ell(y_m)\right)\ .
\end{matrix}\label{eq:postprocessingmapjm}
\end{align}
The following statement shows that  this post-processing function applied to the measurement result of the homodyne detection emulates a logical computational basis measurement. It is an immediate consequence of
Lemma~\ref{lem:logicalmeasurementgkpcoderectdeltastar} and the linearity of the encoding map. For completeness, we give the details in Appendix~\ref{sec:proofoftheoremmeasurement}.
\begin{theorem}[Logical measurement]\label{thm:measurement} Let~$\ket{\Psi} \in (\mathbb{C}^2)^{\otimes n}$ be an~$n$-qubit state. 
Let 
\begin{align}
p(x):=|\langle x,\Psi\rangle|^2 \qquad\textrm{ for }\qquad x\in \{0,1\}^n
\end{align}
be the distribution of outcomes obtained when measuring~$\ket{\Psi}$ in the computational basis. 
Let
    \begin{align}
    \ket{\overline{\Psi}} = \encmapgkp_{\Delta,\ell}^{\otimes m}(\ket{\Psi} \otimes \ket{0}^{\otimes (n'-n)})\in \gkpcoderect{\Delta}{\star}{2^\ell}^{\otimes m}
    \end{align}
be the corresponding encoded state and  
\begin{align}
\overline{p}(x)&:=\int_{\mathsf{post}^{-1}(\{x\})} 
       |\overline{\Psi}(y_1,\ldots,y_m)|^2 dy_1\cdots dy_m\qquad\textrm{ for }\qquad x\in \{0,1\}^n\ 
\end{align}
be the distribution of outputs when applying a homodyne position-measurement to each of the~$m$~modes and post-processing the measurement result 
using the map \eqref{eq:postprocessingmapjm}.
Then 
\begin{align}
\overline{p}(x)&=p(x)\qquad\textrm{ for all }\qquad x\in \{0,1\}^n\ .
\end{align}
\end{theorem}

\paragraph{\bf Our scheme.}
We now combine the preparation procedure  of Theorem~\ref{thm:initialstateprep}, 
 the implementation of unitaries given in Theorem~\ref{thm:implementationlogicalqubit},
and the logical measurement described in Theorem~\ref{thm:measurement}. 
This results in the following scheme 
to realize an~$n'$-qubit circuit~$U=U_s\cdots U_1$ composed of~$s$ one- and two-qubit gates.
The scheme uses~$m+1$~oscillators and~$3$~qubits. We will make a distinction  between
\begin{enumerate}[(i)]
\item
$m$ ``system oscillators'' denoted~$B_1\cdots B_{m}\cong L^2(\mathbb{R})^{\otimes m}$, 
\item
one auxiliary oscillator denoted~$B_{\mathsf{aux}}\cong L^2(\mathbb{R})$.
\item
three auxiliary qubits denoted~$Q_1 Q_2 Q_3 \cong (\mathbb{C}^2)^{\otimes 3}$.
\end{enumerate}
It proceeds as follows:
\begin{enumerate}[(1)]
\item
It uses the unitary preparation procedure of Theorem~\ref{thm:initialstateprep}
for each pair~$B_jQ_1$,~$j\in \{1,\ldots,m\}$, as well as for the pair~$B_{\mathsf{aux}}Q_1$ to prepare an approximation~$\ket{\Phi_{\mathsf{init}}}_{B_1\cdots B_mB_{\mathsf{aux}}Q_1Q_2Q_3}$ to the state~$\ket{\Phi^{\mathsf{ideal}}_{\mathsf{init}}} = \ket{\Sha^\star_\Delta(0)_{2^\ell}}^{\otimes m}\otimes \ket{\Sha_{\Delta,\ell}^{\mathsf{aux}}(0)_2} \otimes \ket
{0}^{\otimes 3}$ on the system modes~$B_1\cdots B_m$, the auxiliary mode~$B_{\mathsf{aux}}$ and the auxiliary qubits~$Q_1Q_2Q_3$.
\item
It applies the implementation~$W_U$ of~$U$ 
described in Theorem~\ref{thm:implementationlogicalqubit} using the mode~$B_{\mathsf{aux}}$ and the auxiliary qubits~$Q_1Q_2 Q_3$.
\item
It applies a homodyne position-measurement to each of the~$m$ system modes~$B_1,\ldots,B_m$, and post-processes the result to obtain an~$n$-bit-string~$x\in \{0,1\}^{n}$.
\end{enumerate}
We note that the scheme  recycles the qubit~$Q_1$ based on the property that the preparation procedure of Theorem~\ref{thm:initialstateprep} approximately stabilizes the qubit in the state~$\ket{0}$ (cf. Lemma~\ref{lem:prepcombcombstate}). 
We note that this scheme can be written as a unitary circuit
\begin{align}
W^{\mathsf{tot}}&= W_U (W^{\mathsf{prep}}\otimes I_{\mathbb{C}^2}^{\otimes 2}) = W_{T}\cdots W_1
\end{align}
on~$L^2(\mathbb{R})^{\otimes m+1}\otimes (\mathbb{C}^2)^{\otimes 3}$ composed of
\begin{align}
T&=T^{\mathsf{prep}}(m,\Delta,\ell)+ T_{\mathsf{logical}}\\
&\leq 42m \log 1/\Delta + 340 s\ell^2
\end{align}
elementary operations~$W_1,\ldots,W_T\in \Uelem{m+1}{3}$.

\paragraph{\bf Error analysis.}  
To analyze how well the produced output distribution~$q$ (see Eq.~\eqref{eq:qdistribution}) approximates the target distribution~$p$ (defined by Eq.~\eqref{eq:samplingdistributionmain}), we use the formalism we introduced in~\cite{cliffordslinearoptics2025}. The corresponding framework applies to general approximate quantum error-correcting codes, and decomposes this task into the analysis of individual building blocks.  For approximate GKP codes, we can use the bounds worked out in~\cite{cliffordslinearoptics2025}.

 {\em Preparation error.} Our protocol
 approximately prepares the (ideal) initial state
 \begin{align}
 \ket{\Phi^{\mathsf{ideal}}_{\mathsf{init}}}=\left(\encmapgkp_{\Delta,\ell}\ket{0^\ell}\right)^{\otimes m}\otimes
 \ket{\Sha_{\Delta,\ell}^{\mathsf{aux}}(0)_2} \otimes \ket{0}^{\otimes 3}\ ,
 \end{align}
 i.e., an encoding of~$\ket{0}^{\otimes n'}$ in the code~$\cL_\Delta(m,\ell)$ (see Eq.~\eqref{eq:eqldeltamell}). According to Theorem~\ref{thm:initialstateprep}, the corresponding 
  ($L^1$-distance) error 
  between the state~$\ket{\Phi_{\mathsf{init}}}=(W^{\mathsf{prep}}\otimes I_{\mathbb{C}^2}^{\otimes 2}) (\ket{\vac}^{\otimes (m+1)}\otimes \ket{0}^{\otimes 3})$ prepared by the actual protocol
  and the ideal initial state can be bounded as 
  \begin{align}
\left\|\proj{\Phi_{\mathsf{init}}}-\proj{\Phi^{\mathsf{ideal}}_{\mathsf{init}}}\right\|_1&\leq   \varepsilon_{\mathsf{prep}}:=50m \left(\sqrt{\Delta}+ 2^{2\ell}\Delta^2\right)\ .\label{eq:preparationerrormain}
  \end{align}

{\em Gate error.} Now consider the (physical) implementation~$W_{U}$
of a (logical)~$n$-qubit circuit~$U$ 
(which we interpret as an $n'$-qubit circuit with trivial action on the dummy qubits) 
of size~$s=\size(U)$. According to Theorem~\ref{thm:implementationlogicalqubit}, this implementation has  a logical gate error of~$\gateerror{\cL_{\Delta}(m,\ell)}(W_{U},U)$ 
upper bounded by
\begin{align}
\varepsilon_{\mathsf{gate}}:=600 s\cdot 2^{2\ell} \Delta\ .
\end{align}

  {\em Error of the outcome distribution.} By definition of the gate error, it follows that the deviation 
of the final state~$\ket{\Phi_{\mathsf{out}}}:=W_U \ket{\Phi_{\mathsf{init}}}$ from the ideal encoded state
\begin{align}
\ket{\Phi^{\mathsf{ideal}}_{\mathsf{out}}}=(\encmapgkp_{\Delta,\ell})^{\otimes m}(U|0^{n'}\rangle)\otimes \ket{\Sha_{\Delta,\ell}^{\mathsf{aux}}(0)_2} \otimes\ket{0}^{\otimes 3}\label{eq:idealoutputstate}
\end{align}
satisfies
\begin{align}
\left\|
\proj{\Phi_{\mathsf{out}}}
-\proj{\Phi^{\mathsf{ideal}}_{\mathsf{out}}}
\right\|_1&\leq \varepsilon_{\mathsf{prep}}+\varepsilon_{\mathsf{gate}}=:
\varepsilon_{\mathsf{final}}(\Delta,\ell,m,s)\ .\label{eq:approximationerrorfinalstate}
\end{align}
We note that
 according to Theorem~\ref{thm:measurement}, applying a homodyne measurement to every system mode~$B_1,\dots, B_m$ of the ideal output state~$\ket{\Phi^{\mathsf{ideal}}_{\mathsf{out}}}$ and post-processing the result yields the distribution~$p$. By Eq.~\eqref{eq:approximationerrorfinalstate} and 
 the data-processing inequality (showing that post-processing a measurement-result does not increase the variational distance) it follows that our scheme approximates the desired ideal output distribution~$p$  with error 
\begin{align}
\left\|p-q\right\|_1 &\leq  \varepsilon_{\mathsf{final}}(\Delta,\ell,m,s)\ .\label{eq:pgdistancmv}
\end{align} 
This approximation is achieved by a circuit with at most~$O(s\ell^2 + m \log1/\Delta)$ gates from the elementary gate set~$\Uelem{m+1}{3}$ of our system of~$m+1$~oscillators and~$3$~qubits.

\paragraph{Bounding the amount of energy required.} 
To realize a logical circuit on~$n'= m\ell$ qubits using~$m+1$ modes and three auxiliary qubits, our scheme applies 
gates belonging to the set~$\Uelem{m+1}{3}$. This includes squeezing operations. The following result bounds the (maximal) amount of energy  generated in this process. 
\begin{theorem} \label{lem:totalsqueezingbound}
  Let~$U=U_s\cdots U_1$ be a unitary circuit on~$n'=m\ell$~qubits of size~$s$, i.e., composed of~$s$ one- and two-qubit gates~$U_1,\ldots,U_s$. Let~$W_U$ be the unitary acting on~$L^2(\mathbb{R})^{\otimes (m+1)}\otimes (\mathbb{C}^2)^{\otimes 3}$ which implements the circuit~$U$ as discussed in Theorem~\ref{thm:implementationlogicalqubit}.
Let~$W^{\mathsf{prep}}$ be the preparation circuit introduced in Theorem~\ref{thm:initialstateprep}.
Then the circuit~$W^{\mathsf{tot}} = W_U (W^{\mathsf{prep}} \otimes I_{\mathbb{C}^2}^{\otimes 2})$ satisfies 
\begin{align}
  \energy(W^\mathsf{tot}) \le s^3 \cdot 2^{891\ell + 62} /\Delta^{21}  \, .
\end{align}
\end{theorem}
The proof of Theorem~\ref{lem:totalsqueezingbound} is given in Section~\ref{sec: momlimitsimplementation}. 
It relies on the notion of fine-grained moment-limiting functions. We refer to Sections~\ref{sec:momentlimitfunctionsonemode},~\ref{sec:multi-modecasemomentlimit} for a detailed introduction and results about how moment-limiting functions can be used to bound the amount of energy used by a multimode circuit.

\paragraph{Completing the proof of Theorem~\ref{thm:main}.} 
Rephrasing Theorem~\ref{lem:totalsqueezingbound} and using the assumption $\Delta \le 2^{-(\ell+1)}$ from Theorem~\ref{thm:initialstateprep} we find
\begin{align}
\Delta &\leq \min\left\{2^{-(\ell+1)}, s^{3/21} \cdot 2^{(891\ell + 62)/21}\cdot  \energy(W^{\mathsf{tot}})^{-1/21}\right\}\ .\label{eq:kappasqueezeboundmv}
\end{align}
 Combining Eqs.~\eqref{eq:pgdistancmv} and~\eqref{eq:kappasqueezeboundmv}
 gives 
 \begin{align}
 \|p-q\|_1 & \le  50m(\sqrt{\Delta}+ 2^{2\ell}\Delta^2) + 600 s\cdot 2^{2\ell}\cdot \Delta\\
 &\le (50m + 600s) \cdot 2^{2\ell} \cdot \sqrt{\Delta} \\
 &\le 2^{10}\cdot (m+s) s^{3/21}\cdot 2^{(891\ell + 62)/42 + 2\ell} \cdot \energy(W^{\mathsf{tot}})^{-1/42}\\
 &\leq (m+s) s^{3/21} \cdot 2^{24\ell + 22} \cdot \energy(W^{\mathsf{tot}})^{-1/42} \\
 &\le (m+s)^2 \cdot 2^{24\ell + 22} \cdot \energy(W^{\mathsf{tot}})^{-1/42}\, .
 \end{align} In the second step we used that $\Delta < 1$ (see Eq.~\eqref{eq:kappasqueezeboundmv}). The third inequality follows from~Eq.~\eqref{eq:kappasqueezeboundmv}. 
 Using that~$\ell = n'/m \le n/m + 1$ we find 
 \begin{align}
  \|p-q\|_1 \le 2^{46} \cdot (s+m)^{2} \cdot 2^{24 n/m} \cdot \energy(W^{\mathsf{tot}})^{-1/42}\, . \label{eq:mainbound}
 \end{align}
 This implies the claim of Theorem~\ref{thm:main}.

 \section{Discussion and outlook\label{sec:discussion}}

We have shown that polynomial-size quantum computations on~$n$~(logical) qubits 
can be weakly (approximately) simulated in the hybrid qubit-oscillator model on~$L^2(\mathbb{R})^{\otimes m}\otimes (\mathbb{C}^2)^{\otimes r}$
with a constant number~$r=O(1)$ of qubits, a varying number~$m$ of bosonic modes, and various bounds on the amount of energy, see Table~\ref{tab:summaryachievability} for a summary.  These achievability results should be contrasted to
Table~\ref{tab:lowerboundsqueezing}, which 
gives lower bounds on the amount of energy in any family of~$2^n$~orthonormal states (i.e.,~$n$~logical qubits) encoded in a hybrid qubit-oscillator systems. Let us briefly discuss each of the three regimes considered in the different columns of these tables:

\begin{table}[h!]
\centering
\begin{tabular}{|c|c|c|c|}
\hline
{\#~$m$ of modes} &~$\Theta(1)$ &~$\Theta(n^\alpha)$,~$\alpha\in (0,1)$  &~$\Theta(n)$ \\
\hline
{amount of energy} &~$\exp(O(n))$ &~$\exp(O(n^{1-\alpha}))$ &~$O(\mathsf{poly}(n))$ \\
\hline
\end{tabular}
\caption{Summary of our achievability results 
for weakly simulating a polynomial-size circuit with an inverse polynomial error~$\varepsilon(n)=O(1/\mathsf{poly}(n))$, see Corollary~\ref{cor:squeezingtradeoff}.\label{tab:summaryachievability}}
\end{table}

\begin{table}[h!]
\centering
\begin{tabular}{|c|c|c|c|}
\hline
{\#~$m$ of modes} &~$\Theta(1)$ &~$\Theta(n^\alpha)$,~$\alpha\in (0,1)$  &~$\Theta(n)$ \\
\hline
{amount of energy} &~$\exp(\Omega(n))$ &~$\exp(\Omega(n^{1-\alpha}))$ &~$\Omega(1)$ \\
\hline
\end{tabular}
\caption{
This table gives lower bounds on the (maximal) amount of energy  
in any orthonormal family~$\{\varphi_j\}_{j=0}^{2^n-1}\subset L^2(\mathbb{R})^{\otimes m}\otimes (\mathbb{C}^2)^{\otimes r}$
of states encoded in a qubit-oscillator system with a constant number~$r=O(1)$ of auxiliary qubits, see Corollary~\ref{cor:constantnumberofmodessqueezinglowerbound}.
 Such a family corresponds to an encoding of~$n$~logical qubits.\label{tab:lowerboundsqueezing}}
\end{table}

\begin{enumerate}[(i)]
\item For a constant number~$m=\Omega(1)$ of modes,
the amount of energy generated in our protocol  matches the lower bound on the amount of energy of~$n$~encoded qubits. This suggests that this construction is optimal in terms of the amount of energy used.

We note that the regime of a constant number of modes and a constant number of qubits was previously considered in~\cite{brenner2024factoring} (co-authored by two of us) where a polynomial-time integer factoring algorithm
based on a hybrid qubit-oscillator system with~$(m,r)=(3,1)$ was proposed. 
This small number of oscillators and qubits in the construction of~\cite{brenner2024factoring} is achieved by using the bosonic modes  both to store and process information. 
Crucially, this requires preparing high-quality approximate (Gaussian envelope) GKP states, which results in an amount of energy which scales as~$\exp(\Theta(n^2))$. Furthermore, the corresponding preparation procedure
 uses qubit-controlled phase space rotations in addition to the elementary operations we consider here.

In contrast, in our present construction, bosonic modes are solely used as quantum memory while gates are performed on a constant number of physical qubits. This leads to the improved scaling of~$\exp(O(n))$ of the amount of energy required. (In addition, our construction sidesteps the need to execute qubit-controlled phase space rotations.)

Translating Shor's algorithm  -- or, more precisely, the corresponding quantum subroutine --  using our method therefore gives a more efficient factoring algorithm 
than the one proposed in~\cite{brenner2024factoring} in terms of the amount of energy required. (We note that Shor's sampling subroutine only requires achieving an error of order~$\varepsilon(n)=O(1/\log n)$, whereas  the error achieved in our scheme is inverse polynomial in~$n$.) Furthermore, if more modes are available, then this amount of energy can further be reduced, as follows.

\item For the case of a polynomial but sublinear number of modes~$m=\Theta(n^\alpha)$,~$\alpha\in (0,1)$, the amount of energy required in our construction is subexponential, and again matches the  dimension-dependent lower bound on the energy (see Table~\ref{tab:lowerboundsqueezing}).

\item Finally, in the case of  a linear number~$m=\Theta(n)$, we show that a polynomial amount of energy is sufficient for our scheme.
This is the most practically interesting case from the viewpoint of scalability. 
We note however, that our construction here does not match the constant lower bound (see Table~\ref{tab:lowerboundsqueezing}).

The origin of the polynomial scaling in our achievability result is the dependence of our bounds on the circuit size, which we assume to be polynomial: the upper bound on the amount of energy required in our construction to sample the output distribution of an~$n$-qubit circuit~$U$ (see Corollary~\ref{cor:squeezingtradeoff}) has a polynomial dependence on the circuit size~$s$ of~$U$. This arises because we bound the sampling error by~$s$~successive applications of the triangle inequality. 

  We anticipate that this circuit-size dependence could be reduced, or even eliminated, by incorporating intermediate error correction.  Specifically, using a linear number of modes, $m = \Theta(n)$, 
we expect that an amount of energy of order $O(1)$ should be sufficient to achieve e.g., an inverse-polynomial error~$\varepsilon(n)$.
\end{enumerate}

We expect the different realizations of our procedure 
to be useful for experimental quantum computing 
in hybrid qubit-oscillator systems. Specifically, our trade-off relation
allows us to determine what computations are realizable when the number of available modes, as well as the amount of energy which can be generated is limited by a given experimental setup.

On a more theoretical level, our results contribute what could be called complexity theory for continuous-variable systems in the direction of hybrid qubit-oscillator setups.  This goes in the direction of Ref.~\cite{chabaud2024bcomplexity},
but with the cubic phase gate used in the definition of the computational complexity class CVBQP replaced by  qubit-oscillator couplings. In the hybrid qubit-oscillator setup, our work directly advances the question of computational complexity under energy constraints. This addresses a problem posed  in Ref.~\cite{chabaud2024bcomplexity}.

Several open complexity-theoretic questions related to hybrid qubit-oscillator systems remain. For example, one could ask 
for upper bounds on the computational power of this setup. Similar to the work Ref.~\cite{upreti2025boundingCbosonic} on the complexity class CVBQP, such bounds could be obtained by devising classical simulation algorithms for hybrid qubit-oscillator circuits. Finally, one could try to compare the computational power of different models of CV quantum computation based on different sources of non-Gaussianity.

\section*{Acknowledgments}

LB, BD and RK gratefully acknowledge support by the European Research Council under
grant agreement no. 101001976 (project EQUIPTNT), as well as the Munich Quantum
Valley, which is supported by the Bavarian state government through the Hightech Agenda
Bayern Plus. 
\newpage
\appendix

\section{Proof of Theorem~\ref{thm:measurement}    \label{sec:proofoftheoremmeasurement}}
For completeness, we restate the claim from the main text.
\begingroup
\renewcommand{\thetheorem}{\ref{thm:measurement}}
\begin{theorem}[Restated] Let~$\ket{\Psi} \in (\mathbb{C}^2)^{\otimes n}$ be an~$n$-qubit state. 
  Let 
  \begin{align}
  p(x):=|\langle x,\Psi\rangle|^2 \qquad\textrm{ for }\qquad x\in \{0,1\}^n
  \end{align}
  be the distribution of outcomes obtained when measuring~$\ket{\Psi}$ in the computational basis. 
  Let
      \begin{align}
      \ket{\overline{\Psi}} = \encmapgkp_{\Delta,\ell}^{\otimes m}(\ket{\Psi} \otimes \ket{0}^{\otimes (n'-n)})\in \gkpcoderect{\Delta}{\star}{2^\ell}^{\otimes m}
      \end{align}
  be the corresponding encoded state and  
  \begin{align}
  \overline{p}(x)&:=\int_{\mathsf{post}^{-1}(\{x\})} 
         |\overline{\Psi}(y_1,\ldots,y_m)|^2 dy_1\cdots dy_m\qquad\textrm{ for }\qquad x\in \{0,1\}^n\ 
  \end{align}
  be the distribution of outputs when applying a homodyne position-measurement to each of the~$m$~modes and post-processing the measurement result 
  using the map \eqref{eq:postprocessingmapjm}.
  Then 
  \begin{align}
  \overline{p}(x)&=p(x)\qquad\textrm{ for all }\qquad x\in \{0,1\}^n\ .
  \end{align}
  \end{theorem}
\addtocounter{theorem}{-1} 
\endgroup
\begin{proof}
Expanding~$\ket{\Psi}$ in the computational basis, we have 
\begin{align}
  \ket{\Psi} = \sum_{x \in \{0,1\}^n} \sqrt{p(x)}e^{i\theta_x} \ket{x}
\end{align}
for some phases~$\theta_x\in\mathbb{R}$,~$x\in \{0,1\}^n$.  By linearity of the embedding map~$\encmapgkp_{\Delta,\ell}$ we have
  \begin{align}
    \ket{\overline{\Psi}} = \sum_{x \in \{0,1\}^n} \sqrt{p(x)}e^{i\theta_x}\ket{\phi_{x}}\, ,
  \end{align}
  where we introduced the states
  \begin{align} 
  \ket{\phi_{x}} = \encmapgkp_{\Delta,\ell}^{\otimes m}\left(\ket{x} \otimes \ket{0}^{\otimes (n'-n)}\right)\in L^2(\mathbb{R})^{\otimes m} \qquad \textrm{for}\qquad x \in \{0,1\}^n\, .
  \end{align}
Let~$x \in \{0,1\}^n$ be arbitrary. 
Using that $\ket{\overline{\Psi}}\in L^2(\mathbb{R})^{\otimes m}\cong L^2(\mathbb{R}^m)$ we have 
\begin{align}
  \overline{p}(x) &= \int _{\post^{-1}(\{x\})}|\overline{\Psi}(y)|^2 dy\\
  &=\sum_{a_1,a_2\in \{0,1\}^n} \sqrt{p(a_1) p(a_2)} e^{i(\theta_{a_2} - \theta_{a_1})} \int_{\post^{-1}(\{x\})} \overline{\phi_{a_1}(y_1, \dots, y_m)} \phi_{a_2}(y_1, \dots, y_m) dy_1 \cdots dy_m\\
  &= \sum_{a\in \{0,1\}^n} p(a)\int_{\post^{-1}(\{x\})} |\phi_{a}(y_1, \dots, y_m)|^2 dy_1 \cdots dy_m\, ,
\end{align}
where we used that the functions~$\{\phi_a\}_{a\in \{0,1\}^n}\subset 
L^2(\mathbb{R})^{\otimes m}\cong L^2(\mathbb{R}^m)$ have pairwise disjoint support by construction. By construction  
(see Eqs.~\eqref{eq:encmappostcontainmone} and~\eqref{eq:encmappostcontainmtwo} of Lemma~\ref{lem:logicalmeasurementgkpcoderectdeltastar})
the sets~$\post^{-1}(\{x\})$ and~$\supp(\phi_{a})$ are disjoint unless~$x=a$, and~$\supp(\phi_{x}) \subset \post^{-1}(\{x\})$. 
 It follows that
\begin{align}
 \overline{p}(x) &= p(x) \|\phi_x\|^2=p(x) 
\end{align}
as claimed. Here we used that~$\phi_x$ is normalized because~$\encmapgkp_{\Delta,\ell}^{\otimes m}$ is an isometry. 
\end{proof}

\section{Elementary unitary operations in the hybrid model}
Consider a system of~$m$ oscillator and~$r$~qubits, i.e., Hilbert space~$L^2(\mathbb{R})^{\otimes m}\otimes (\mathbb{C}^2)^{\otimes r}$. 
We often consider the set of 
the unitaries
\begin{align} \label{eq:defgroupgenmr}
\Uelemunbounded{m}{r}:= \{&\ctrl_a e^{-itP_j},\ctrl_a e^{itQ_j},e^{-itP_j}, e^{itQ_j},(M_\beta)_j\}_{\substack{t\in\mathbb{R},\,\beta>0\\
j\in \{1,\ldots,m\}\\
a\in \{1,\ldots,r\}}}\\
\qquad &\cup \{U_a,U_{a,b}\ |\ U_a, U_{a,b} \textrm{ one- or two-qubit unitaries}\}_{a,\,b\in \{1,\ldots,r\}}\ \label{eq:mygroupdefinitionmultiqubit}
\end{align}
consisting of (qubit-)controlled single-mode displacements, single-mode squeezing as well as one- and two-qubit unitaries. Here we omit identities and denote the qubits and  oscillators the operators act on by indices~$a,b\in \{1,\ldots,r\}$ and~$j\in \{1,\ldots,m\}$, respectively. The group  generated by these unitaries will be denoted
$\langle \Uelemunbounded{m}{r}\rangle$.

The set~$\Uelemunbounded{m}{r}$ introduced in Eq.~\eqref{eq:defgroupgenmr} includes unitaries of arbitrary strength. In contrast, the set~$\Uelem{m}{r}$ of elementary unitary operations we are primarily interested in consists of the subset of bounded-strength unitaries.
It will be convenient to introduce the subset
\begin{align}
\Uelem{m}{r}(\alpha,\zeta):=\{&\ctrl_a e^{-itP_j},\ctrl_a e^{itQ_j},e^{-itP_j}, e^{itQ_j},(M_\beta)_j\}_{\substack{t\in (-\zeta,\zeta),\,\beta\in (\alpha^{-1},\alpha)\\
j\in \{1,\ldots,m\}\\
a\in \{1,\ldots,r\}}}\\
\qquad &\cup \{U_a,U_{a,b}\ |\ U_a, U_{a,b} \textrm{ one- or two-qubit unitary }\}_{a,\,b\in \{1,\ldots,r\}}\ ,\label{eq:mygroupdefinitionmultiqubitbounded}
\end{align}
of unitaries with displacement and squeezing bounded by~$\zeta\geq 1$ and~$\alpha\geq 1$, respectively.
Then we formally have~$\Uelem{m}{r}=\bigcup_{\zeta,\alpha\in O(1)} \Uelem{m}{r}(\zeta,\alpha)$.

In the following, we derive bounds on the amount of energy generated by an element~$U\in \Uelem{m}{r}(\alpha,\zeta)$ (for fixed parameters~$(\zeta,\alpha)$). More generally, we establish bounds on the amount of energy generated by elements~$U\in \langle \Uelem{m}{r}\rangle=\langle \Uelem{m}{r}(\alpha,\zeta)\rangle$ specified as products  (circuits)~$U=U_T\cdots U_1$, where~$U_t\in\Uelem{m}{r}(\alpha,\zeta)$ for each~$t\in \{1,\ldots,T\}$.

\section{Moment-limits in the one-mode case} \label{sec:momentlimitfunctionsonemode}
In this section, we introduce the notion of a moment-limiting function
of a unitary~$U$ and give explicit examples. We restrict to unitaries on~$L^2(\mathbb{R})\otimes (\mathbb{C}^2)^{\otimes r}$, i.e., we consider the one-mode case with~$r\geq 1$ qubits.

\subsection{Definition and basic properties of moment-limiting functions}
Recall that the  Fourier transform~$\cF: L^2(\bb{R}) \rightarrow L^2(\bb{R})$ is the unique unitary acting on an element~$f \in L^1(\mathbb{R}) \cap L^2(\mathbb{R})$
 as 
\begin{align}
\mathcal{F}(f)(p)=\widehat{f}(p)=\frac{1}{\sqrt{2 \pi}} \int f(x) e^{-i p x} d x\ . \label{eq: integral def Fourier}
\end{align}
For~$R_1,R_2\in\mathbb{R}$, let~$\Pi_{[R_1,R_2]}$ denote the projection onto the subspace of~$L^2(\mathbb{R})$ of functions having support on~$[R_1,R_2]$. Similarly, let~$\widehat{\Pi}_{[R_1,R_2]}=\cF^\dagger \Pi_{[R_1,R_2]}\cF$ denote the projection onto the subspace of~$L^2(\mathbb{R})$ of functions whose Fourier transform has support on~$[R_1,R_2]$.  We note that~$\Pi_{[R_1,R_2]}$ is a spectral projection associated with the position-operator~$Q$, whereas~$\widehat{\Pi}_{[R_1,R_2]}$ is a spectral projection associated with the momentum-operator~$P$.

It will be convenient to introduce the following notion.
\begin{definition}[Fine-grained moment-limiting function]\label{def:finegrainedsinglemode}
Consider a fixed unitary~$U\in \langle \Uelemunbounded{1}{r}\rangle$ and an entrywise invertible affine linear transformation 
\begin{align}
\begin{matrix}
\varphi:&\mathbb{R}^4 & \rightarrow &\mathbb{R}^4\\
& (R_1,R_2,\widehat{R}_1,\widehat{R}_2)&\mapsto & (a_1 R_1+b_1,a_2R_2+b_2,\widehat{a}_1 \widehat{R}_1+\widehat{b}_1,\widehat{a}_2\widehat{R}_2+\widehat{b}_2)\ .
\end{matrix}\label{eq:affinelinearoperation}
\end{align}
Then~$\varphi$ is called  a fine-grained moment-limiting function for~$U$
if  for all~$R_1,R_2,\widehat{R}_1,\widehat{R}_2\in\mathbb{R}$, we have the operator inequalities 
\begin{align}
\begin{matrix} U \Pi_{[R_1,R_2]}U^\dagger&\leq \Pi_{[R_1',R'_2]}\\
 U \widehat{\Pi}_{[\widehat{R}_1,\widehat{R}_2]}U^\dagger&\leq \widehat{\Pi}_{[\widehat{R}_1',\widehat{R}_2']}
 \end{matrix}\qquad\textrm{ where }\qquad (R_1',R_2',\widehat{R}_1',\widehat{R}_2'):=\varphi(R_1,R_2,\widehat{R}_1,\widehat{R}_2)\ .
\end{align}
\end{definition}

\subsubsection{Explicit moment-limiting functions for generators\label{sec:explicitmomentlimitingfunctions}}

In Section~\ref{sec:explicitmomentlimitingfunctions}, we give explicit fine-grained moment-limiting functions for the generators~$U\in\Uelemunbounded{1}{r}$.
We will then argue that a fine-grained moment-limiting function can be obtained in terms of two parameters~$\eta(V)$ and~$\xi(V)$ only, see Lemma~\ref{lem:varphiUVbasic} for a detailed statement.  

The relevant parameters~$(\eta(V),\xi(V))$ for a generator~$V\in\Uelemunbounded{1}{r}$ are defined as follows. 
\begin{definition}[Squeezing and displacement parameters of generators]\label{def:squeezingdisplacementparameters}
    For~$V\in \Uelemunbounded{1}{r}$, we define a pair~$(\eta(U),\xi(U))\in  (0,\infty)\times [0,\infty)$ of squeezing and displacement parameters as 
  \begin{align}
    \eta(V) &= 
    \begin{cases}
      \alpha &\text{ if } V = M_\alpha\qquad\textrm{ for some~$\alpha>0$} \\
      1 & \text{otherwise}\ ,
    \end{cases}  \\
    \xi(V) &= 
    \begin{cases}
      |t| &\text{ if } V \in \{e^{it P}, e^{it Q}, \ctrl_a e^{it P}, \ctrl_a e^{it Q}\}\qquad\textrm{ for some }t\in \mathbb{R}, a\in \{1,\ldots,r\} \\
      0 & \text{otherwise}\ .
    \end{cases}
  \end{align}
\end{definition}

We start by establishing the following operator inequalities associated with generators, i.e., elements of~$\Uelemunbounded{1}{r}$.
\begin{lemma}[Fine-grained moment-limiting functions for generators]\label{lem:momentlimitsonunitaries}
Let~$\alpha>0$,~$t\in\mathbb{R}$  and~$a\in \{1,\ldots,r\}$ be arbitrary. Then the following holds.
\begin{enumerate}[(i)]
\item
The functions \begin{align}
\varphi_{e^{-itP}}(R_1,R_2,\widehat{R}_1,\widehat{R}_2)&=(R_1+t,R_2+t,\widehat{R}_1,\widehat{R}_2) \ , \\
\varphi_{e^{itQ}}(R_1,R_2,\widehat{R}_1,\widehat{R}_2)&=(R_1,R_2,\widehat{R}_1+t,\widehat{R}_2+t) \ , \\
\varphi_{M_\alpha}(R_1,R_2,\widehat{R}_1,\widehat{R}_2)&=(\alpha R_1,\alpha R_2,\widehat{R}_1/\alpha,\widehat{R}_2/\alpha)\ 
\end{align}
are fine-grained moment-limiting functions for~$e^{-itP}$,~$e^{itQ}$ and~$M_\alpha$, respectively.
\item
The functions 
\begin{align}
\varphi_{\ctrl_a e^{-itP}}(R_1,R_2,\widehat{R}_1,\widehat{R}_2)&=(R_1-|t|,R_2+|t|,\widehat{R}_1,\widehat{R}_2) \ , \\
\varphi_{\ctrl_a e^{itQ}}(R_1,R_2,\widehat{R}_1,\widehat{R}_2)&=(R_1,R_2,\widehat{R}_1-|t|,\widehat{R}_2+|t|)\ 
\end{align}
are fine-grained moment-limiting functions for~$\ctrl_a e^{-itP}$ and~$\ctrl_a e^{itQ}$.
\item \label{eq:claimiii}
The function
\begin{align}
\varphi_U(R_1,R_2,\widehat{R}_1,\widehat{R}_2)&=(R_1,R_2,\widehat{R}_1,\widehat{R}_2)
\end{align}
is a fine-grained moment-limiting function for any one- or two-qubit unitary~$U$.
\end{enumerate}
\end{lemma}
\begin{proof}
It follows  from the definitions that 
\begin{align}
e^{-itP}\Pi_{[R_1,R_2]}e^{itP}&=\Pi_{[R_1+t,R_2+t]}\label{eq:invarianceronetwoone}\\
e^{-itP}\widehat{\Pi}_{[R_1,R_2]}e^{itP}&=\widehat{\Pi}_{[R_1,R_2]}\ .\label{eq:invariancewidehatpironertwo}
\end{align}
Similarly, we have 
\begin{align}
e^{itQ}\Pi_{[R_1,R_2]}e^{-itQ}&=\Pi_{[R_1,R_2]}\\
e^{itQ}\widehat{\Pi}_{[R_1,R_2]}e^{-itQ}&=\widehat{\Pi}_{[R_1+t,R_2+t]}\ .
\end{align}
Finally, since~$M_\alpha^\dagger QM_\alpha=\alpha Q$ and~$M_\alpha^\dagger PM_\alpha=\alpha^{-1} P$ we have
\begin{align}
M_\alpha\Pi_{[R_1,R_2]}M_\alpha^\dagger&=\Pi_{[\alpha R_1,\alpha R_2]}\\
M_\alpha\widehat{\Pi}_{[R_1,R_2]}M_\alpha^\dagger&=\widehat{\Pi}_{[ R_1/\alpha, R_2/\alpha]}\ .
\end{align}

We claim that 
\begin{align} \label{eq:qubitctrlprojbb}
  \ctrl e^{-it P} \left(\Pi_{[R_1,R_2]} \otimes I\right) \ctrl e^{it P}  &\leq  \Pi_{[R_1-|t|,R_2+|t|]}  \otimes I \\
    \ctrl e^{-it P} \left( \widehat{\Pi}_{[R_1,R_2]} \otimes I \right) \ctrl e^{it P}  &=  \widehat{\Pi}_{[R_1,R_2]}  \otimes I  \label{eq:qubitctrlprojbbsecond}
\end{align} 
Eq.~\eqref{eq:qubitctrlprojbb} follows from
\begin{align}
\ctrl e^{-it P} (\Pi_{[R_1,R_2]}\otimes I) \ctrl e^{it P} &= \Pi_{[R_1,R_2]} \otimes \proj{0}  +  e^{-it P} \Pi_{[R_1,R_2]} e^{it P} \otimes \proj{1} \\
  &= \Pi_{[R_1,R_2]} \otimes \proj{0} + \Pi_{[R_1+t,R_2+t]}  \otimes \proj{1} \\
  &\leq  \Pi_{[R_1-|t|,R_2+|t|]} \otimes I
\end{align} 
where we used Eq.~\eqref{eq:invarianceronetwoone} and the fact that~$[R_1,R_2]\subset [R_1-|t|,R_2+|t|]$ and~$[R_1+t,R_2+t]\subset [R_1-|t|,R_2+|t|]$ for~$t\in\mathbb{R}$ to obtain the last operator inequality.
Eq.~\eqref{eq:qubitctrlprojbbsecond} follows immediately from~\eqref{eq:invariancewidehatpironertwo}.

By similar arguments, we can show that
\begin{align} \label{eq:qubitctrlprojbbx}
  \ctrl e^{it Q} \left(\Pi_{[R_1,R_2]} \otimes I \right) \ctrl e^{-it Q}  &=  \Pi_{[R_1,R_2]}  \otimes I \\
    \ctrl e^{it Q} \left(\widehat{\Pi}_{[R_1,R_2]}  \otimes I \right) \ctrl e^{-it Q}  &\leq \widehat{\Pi}_{[R_1-|t|,R_2+|t|]}  \otimes I  \ .\label{eq:qubitctrlprojbbsecondx}
\end{align} 

Claim~\eqref{eq:claimiii} follows from the fact that one- and two-qubit unitaries act trivially on the space~$L^2(\bbR)$.
\end{proof}

For two functions~$f,g:\mathbb{R}\rightarrow \mathbb{R}$, we write~$f\leq g$ if and only if~$f(R)\leq g(R)$ for all~$R\in\mathbb{R}$.  The following definition will be useful.
\begin{definition}
Let~$\varphi,\chi:\mathbb{R}^4\rightarrow\mathbb{R}^4$ be two functions of the form
\begin{align}
\varphi(R_1,R_2,\widehat{R}_1,\widehat{R}_2)&=(\varphi^{(1)}(R_1),\varphi^{(2)}(R_2),\varphi^{(3)}(\widehat{R}_1),\varphi^{(4)}(\widehat{R}_2))\\
\chi(R_1,R_2,\widehat{R}_1,\widehat{R}_2)&=(\chi^{(1)}(R_1),\chi^{(2)}(R_2),\chi^{(3)}(\widehat{R}_1),\chi^{(4)}(\widehat{R}_2))
\end{align}
for~$(R_1,R_2,\widehat{R}_1,\widehat{R}_2)\in\mathbb{R}^4$. We say that~$\chi$ dominates~$\varphi$, and denote this  as~$\varphi\leq \chi$ if
\begin{align}
\begin{aligned}
\varphi^{(1)}&\geq \chi^{(1)}\\
\varphi^{(2)}&\leq \chi^{(2)}
\end{aligned}\qquad\textrm{ and }\qquad 
\begin{aligned}
\varphi^{(3)}&\geq \chi^{(3)}\\
\varphi^{(4)}&\leq \chi^{(4)} \ .
\end{aligned}
\end{align}
\end{definition}
\noindent We note that by definition, the condition~$\varphi\leq \chi$ is equivalent to the inclusions
\begin{align}
\begin{aligned}
[\varphi^{(1)}(R_1),\varphi^{(2)}(R_2)]&\subseteq [\chi^{(1)}(R_1),\chi^{(2)}(R_2)]\\
[\varphi^{(3)}(\widehat{R}_1),\varphi^{(4)}(\widehat{R}_2)]&\subseteq [\chi^{(3)}(\widehat{R}_1),\chi^{(4)}(\widehat{R}_2)]
\end{aligned}\qquad\textrm{ for all }R_1,R_2,\widehat{R}_1,\widehat{R}_2\in\mathbb{R}\ .\label{eq:inclusionpropertybnd}
\end{align}
The significance of this definition is the following lemma. 
\begin{lemma}\label{lem:monotonicityfinegrained}
Let~$U\in \langle \Uelemunbounded{1}{r}\rangle$ and~$\varphi_U:\mathbb{R}^4\rightarrow\mathbb{R}^4$ 
be fine-grained moment-limiting function for~$U$.  
Let~$\chi:\mathbb{R}^4\rightarrow\mathbb{R}^4$ be an invertible entrywise affine-linear function such that~$\varphi_U\leq \chi$. Then~$\chi$ is a fine-grained moment-limiting function for~$U$. 
\end{lemma}
\begin{proof}
This follows immediately from the definition of a fine-grained moment-limiting function
and the alternative characterization~\eqref{eq:inclusionpropertybnd} of the condition~$\varphi_U\leq \chi$. 
\end{proof}

The following gives a fine-grained moment-limiting function  for any generator in terms of its squeezing and displacement parameters. 
\begin{lemma}  \label{lem:varphiUVbasic}
  Let~$U\in  \Uelemunbounded{1}{r}$. Let~$(\eta(U),\xi(U))\in (0,\infty)\times [0,\infty)$ be the squeezing and displacement parameters introduced in Definition~\ref{def:squeezingdisplacementparameters}.  Define 
  \begin{align}
  \chi_U(R_1,R_2,\widehat{R}_1,\widehat{R}_2)&:=\left(\eta(U)R_1-\xi(U),
  \eta(U)R_2+\xi(U),\eta(U)^{-1}\widehat{R}_1-\xi(U),\eta(U)^{-1}\widehat{R}_2+\xi(U) 
  \right)      \label{eq:varchiUprime}
  \end{align}
  for~$(R_1,R_2,\widehat{R}_1,\widehat{R}_2)\in \mathbb{R}^4$. Then~$\chi_U$ is a fine-grained moment-limiting function for~$U$.
\end{lemma}

\begin{proof}
Recall from Lemma~\ref{lem:momentlimitsonunitaries} that for any~$t\in\mathbb{R}$,~$\alpha>0$ and~$a\in \{1,\ldots,r\}$ we have the moment-limiting functions 
  \begin{align}
    \varphi_{e^{-itP}}\left(R_1, R_2, \widehat{R}_1, \widehat{R}_2\right) &= (R_1 + t, R_2 + t, \widehat{R}_1, \widehat{R}_2) \ , \\
    \varphi_{e^{itQ}}\left(R_1, R_2, \widehat{R}_1, \widehat{R}_2\right) &= (R_1, R_2, \widehat{R}_1 + t, \widehat{R}_2+ t) \ , \\
    \varphi_{M_\alpha}\left(R_1, R_2, \widehat{R}_1, \widehat{R}_2\right) &= (\alpha R_1, \alpha R_2, \widehat{R}_1 / \alpha, \widehat{R}_2 / \alpha) \ , \\
    \varphi_{\ctrl_a e^{-itP}}\left(R_1, R_2, \widehat{R}_1, \widehat{R}_2\right) &= (R_1-|t|, R_2 + |t|, \widehat{R}_1, \widehat{R}_2) \ , \\
    \varphi_{\ctrl_a e^{itQ}}\left(R_1, R_2, \widehat{R}_1, \widehat{R}_2\right) &= (R_1, R_2, \widehat{R}_1-|t|, \widehat{R}_2 + |t|)  \ , \\
      \varphi_{W}\left(R_1, R_2, \widehat{R}_1, \widehat{R}_2\right) &= (R_1, R_2, \widehat{R}_1, \widehat{R}_2) 
  \end{align}
  for the single-mode unitaries~$e^{-itP}$ and $e^{itQ}$, the squeezing unitary~$M_\alpha$, the qubit-controlled unitaries~$\ctrl_a e^{-itP}$ and~$\ctrl_a e^{itQ}$, and any one- or two-qubit unitary~$W$. On the other hand, we have 
  \begin{align}
  (\eta(U),\xi(U)):=\begin{cases}
  (1,|t|)\qquad &\textrm{ if } U\in \{e^{-itP},e^{itQ},\ctrl e^{-itP},\ctrl e^{itQ}\}\\
  (\alpha,0)&\textrm{ if }U=M_\alpha\\
  (1,0)&\textrm{ if }U\textrm{ is a one- or two-qubit unitary}\ 
  \end{cases}  
  \end{align}
  by the definition of the squeezing and displacement parameters.
  It follows that~$\varphi_U\leq \chi_U$. This implies the claim because of Lemma~\ref{lem:monotonicityfinegrained}.
  \end{proof}

\begin{lemma}[Fine-grained moment-limiting function and energy]\label{lem:finegrainedenergy}
Let $U\in  \langle\Uelemunbounded{1}{r}\rangle$.
Let $(\eta,\xi)\in (0,\infty)\times [0,\infty)$ be such that 
  \begin{align}
  \chi_U(R_1,R_2,\widehat{R}_1,\widehat{R}_2)&:=\left(\eta R_1-\xi,
  \eta R_2+\xi,\eta^{-1}\widehat{R}_1-\xi,\eta ^{-1}\widehat{R}_2+\xi
  \right)      \label{eq:varchiUprime}
  \end{align}
  is a fine-grained moment-limiting function for $U$ (cf. Definition~\ref{def:finegrainedsinglemode}).
  There are two trivariate polynomials $u(x_1,x_2,x_3)$ and $v(x_1,x_2,x_3)$, where both $u$ and $v$ are sums of
  bivariate polynomials of degree at most~$3$ in each variable such that the following holds. The evolved state
  \begin{align}
\ket{\Psi}&=U^\dagger (\ket{\mathsf{vac}}\otimes \ket{0})\ 
\end{align}
has energy upper bounded by
\begin{align} \label{eq:claimenegryuv}
  \langle \Psi, (Q^2+P^2)\Psi\rangle &\leq u(\eta,\eta^{-1},\xi)+v(\eta,\eta^{-1},\xi)\ .
  \end{align}
Moreover, the polynomials $u$ and $v$ satisfy
\begin{align}
  \max_{\eta \in [1/q,q], \xi \in [0,s]}
   \left(u(\eta,\eta^{-1},\xi)+v(\eta,\eta^{-1},\xi)\right)&\leq 168 q^3 (2+s^3)
   \ \qquad\textrm{ for }\qquad q\geq 1\textrm{ and }s\geq 0\ .\label{eq:etabboundestimateadm}
  \end{align}
\end{lemma}
\begin{proof}
Let us omit identities on the qubit for brevity.
It is easy to see that 
\begin{align}
Q^2 & \leq \sum_{z\in \mathbb{N}} z^2(\Pi_{[z-1,z]}+\Pi_{[-z,-z+1]})\ .
\end{align}
It follows from the definition of fine-grained moment-limiting functions that
\begin{align}
UQ^2 U^\dagger & \leq \sum_{z\in \mathbb{N}}
z^2(\Pi_{[
\eta (z-1)-\xi,\eta z+\xi]}+\Pi_{[\eta(-z)-\xi,\eta(-z+1)+\xi]})\\
&=:\Omega\ .\label{eq:Uqudgerr}
\end{align}
For $\Psi\in L^2(\mathbb{R})$ the operator $\Omega$ is a multiplication operator acting as
\begin{align}
(\Omega \Psi)(x)&=\omega(x)\Psi(x)\qquad \textrm{ for }\qquad x\in \mathbb{R}
\end{align}
where 
\begin{align}
\omega(x)&:=\sum_{\substack{z\in \mathbb{N}:\\
x\in [\eta z-\eta-\xi,\eta z+\xi]}}z^2
+ \sum_{\substack{z\in \mathbb{N}:\\
x\in [-\eta z-\xi,-\eta z+\eta+\xi]}}z^2\ .
\end{align}
We have $x\in [\eta z-\eta-\xi,\eta z+\xi]$ if and only if
\begin{align}
\eta z-\eta-\xi\leq x\leq \eta z+\xi
\end{align}
or equivalently
\begin{align}
z-1\leq x/\eta+\xi/\eta\leq z+2\xi/\eta\ .
\end{align}
In other words, the value $\overline{x}:=x/\eta+\xi/\eta$ has to be contained in an interval of length $2\xi/\eta+1$ containing the integer~$z$. 
There are at most
\begin{align}
\lceil 2\xi/\eta+1\rceil&\leq 2\xi/\eta+2
\end{align}
such integers $z\in\mathbb{N}$, and each such integer
is upper bounded by
\begin{align}
\overline{x}+2\xi/\eta+1&\leq x/\eta+\xi/\eta+(2\xi/\eta+1)\\
&=x/\eta +3\xi/\eta+1\ .
\end{align}
It follows from this that
\begin{align}
\sum_{\substack{z\in \mathbb{N}:\\
x\in [\eta z-\eta-\xi,\eta z+\xi]}}z^2
\leq (2\xi/\eta+2)(x/\eta +3\xi/\eta+1)^2 \ .
\end{align}
By similar reasoning,  we have
\begin{align}
x\in [-\eta z-\xi,-\eta z+\eta+\xi]
\end{align}
if and only if
\begin{align}
-\eta z-\xi\leq x\leq -\eta z+\eta+\xi
\end{align}
or
\begin{align}
-z&\leq x/\eta+\xi/\eta \leq -z+1+2\xi/\eta\ , 
\end{align}
that is,
\begin{align}
z-1-2\xi/\eta\leq -\overline{x}\leq z \ ,
\end{align}
i.e., $-\overline{x}$ is contained in an interval of length $2\xi/\eta+1$ around the integer~$z\in \mathbb{N}$. Such an integer $z$ is necessary upper bounded as
\begin{align}
z&\leq -\overline{x}+2\xi/\eta+1\\
&=-x/\eta-\xi/\eta+2\xi/\eta+1\\
&=-x/\eta+\xi/\eta+1 \ .
\end{align}
It follows that
\begin{align}
\sum_{\substack{z\in \mathbb{N}:\\
x\in [-\eta z-\xi,-\eta z+\eta+\xi]}}z^2\
&\leq (2\xi/\eta+2)(-x/\eta+\xi/\eta+1)^2 \ .
\end{align}
In summary, we obtain
\begin{align}
\omega(x)&\leq (2\xi/\eta+2)
\left((x/\eta+3\xi/\eta+1)^2+
(-x/\eta+\xi/\eta+1)^2\right)\\
&=c_2x^2+c_1 x+c_0
\end{align}
where
\begin{align}
  c_2 &= 4 (\xi/\eta^3 + 1/\eta^2)=:p_2(\eta^{-1},\xi) \\
  c_1 &= 8  (\xi^2/\eta^3 + \xi /\eta^2) =:p_1(\eta^{-1},\xi)\\
  c_0 &= 4 + 20 \xi/\eta + 36 \xi^2/\eta^2 + 
 20 \xi^3/\eta^3 =:p_0(\eta^{-1},\xi)
\end{align}
for bivariate polynomials $p_0,p_1,p_2$
of degree at most~$3$ in each variable.
In particular, we conclude that
\begin{align}
\Omega &\leq p_2(\eta^{-1},\xi) Q^2+ p_1(\eta^{-1},\xi) Q+ p_0(\eta^{-1},\xi)\ . 
\end{align}
This implies that
\begin{align}
UQ^2 U^\dagger &\leq p_2(\eta^{-1},\xi)Q^2+p_1(\eta^{-1},\xi)Q+p_0(\eta^{-1},\xi)\ \label{eq:positionupperboundq}
\end{align}
by Eq.~\eqref{eq:Uqudgerr}. 

By identical arguments for $P^2$ (working in Fourier space with the operators $\widehat{\Pi}_{[R_1,R_2]}$, and a corresponding multiplication operator~$\widehat{\Omega}$ in the momentum-basis, we obtain (by exchanging $\eta$ with $\eta^{-1}$) the operator inequality
\begin{align}
UP^2 U^\dagger &\leq p_2(\eta,\xi)P^2+p_1(\eta,\xi)P+p_0(\eta,\xi)\ .\label{eq:momentumupperboundq}
\end{align}

Consider a state of the form 
\begin{align}
\ket{\Psi}&=U^\dagger (\ket{\mathsf{vac}}\otimes \ket{0})\ .
\end{align}
Combining Eqs.~\eqref{eq:positionupperboundq} and~\eqref{eq:momentumupperboundq} gives
\begin{align}
\langle \Psi, (Q^2+P^2) \Psi\rangle
&=(\langle \mathsf{vac}|\otimes \langle 0|)
U (Q^2+P^2)U^\dagger (\ket{\mathsf{vac}}\otimes\ket{0})\\
&=(p_0(\eta,\xi)+p_0(\eta^{-1},\xi))
+ (p_2(\eta,\xi)+p_2(\eta^{-1},\xi))
\end{align}
where we used that  
\begin{align}
\langle \mathsf{vac},Q\mathsf{vac}\rangle&=\langle \mathsf{\vac},P\mathsf{vac}\rangle=0\\
\langle \mathsf{vac},Q^2\mathsf{vac}\rangle&=\langle \mathsf{\vac},P^2\mathsf{vac}\rangle=1\ .
\end{align}
Claim~\eqref{eq:claimenegryuv} follows from this by setting
\begin{align}
u(x_1,x_2,x_3)&=p_0(x_1,x_3)+p_0(x_2,x_3)\\
v(x_1,x_2,x_3)&=p_2(x_1,x_3)+p_2(x_2,x_3)\ .
\end{align}
Next, we prove Claim~\eqref{eq:etabboundestimateadm}.
Since
\begin{align}
u(\eta,\eta^{-1},\xi)&=u(\eta^{-1},\eta,\xi)\\
v(\eta,\eta^{-1},\xi)&=v(\eta^{-1},\eta,\xi)\ .
\end{align}
we may without loss of generality assume that $\eta \in [1,q]$. We then have 
\begin{align}
u(\eta,\eta^{-1},\xi)&=
8 + 20(\eta + 1/\eta) \xi
+36(\eta^2 + 1/\eta^2) \xi^2
+20(\eta^3 + 1/\eta^3) \xi^3 \\
&\leq 8 + 20(\eta + 1) \xi
+36(\eta^2 + 1) \xi^2
+20(\eta^3 + 1) \xi^3 \\
&\leq 8 + (q^3+1)(20 \xi
+36 \xi^2
+20 \xi^3) \\
&\leq 8 + 2q^3(20 \xi
+36 \xi^2
+20 \xi^3)  \\
&\leq 8 + 2q^3(20 s
+36 s^2
+20 s^3) \\ 
&\leq 8 + 152 q^3(1+s^3)  \\ 
&\leq 160 q^3(1+s^3) 
\end{align}
where we used that
\begin{align}
20 s
+36 s^2
+20 s^3 &\leq \begin{cases}
 20+36+20=76\qquad \textrm{ for }s\leq 1\\
 66s^3\qquad\textrm{ for }s\geq 1\ .
\end{cases}
\end{align}
and
\begin{align}
v(\eta,\eta^{-1},\xi)
&= 4 \xi (\eta^3 + 1/\eta^3) + 4(\eta^2 + 1/\eta^2) \\
&\leq 4 \xi (\eta^3 + 1) + 4(\eta^2 + 1) \\
&\leq 4 \xi (\eta^3 + 1) + 4(\eta^2 + 1) \\
&\leq 4 \xi (q^3 + 1) + 4(q^2 + 1) \\
&\leq 4 (1 + \xi) (q^3 + 1) \\
&\leq 8 (1 + \xi) q^3  \\
&\leq 8 (1 + s) q^3 \ .
\end{align}
This implies Claim~\eqref{eq:etabboundestimateadm} since $s\leq 1+s^3$ for all $s\geq 0$.
\end{proof}

\subsection{Moment-limiting functions for circuits} \label{sec:momlimitcircuits}

In this section, we derive fine-grained moment-limiting functions for circuits~$U=U_T\cdots U_1\in \langle \Uelemunbounded{1}{r}\rangle$ composed of unitaries~$U_t\in \Uelemunbounded{1}{r}$ for~$t\in \{1,\ldots,T\}$. We achieve this by introducing two parameters~$\overline{g}(U)$ and~$\overline{\xi}(U)$ which can be understood as a generalization of the ``local'' parameters~$\eta(V)$ and~$\xi(V)$ of the generators~$V\in \Uelemunbounded{1}{r}$.

The relevant quantities are  defined as follows, where we use the squeezing and displacement parameters 
$(\eta(V),\xi(V))\in (0,\infty)\times [0,\infty)$ for every generator~$V\in\Uelemunbounded{1}{r}$ (see  Definition~\ref{def:squeezingdisplacementparameters}). 

\begin{definition}[Squeezing and displacement parameters of circuits]\label{def:squeezingdisplacementcircuits}
Let~$g:(0,\infty)\rightarrow (0,\infty)$ denote the function~$g(x):=\max \{x,1/x\}$. 
Consider a product 
\begin{align}
U=U_T \cdots U_1\qquad\textrm{ with  }\qquad U_t\in\Uelemunbounded{1}{r}\quad \textrm{for }\quad t\in\{1,\ldots,T\}\ .\label{eq:Ucircuitdvz}
\end{align}
Define the quantities
  \begin{align}
  \squeezingparam(U)&:=\max_{t\in \{1,\ldots,T\}} \max_{p\in \{0,\ldots,T-t\}}g\left(\prod_{s=1}^t \eta(U_{p+s})\right)\\\
  \overline{\xi}(U)&:=  \sum_{j=1}^T \xi(U_j)\ .
  \end{align}
   We call~$(\squeezingparam(U),\overline{\xi}(U))$ the squeezing and displacement parameters of the circuit~$U$. 
  \end{definition}
\noindent We note that in the definition of~$\squeezingparam(U)$, the inner maximum is over all (products of) consecutive sequences~$U_{p+1},\ldots,U_{p+t}$ of unitaries of length~$t$, i.e., subcircuits~$U_{p+1}\cdots U_{p+t}$ of size~$t$.

The following is an immediate consequence of the definitions.
\begin{lemma}[Squeezing and displacement parameters of adjoint circuits]\label{lem:daggerrelationoverlineg}
Let~$U$ be a circuit as in Eq.~\eqref{eq:Ucircuitdvz}. Then 
\begin{align}
\squeezingparam(U^\dagger)&=\squeezingparam(U)\label{eq:sqmvzv}\\
\overline{\xi}(U^\dagger)&=\overline{\xi}(U) .\label{eq:xidefmzv}
\end{align}
\end{lemma}
\begin{proof}
Eq.~\eqref{eq:xidefmzv}
 follows immediately from the definitions: We have
\begin{align}
\overline{\xi}(U^\dagger)&=\sum_{t=1}^T \xi(U_t^\dagger)=\sum_{t=1}^T\xi(U_t)=\overline{\xi}(U)
\end{align}
 because~$\xi(U^\dagger)=\xi(U)$ for every~$U\in \Uelem{1}{r}$, see Definition~\ref{def:squeezingdisplacementparameters}.

To prove Eq.~\eqref{eq:sqmvzv}, let us set~$V_t:=U^\dagger_{T-t+1}$ for~$t\in \{1,\ldots,T\}$. Then~$U^\dagger=V_T\cdots V_1$. 
It follows from the definitions that there are~$t\in \{1,\ldots,T\}$ and~$p\in \{0,\ldots,T-t\}$ such that 
\begin{align}
  \squeezingparam(U^\dagger)&=
  g\left(\prod_{s=1}^t \eta(V_{p+s})\right)\\
  &=g\left(\prod_{s=1}^t \eta(U^\dagger_{T-(p+s)+1})\right)\\
  &=g\left(\prod_{s=1}^t \eta(U_{T-(p+s)+1})^{-1}\right)\\
  &=g\left(\prod_{s=1}^t \eta(U_{T-(p+s)+1})\right) \ .
  \end{align}
Here we used that~$\eta(U^\dagger)=1/\eta(U)$ and~$g(1/x)=g(x)$ by definition of~$g$. We can rewrite this as 
\begin{align}
  \squeezingparam(U^\dagger)&=g\left(\prod_{s=1}^{t}\eta(U_{p'+s})\right)
\end{align}
where~$p':=T-(p+t)$. 
Since this corresponds to a subcircuit~$U_{p'+t}\cdots U_{p'+1}$ of~$U$, it follows that
\begin{align}
  \squeezingparam(U^\dagger)&\leq \squeezingparam(U)\ .
\end{align}
Interchanging the roles of~$(U,U^\dagger)$ gives the claim.
\end{proof}

To argue that the quantities~$(\overline{\xi}(U),\squeezingparam(U))$ give rise to a moment-limiting function for the circuit~$U=U_T\cdots U_1$ (respectively partially implemented versions~$U^{(t)}=U_t\cdots U_1$), we need to study compositions of moment-limiting functions. 
\begin{lemma}\label{lem:compositionfinegrained}
Let~$U_1,U_2\in \langle \Uelemunbounded{1}{r}\rangle$. Let~$\varphi_{1},\varphi_{2}:\mathbb{R}^4\rightarrow\mathbb{R}^4$ be
 fine-grained moment-limiting functions for~$U_1$ and~$U_2$, respectively. Then
the composed map~$\varphi_2\circ \varphi_1:\mathbb{R}^4\rightarrow\mathbb{R}^4$ is a fine-grained moment-limiting function for the composition~$U_2U_1$.
\end{lemma}
\begin{proof}
This follows immediately from the definition of a fine-grained moment-limiting function.
\end{proof}

Using Lemma~\ref{lem:compositionfinegrained}, we obtain the following  fine-grained moment-limiting function for any circuit composed of generators. We again denote by~$(\eta(V),{\xi}(V))\in (0,\infty)\times [0,\infty)$ the squeezing and displacement parameters introduced in Definition~\ref{def:squeezingdisplacementparameters}
of a generator~$V\in\Uelemunbounded{1}{r}$.
\begin{lemma}[Fine-grained moment-limit functions for circuits]
\label{lem:varphiUcontactenatedbound}
Let~$U=U_T \cdots U_1$ with~$U_t\in\Uelemunbounded{1}{r}$ for~$t\in\{1,\ldots,T\}$.
Define
\begin{align}
\eta&=\prod_{t=1}^T \eta(U_t)
\qquad\textrm{ and }\qquad
\begin{aligned}
\xi&=\sum_{j=1}^T  \xi(U_j) \prod_{s=j+1}^m \eta(U_s)\\
\widehat{\xi}&=\sum_{j=1}^T  \xi(U_j) \prod_{s=j+1}^m \eta^{-1}(U_s)\ .
\end{aligned}
\end{align}
Then
\begin{align}
    \chi_U(R_1,R_2,\widehat{R}_1,\widehat{R}_2)&=(\eta R_1-{\xi},\eta R_2+{\xi},\eta^{-1}\widehat{R}_1-{\widehat{\xi}},\eta^{-1}\widehat{R}_2+{\widehat{\xi}})
\end{align}
is a fine-grained moment-limiting function for~$U$.
\end{lemma}

\begin{proof}
We first note that affine-linear functions compose as follows.
For~$\alpha,\xi>0$
define~$f_{\alpha,\xi}(R)=\alpha R+\xi$. Let~$\alpha_j>0$ and~$\xi_j\in\mathbb{R}$ for~$j\in\{1,\ldots, m\}$. Then
\begin{align}
\label{eq:affineeq} f_{\alpha_m,\xi_m}\circ \cdots\circ f_{\alpha_1,\xi_1} &= f_{\prod_{j=1}^m \alpha_j, \sum_{j=1}^m \xi_j\prod_{s=j+1}^m \alpha_s} \ . 
\end{align}
Eq.~\eqref{eq:affineeq} can be shown by induction.

For~$t\in \{1,\ldots,T\}$, let us write~$\eta_t:=\eta(U_t)$ and~$\xi_t:=\xi(U_t)$, and let us define
\begin{align}
  \chi_t(R_1,R_2,\widehat{R}_1,\widehat{R}_2)&:=(\eta_t R_1-\xi_t,\eta_t R_2+\xi_t,\eta_t^{-1}\widehat{R}_1-\xi_t,\eta_t^{-1}\widehat{R}_2+\xi_t)\ .
\end{align}
Then~$\chi_t$ is a fine-grained moment-limiting function for~$U_t$ according to Lemma~\ref{lem:varphiUVbasic}. With Lemma~\ref{lem:compositionfinegrained} (used inductively), it follows that 
\begin{align}
  \chi&:=\chi_T\circ \cdots \circ\chi_1
\end{align}
is a fine-grained moment-limiting function for~$U$.  
Straightforward computation using Eq.~\eqref{eq:affineeq} gives
\begin{align}
\chi(R_1,R_2,\widehat{R}_1,\widehat{R}_2)&=(\eta R_1-\xi,\eta R_2+\xi,\eta^{-1}\widehat{R}_1-\xi,\eta^{-1}\widehat{R}_2+\xi)
\end{align}
where
\begin{align}
\eta&=\prod_{t=1}^T \eta_t
\qquad\textrm{ and }\qquad
\begin{aligned}
\xi&=\sum_{j=1}^T  \xi_j \prod_{s=j+1}^m \eta_s\\
\widehat{\xi}&=\sum_{j=1}^T  \xi_j \prod_{s=j+1}^m \eta^{-1}_s\ .
\end{aligned}
\end{align}
This is the claim.
\end{proof}

By a partial implementation of a circuit~$U_T\cdots U_1$ we mean a product~$U_t\cdots U_1$ with~$t<T$. We show 
that the energy of any intermediate state in a partially implement circuit can be bounded as follows. This result is for the case of $1$~mode and $r$~qubits. 

\begin{lemma}[Fine-grained moment-limiting function for (partially implemented) circuits and energy: single-mode case]  \label{lem:PhiUpartiallyimmplementedfinegrained}
Let $r\in\mathbb{N}_0$.
  Let~$U=U_T \cdots U_1$ with~$U_t\in\langle \Uelemunbounded{1}{r}\rangle$ for~$t\in\{1,\ldots,T\}$ be given. Define    \begin{align}
   U^{(t)}&=\begin{cases}
   I\qquad&\textrm{ for } t=0\\
      U_t\cdots U_1&\textrm{ otherwise}\ 
      \end{cases}
      \end{align}
      and
      \begin{align}
      \ket{\Psi^{(t)}}&:=(U^{(t)})^\dagger(\ket{\mathsf{vac}}\otimes\ket{0}^{\otimes r})\ 
      \end{align}
      for each~$t\in \{0,\ldots,T\}$. 
      Let~$(\squeezingparam(U),\overline{\xi}(U))$ be the squeezing and displacement parameters of the circuit~$U$ introduced in Definition~\ref{def:squeezingdisplacementcircuits}.
        Then 
  \begin{align}
  \langle \Psi^{(t)},(Q^2+P^2)\otimes I_{(\mathbb{C}^2)^{\otimes r}}\Psi^{(t)}\rangle &\leq 168 \squeezingparam(U)^6\cdot \left(2+\overline{\xi}(U)^3\right)\qquad\textrm{ for  each }\qquad t\in \{0,\ldots,T\}\ .
   \end{align}
        \end{lemma}
\begin{proof}
Let~$t\in \{1,\ldots,T\}$ be arbitrary. Let~$\chi_{U^{(t)}}:\mathbb{R}^4\rightarrow \mathbb{R}^4$
be the function
  \begin{align}
  \chi_{U^{(t)}}(R_1,R_2,\widehat{R}_1,\widehat{R}_2)&:=\left(\eta^{(t)}R_1-\xi^{(t)},
  \eta^{(t)}R_2+\xi^{(t)},(\eta^{(t)})^{-1}\widehat{R}_1-\widehat{\xi}^{(t)},(\eta^{(t)})^{-1}\widehat{R}_2+\widehat{\xi}^{(t)}
  \right) 
  \end{align}
defined using 
\begin{align}
\eta^{(t)}&=\prod_{s=1}^t \eta(U_s)
\qquad\textrm{ and }\qquad
\begin{aligned}
\xi^{(t)}&=\sum_{j=1}^t  \xi(U_j) \prod_{r=j+1}^m \eta(U_r)\\
\widehat{\xi}^{(t)}&=\sum_{j=1}^t  \xi(U_j) \prod_{r=j+1}^m \eta^{-1}(U_r)\ .
\end{aligned}
\end{align}
According to 
Lemma~\ref{lem:varphiUcontactenatedbound}, the function~$\chi_{U^{(t)}}$ is a fine-grained moment-limiting function for~$U^{(t)}$. Defining
\begin{align}
b^{(t)}&:=\max\{\xi^{(t)},\widehat{\xi}^{(t)}\}
\end{align}
it follows that
\begin{align}
  \psi_{t}(R_1,R_2,\widehat{R}_1,\widehat{R}_2)&:=\left(\eta^{(t)}R_1-b^{(t)},
  \eta^{(t)}R_2+b^{(t)},(\eta^{(t)})^{-1}\widehat{R}_1-b^{(t)},(\eta^{(t)})^{-1}\widehat{R}_2+b^{(t)}
  \right) 
  \end{align}
is a fine-grained moment-limiting function for $U^{(t)}$ (see Lemma~\ref{lem:monotonicityfinegrained}). It thus follows from Lemma~\ref{lem:finegrainedenergy}
that
\begin{align}
  \langle \Psi^{(t)},(Q^2+P^2)\Psi^{(t)}\rangle &\leq 
  u(\eta^{(t)},(\eta^{(t)})^{-1},b^{(t)})+v(\eta^{(t)},(\eta^{(t)})^{-1},b^{(t)})\
  \end{align}
 It is easy to check that
 \begin{align}
\max \{\eta^{(t)},(\eta^{(t)})^{-1}\}&\leq \squeezingparam(U)\\
b^{(t)}&\leq \squeezingparam(U)\overline{\xi}(U)\ ,
 \end{align}
 that is, 
 \begin{align}
 \eta^{(t)}&\in [\squeezingparam(U)^{-1},\squeezingparam(U)]\\
b^{(t)}&\leq \squeezingparam(U)\overline{\xi}(U)
 \end{align}
 for any $t\in \{0,\ldots,T\}$.  
 The claim thus follows from  
 \begin{align}
 \max_{\eta \in [1/q,q], b\leq s}
  \left(u(\eta,\eta^{-1},b)+v(\eta,\eta^{-1},b)\right)&\leq 
  168 q^3 (2+s^3)\ \qquad\textrm{ for }\qquad q\geq 1\textrm{ and }s\geq 0\ ,
 \end{align}
see Eq.~\eqref{eq:etabboundestimateadm} in Lemma~\ref{lem:finegrainedenergy}, with 
 \begin{align}
 q&=\squeezingparam(U)\\
 s&=\squeezingparam(U)\overline{\xi}(U)\ , 
 \end{align}
 which implies that
 \begin{align}
 \langle \Psi^{(t)},(Q^2+P^2)\Psi^{(t)}\rangle
 &\leq  168 \squeezingparam(U)^3\cdot \left(2+\squeezingparam(U)^3 \overline{\xi}(U)^3\right)\\
 &\leq 168 \squeezingparam(U)^6\cdot \left(2+\overline{\xi}(U)^3\right)\ .
 \end{align}Here we used that $\squeezingparam(U)\geq 1$ by definition.  This is the claim. 
 \end{proof}

\subsection{Squeezing and displacement parameters in terms of subcircuits} \label{sec:momlimsubcircuits}
In this section, we show how the squeezing and displacement parameters of circuits (see Definition~\ref{def:squeezingdisplacementcircuits}) can be bounded in terms of the squeezing and displacement parameters of the respective subcircuits.

\begin{lemma}[Squeezing and displacement parameters in terms of subcircuits]\label{lem:squeezingdisplacementsubcircuits} 

Let~$\{U^{(a)}\}_{a\in \{1,\ldots,L\}}\subset \langle \Uelemunbounded{1}{r}\rangle$ be a family of circuits, where for each~$a\in \{1,\ldots,L\}$, the circuit~$U^{(a)}$ is of the form
\begin{align}
U^{(a)}&=U_{T^{(a)}}^{(a)}\cdots U^{(a)}_1\qquad\textrm{ with }\qquad U^{(a)}_t\in \Uelemunbounded{1}{r}\quad\textrm{ for each }\quad t\in \{1,\ldots,T^{(a)}\}\ .
\end{align}
Define the quantities
\begin{align}
\begin{aligned}
\eta(U^{(a)})&=\prod_{t=1}^{T^{(a)}}\eta(U_{t}^{(a)})\\
\overline{\xi}(U^{(a)})&=\sum_{t=1}^{T^{(a)}}\xi(U_t^{(a)})
\end{aligned}\qquad\textrm{ for each }\qquad a\in \{1,\ldots,L\}\ .
\end{align}
Consider the circuit
\begin{align}
U&=\prod_{a=1}^L U^{(a)}=\prod_{a=1}^L \left(U_{T^{(a)}}^{(a)}\cdots U_1^{(a)}\right)\ .
\end{align}
Then the squeezing and displacement parameters~$(\overline{g}(U),\overline{\xi}(U))$ of the circuit~$U$ (see Definition~\ref{def:squeezingdisplacementcircuits}) satisfy
\begin{align}
\overline{\xi}(U)&= \sum_{a=1}^L \overline{\xi}(U^{(a)})\label{eq:immediatevdaxiU} \ ,\\
\squeezingparam(U)&\leq \prod_{a=1}^L \overline{g}(U^{(a)}) \label{eq: bound galpha circuitone} \ .
\end{align}
Furthermore, we have 
\begin{align}
\squeezingparam(U)&\leq \left(\max_{a\in \{1,\ldots,L\}}\squeezingparam(U^{(a)})\right)^2\cdot \prod_{a=1}^L g(\eta(U^{(a)}))\ \label{eq: bound galpha circuittwo}
\end{align}
where we again use the function~$g(x)=\max\{x,1/x\}$. 
\end{lemma}
We note that each term~$\overline{\xi}(U^{(a)})$ in Eq.~\eqref{eq:immediatevdaxiU}
is the squeezing parameter of a full implementation of~$U^{(a)}$. Similarly, each term~$g(\eta(U^{(a)}))$
is associated with the squeezing introduced by a full implementation of~$U^{(a)}$.
In contrast, 
the scalar~$\squeezingparam(U^{(a)})$ quantifies the squeezing
for a possibly partial implementation of~$U^{(a)}$ (see Definition~\ref{def:squeezingdisplacementcircuits}). 
\begin{proof}
Let~$T:=\sum_{a=1}^L T^{(a)}$ be the size of~$U$ when decomposed into elements of~$\Uelemunbounded{1}{r}$. Let us write
\begin{align}
U&:=U_T\cdots U_1\qquad\textrm{ where }\qquad U_t\in \Uelemunbounded{1}{r}\qquad\textrm{ for each }\qquad t\in \{1,\ldots,T\}. 
\end{align}
By definition, we have
\begin{align}
\overline{\xi}(U)&=\sum_{t=1}^T 
\xi(U_t)=\sum_{a=1}^L \overline{\xi}(U^{(a)})\ 
\end{align}
which is Claim~\eqref{eq:immediatevdaxiU}.

Let us show Eq.~\eqref{eq: bound galpha circuitone}. Suppose that for some~$t\in \{1,\ldots,T\}$ and~$p\in \{0,\ldots,T-t\}$, we have 
 \begin{align}
 \overline{g}(U)&=g\left(\prod_{s=1}^t \eta(U_{p+s})\right)\ ,\label{eq:overlinegudefn}
 \end{align}
 i.e., the maximum is achieved on the subcircuit~$U_{p+t}\cdots U_{p+1}$. 
 It is easy to check that the subcircuit is a product
 \begin{align}
 U_{p+t}\cdots U_{p+1}&=V^{(a+b-1)}\cdots V^{(a)}
 \end{align}
 for some~$a$ and $b$, where each factor~$V^{(a)}$ is a subcircuit (product of consecutive gates) of~$U^{(a)}$.
 Using the identity
 \begin{align}
g(xy)\leq g(x)g(y)\qquad\textrm{ for all }\qquad x,y>0\label{eq:productpropertyg}
\end{align}
and Eq.~\eqref{eq:overlinegudefn} it follows that 
\begin{align}
\overline{g}(U)&\leq \overline{g}(V^{(a+b-1)})\cdots \overline{g}(V^{(a)})\\
&\leq \overline{g}(U^{(a+b-1)})\cdots \overline{g}(U^{(a)})\\
&\leq \prod_{a=1}^L \overline{g}(U^{(a)})
\end{align}
where in the last line, we used that 
\begin{align}
g(x)\geq 1\qquad\textrm{ for all } x>0\ .\label{eq:gxgeqbound}
\end{align}
This establishes Eq.~\eqref{eq: bound galpha circuitone}.

The proof of Eq.~\eqref{eq: bound galpha circuittwo} proceeds in a similar fashion. For any~$t\in \{1,\ldots,T\}$ and~$p\in \{0,\ldots,T-t\}$, there exist
$a,c$ and~$u,v$ such that
\begin{align}
\prod_{s=1}^t U_{p+s}
&=\left(U_v^{(c+1)}\cdots U_1^{(c+1)}\right)
\left(\prod_{b=a}^c U^{(b)}\right)\left(
U^{(a-1)}_{T^{(a-1)}}\cdots U^{(a-1)}_{u}
\right)\\
&=\left(U_v^{(c+1)}\cdots U_1^{(c+1)}\right)
\left(\prod_{b=a}^c
U^{(b)}_{T^{(b)}}\cdots U^{(b)}_1
\right)\left(
U^{(a-1)}_{T^{(a-1)}}\cdots U^{(a-1)}_{u}
\right)\ .
\end{align}
In other words, any consecutive product of unitaries~$\{U_t\}_{t=1}^T$ 
is a product of 
\begin{enumerate}[(i)]
\item
a partial implementation of~$U^{(a-1)}$ or the identity,
\item
a full implementation of each~$U^{(b)}$, with~$b\in \{a,\ldots,c\}$,
\item
a partial implementation of~$U^{(c+1)}$  or the identity.
\end{enumerate}
By the same reasoning as before, we have 
\begin{align}
\prod_{r=1}^t\eta(U_{p+s})&=
\left(\eta(U_v^{(c+1)})\cdots \eta(U_1^{(c+1)})\right)
\left(\prod_{b=a}^c
\eta(U^{(b)}_{T^{(b)}})\cdots \eta(U^{(b)}_1)
\right)\left(
\eta(U^{(a-1)}_{T^{(a-1)}})\cdots \eta(U^{(a-1)}_{u})
\right)\\
&=\left(\eta(U_v^{(c+1)})\cdots \eta(U_1^{(c+1)})\right)
\left(\prod_{b=a}^c
\eta(U^{(b)})
\right)\left(
\eta(U^{(a-1)}_{T^{(a-1)}})\cdots \eta(U^{(a-1)}_{u})
\right) \ .
\end{align}
Using Eq.~\eqref{eq:productpropertyg} 
we obtain
\begin{align}
g\left(
\prod_{s=1}^t\eta(U_{p+s})\right)
&\leq  g\left(\eta(U_v^{(c+1)})\cdots \eta(U_1^{(c+1)})\right)
\left(\prod_{b=a}^c
g(\eta(U^{(b)}))
\right)
g\left(
\eta(U^{(a-1)}_{T^{(a-1)}})\cdots \eta(U^{(a-1)}_{u})
\right)\\
&\leq \squeezingparam(U^{(c+1)})\left(\prod_{b=a}^c
g(\eta(U^{(b)}))
\right)
\squeezingparam(U^{(a-1)})\ .
\end{align}
Using Eq.~\eqref{eq:gxgeqbound} 
we can bound this as
\begin{align}
g\left(\prod_{s=1}^t\eta(U_{p+s})\right)
&\leq \left(\max_{a\in \{1,\ldots,L\}}\squeezingparam(U^{(a)})\right)^2\cdot \prod_{a=1}^L g(\eta(U^{(a)}))\ .
\end{align}
Since~$p$ and~$t$ were arbitrary,
we obtain Eq.~\eqref{eq: bound galpha circuittwo} as claimed.
\end{proof}

\begin{figure}[H]
  \centering
  \includegraphics[scale=0.9]{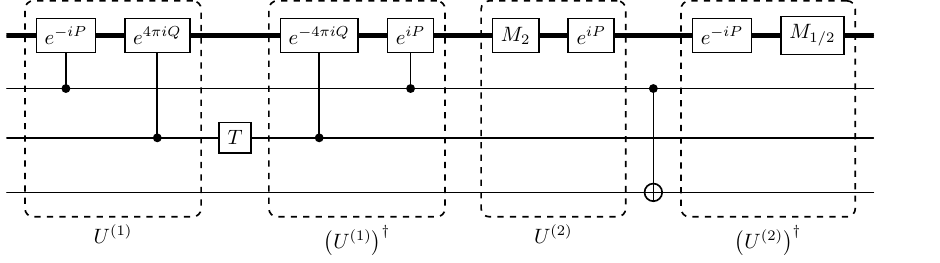}
  \caption{ An example of a dressed circuit introduced in Lemma~\ref{lem:momentlimitdressedcircuit} with qubit gates~$V^{(1)} = T_2$ and~$V^{(2)} = \mathsf{CNOT}_{1,3}$. The bold wires represent oscillators whereas the remaining three wires represent physical qubits.} 
  \label{fig:dressedcircuit}
\end{figure}

We often consider circuits where gates acting on qubits only are conjugated with
unitaries acting on oscillators. We refer to these as dressed circuits, see  Fig.~\ref{fig:dressedcircuit} for an example. 
We have the following:
\begin{lemma}[Moment-limits for dressed circuits] \label{lem:momentlimitdressedcircuit}
   Let~$\{U^{(a)}\}_{a\in \{1,\dots,L\}}\subset \Uelemunbounded{1}{r}$ be a family of circuits  where for each~$a\in \{1,\ldots,L\}$, the circuit~$U^{(a)}$ is of the form
   \begin{align}
   U^{(a)}&=U_{T^{(a)}}^{(a)}\cdots U^{(a)}_1\qquad\textrm{ with }\qquad U^{(a)}_t\in \Uelemunbounded{1}{r}\qquad\textrm{ for each }\qquad t\in \{1,\ldots,T^{(a)}\} \ .
   \end{align}
    Let~$\{V^{(a)}\}_{a \in \{1,\dots,L\}}$ be a sequence of one- or two-qubit unitaries on~$(\mathbb{C}^2)^{\otimes r}$.   Define the circuit 
  \begin{align}
  U = \prod_{a=1}^{L}\left(\left(U^{(a)}\right)^\dagger \left(I_{L^2(\mathbb{R})}\otimes V^{(a)}\right) U^{(a)}\right) \, .
\end{align}
  Then 
\begin{align}
\overline{\xi}(U)&= 2\sum_{a=1}^L \overline{\xi}(U^{(a)})\label{eq:xiubarvaz}\\
\squeezingparam(U)&\leq \left(\max_{a\in \{1,\ldots,L\}}\squeezingparam(U^{(a)})\right)^2\ .\label{eq:sqzingparm}
\end{align}
\end{lemma}
\begin{proof}
Define
\begin{align}
W^{(a)}=\left(U^{(a)}\right)^\dagger \left(I_{L^2(\mathbb{R})}\otimes V^{(a)}\right) U^{(a)}\qquad\textrm{ for }\qquad a\in \{1,\ldots,L\}\ 
\end{align}
such that~$U=\prod_{a=1}^L W^{(a)}$. It follows that
\begin{align}
\overline{\xi}(U)&=\sum_{a=1}^L \overline{\xi}(W^{(a)})\\
\overline{g}(U)&\leq \left(\max_{a\in \{1,\ldots,L\}}\overline{g}(W^{(a)})\right)^2\cdot \prod_{a=1}^L g(\eta(W^{(a)}))\label{eq:umvdavz}
\end{align}
according to Lemma~\ref{lem:squeezingdisplacementsubcircuits} (see Eqs.~\eqref{eq:immediatevdaxiU} and~\eqref{eq: bound galpha circuittwo}). Eq.~\eqref{eq:xiubarvaz} follows since for every~$a\in \{1,\ldots,L\}$ we have 
\begin{align}
 \overline{\xi}(W^{(a)})&=\sum_{t=1}^{T^{(a)}}\left(\xi(U^{(a)}_t)+\xi(I_{L^2(\mathbb{R})}\otimes V^{(a)} )+\xi((U^{(a)}_t)^\dagger)\right)\\
 &=2\sum_{t=1}^{T^{(a)}} \xi(U^{(a)}_t)\\
 &=2\overline{\xi}(U^{(a)})\, ,
\end{align}
where we used that~$\xi(U)=\xi(U^\dagger)$ for every~$U\in\Uelemunbounded{1}{r}$ by Definition~\ref{def:squeezingdisplacementparameters}.

To show Eq.~\eqref{eq:sqzingparm}, observe that for every~$a\in \{1,\ldots,L\}$ we have 
\begin{align}
\eta(W^{(a)})&= \left(\prod_{t=1}^{T^{(a)}}\eta((U^{(a)}_t)^\dagger)\right)\eta(I_{L^2(\mathbb{R})}\otimes V^{(a)}) \left(\prod_{t=1}^{T^{(a)}}\eta(U^{(a)}_t)\right)
=1
\end{align}
since~$\eta(I_{L^2(\mathbb{R})}\otimes V)=1$ for any unitary acting trivially on the oscillator, and~$\eta(U)=1/\eta(U^\dagger)$ for every~$U\in\Uelemunbounded{1}{r}$.
Eq.~\eqref{eq:umvdavz} therefore implies that
\begin{align}\overline{g}(U)&\leq \left(\max_{a\in \{1,\ldots,L\}}\overline{g}(W^{(a)})\right)^2\ . \label{eq:dressedbound1}
\end{align}
It follows that 
\begin{align}
  \overline{g}(W^{(a)}) &= \overline{g}((U^{(a)})^\dagger \left(I_{L^2(\mathbb{R})}\otimes V^{(a)}\right) U^{(a)}) \\
                        & = \overline{g}((U^{(a)})^\dagger U^{(a)}) \, . \label{eq:dressedbound2}
\end{align}
The second identity follows directly from Definition~\ref{def:squeezingdisplacementcircuits} together with the fact that~$\eta(I_{L^2(\mathbb{R})}\otimes V) = 1$ for any unitary acting trivially on the oscillator.
We claim that 
\begin{align}
  \overline{g}((U^{(a)})^\dagger U^{(a)})&\le \overline{g}(U^{(a)})\qquad \textrm{for all} \qquad a \in \{1, \dots, L\}\, .  \label{eq:overlinegdressed}
\end{align}

\begin{proof}[Proof of Eq.~\eqref{eq:overlinegdressed}] Let~$a \in \{1, \dots, L\}$.
Write~$\left(U^{(a)}\right)^\dagger U^{(a)} = V_{2T(a)}^{(a)} \cdots V_1^{(a)}$ where we define~$V^{(a)}_{t} = U_t^{(a)}$ if~$t \in \{1, \dots, T(a)\}$ and~$V^{(a)}_t = (U^{(a)}_{2T(a) - t + 1})^\dagger$ if~$t\in \{T(a)+1, \dots, 2T(a)\}$.
In particular, we have~$V^{(a)}_t \in \Uelemunbounded{1}{r}$ for all~$t\in \{1, \dots, 2T(a)\}$.
Let~$t \in \{1, \dots, 2T(a)\}$ and~$p \in \{0,\dots, 2T(a)-t \}$ be such that 
\begin{align}
  \overline{g}((U^{(a)})^\dagger U^{(a)}) = g\left(\prod_{s=1}^t \eta(V_{p+s}^{(a)}) \right) \, . \label{eq:maximizingcircuitdressed}
\end{align}
Consider the corresponding subcircuit
\begin{align} 
  \prod_{s=1}^{t} V^{(a)}_{p+s}\, .
\end{align}
Define the reduced subcircuit of~$\prod_{s=1}^{t} V^{(a)}_{p+s}$ obtained by successive cancellation of all adjacent mutually inverse pairs of unitaries, i.e, subsequences of unitaries of the form~$(U^{(a)}_{q})^\dagger U^{(a)}_q$ for some~$q \in \{1, \dots, T(a)\}$. Then the reduced subcircuit is either the identity or we can write it as
\begin{align}
  \prod_{s=1}^{t'} V^{(a)}_{p'+s}
\end{align} where~$p' \ge T(a)$ or~$p' + t' <  T(a)$, that is,~$\prod_{s=1}^{t'} V^{(a)}_{p'+s}$ is either a subcircuit of~$U^{(a)}$ or a subcircuit of~$(U^{(a)})^\dagger$.
By definition we have~$\eta(U) = 1/\eta(U^\dagger)$ for all~$U \in \Uelemunbounded{1}{r}$. It follows that 
\begin{align}
  g\left( \prod_{s=1}^t \eta(V_{p+s}^{(a)})\right) &= g\left( \prod_{s=1}^{t'} \eta(V_{p'+s}^{(a)})\right) \\
  &\le \max\{\overline{g}(U^{(a)}),\overline{g}((U^{(a)})^\dagger)\} \\
  &=  \overline{g}(U^{(a)}) \ , \label{eq:boundWaUa}
\end{align}
where the inequality follows from Definition~\ref{def:squeezingdisplacementcircuits} and the last identity is implied by Lemma~\ref{lem:daggerrelationoverlineg}. The claim follows by combining with Eq.~\eqref{eq:maximizingcircuitdressed}. 
\end{proof}
Eqs.~\eqref{eq:dressedbound2} and~\eqref{eq:overlinegdressed} in combination with Eq.~\eqref{eq:dressedbound1} imply Claim~\eqref{eq:sqzingparm}.
\end{proof}

\section{Moment-limiting functions for the multimode case\label{sec:multi-modecasemomentlimit}}
We extend the concept of fine-grained moment-limiting functions 
to the setting of multiple oscillators and qubits as follows.

\begin{definition} \label{def:momentlimitmulti-mode}
  Let~$m \in \mathbb{N}$ and~$r \in \mathbb{N}_0$. 
  A  pair~
  \begin{align}
  (\underline{\Phi},\underline{\widehat{\Phi}})=\left((\Phi_1,\ldots,\Phi_m),(
  \widehat{\Phi}_1,\ldots,\widehat{\Phi}_m)\right)
  \end{align}
   of $m$-tuples of  entrywise affine-linear functions~$\Phi_\alpha,\widehat{\Phi}_\alpha:\mathbb{R}^2\rightarrow \mathbb{R}^2$,~$\alpha\in \{1,\ldots,m\}$ is called a fine-grained moment-limiting function for a unitary~$U\in \langle\Uelemunbounded{m}{r}\rangle$
  if
    \begin{align}
    \begin{aligned}
  U \Pi(\cJ_1,\ldots,\cJ_m)U^\dagger &\leq \Pi(\Phi_1(\cJ_1),\ldots,\Phi_m(\cJ_m))\\
    U \widehat{\Pi}(\cJ_1,\ldots,J_m)U^\dagger &\leq \widehat{\Pi}(\widehat{\Phi}_1(\cJ_1),\ldots,\widehat{\Phi}_m(cJ_m))
    \end{aligned}\qquad\textrm{ for all $m$-tuples of intervals } (\cJ_1,\ldots,\cJ_m)
  \end{align}
  where we write
  \begin{align}
  \Pi(\cJ_1,\ldots,\cJ_m)&:=\left(\Pi_{\cJ_m}\otimes\cdots\otimes \Pi_{\cJ_m}\right)\otimes I_{\mathbb{C}^2}^{\otimes r}\\
    \widehat{\Pi}(\cJ_1,\ldots,\cJ_m)&:=\big(\widehat{\Pi}_{\cJ_1}\otimes\cdots\otimes \widehat{\Pi}_{\cJ_m]}\big)\otimes I_{\mathbb{C}^2}^{\otimes r}\  , \end{align}
    and where for an entrywise affine-linear function $\Phi=(\Phi_1,\Phi_2):\mathbb{R}^2\rightarrow\mathbb{R}^2$ and an interval $\cJ=[R_1,R_2]$ we set 
    $\Phi(\cJ)=[\Phi_1(R_1),\Phi_2(R_2)]$.        
  \end{definition}

  We are interested in obtaining fine-grained moment-limiting functions for multimode circuits. 
  It will be convenient to omit single- and two-qubit unitaries from our considerations. They have no effect on moment-limits as expressed by the following lemma.
  \begin{lemma}[Removing qubit-only unitaries]\label{lem:trivialqubitactmva}
  Let~$m,r\in\mathbb{N}$. 
  Let 
    \begin{align}
  U=U_T\cdots U_1\qquad\textrm{ where }\qquad U_t\in \Uelemunbounded{m}{r}\qquad\textrm{ for every }\qquad t\in \{1,\ldots,T\}\ 
  \end{align}
  be a  circuit on~$m$ oscillators and~$r$~qubits.
  Let 
  \begin{align}
  \{t_1<\cdots <t_{J}\}=\{t\in \{1,\ldots,T\}\ |\ U_t\textrm{ acts non-trivially on an oscillator}\}
  \end{align}
  be the gate locations where a unitary is applied to some mode~$B_\alpha$,~$\alpha\in \{1,\ldots,m\}$.
  Let 
  \begin{align}
  V&=V_{J}\cdots V_1\qquad\textrm{ with }\qquad V_{j}:=U_{t_j}\qquad \textrm{ for }\qquad j\in \{1,\ldots,J\}
  \end{align}
  be the circuit obtained from~$U$ by removing all gates which  act on qubits only.
  Then the following holds:
  Suppose~$(\underline{\Phi},\underline{\widehat{\Phi}})$ is a fine-grained moment-limiting function for $V=V_J\cdots V_1$. 
  Then~$(\underline{\Phi},\underline{\widehat{\Phi}})$ is a fine-grained moment-limiting function for $U^{(T)}=U_T\cdots U_1$. 
     \end{lemma}
     \begin{proof}
     This follows immediately from the fact that a unitary of the form~$I_{L^2(\mathbb{R})^{\otimes m}}\otimes V$
with~$V$ an~$r$-qubit unitary leaves the operators~$\Pi(\cJ_1,\ldots,\cJ_m)$ and~$\widehat{\Pi}(\cJ_1,\ldots,\cJ_m)$ invariant under conjugation.
     \end{proof}
  
  In the following, we argue that a fine-grained moment-limiting function~$(\underline{\Phi},\underline{\widehat{\Phi}})$ can be obtained for any circuit~$U=U_T\cdots U_1\in \langle\Uelemunbounded{m}{r}\rangle$ by considering~$m$~different derived circuits~$U|_{B_\alpha}\in \langle \Uelemunbounded{1}{r}\rangle$ acting on a single oscillator~$B_\alpha\cong L^2(\mathbb{R})$,~$\alpha\in \{1,\ldots,m\}$ only. The following definition will be useful.
  \begin{definition}[Single-mode restricted derived circuit.]\label{def:singlemoderestrictcircuit}
  Let~$m\geq 2$ and~$r\in \mathbb{N}$. Consider a circuit 
  \begin{align}
  U=U_T\cdots U_1\qquad\textrm{ where }\qquad U_t\in \Uelemunbounded{m}{r}\qquad\textrm{ for every }\qquad t\in \{1,\ldots,T\}\ 
  \end{align}
  on~$m$ oscillators and~$r$~qubits, denoted~$B_1\cdots B_mQ_1\cdots Q_r\cong L^2(\mathbb{R})^{\otimes m}\otimes (\mathbb{C}^2)^{\otimes r}$.    For every~$\alpha \in \{1,\ldots,m\}$, define the derived circuit~$U|_{B_\alpha}$ restricted to mode~$B_\alpha$ as follows.
  Let 
  \begin{align}
  \{t_1<\ldots< t_{T_\alpha}\}:=\{t\in \{1,\ldots,T\}\ |\ U_t\textrm{ acts non-trivially only on mode~$B_\alpha$}\}
  \end{align}
  by the circuit locations in~$U$ where a (possibly qubit-controlled) unitary is applied to mode~$B_\alpha$.
  Then define
  \begin{align}
  U|_{B_\alpha}&:=U_{t_{T_\alpha}} \cdots U_{t_1}\ .
  \end{align}
      \end{definition}
    \noindent We note that the collection~$\{U|_{B_\alpha}\}_{\alpha=1}^m$ of single-mode restricted circuits does not depend on the single- and two-qubit unitaries (acting trivially on the oscillators) in the circuit~$U$. These unitaries have no effect on moment limits and can be omitted, see Lemma~\ref{lem:trivialqubitactmva}.
          
  The significance of Definition~\ref{def:singlemoderestrictcircuit} is clarified by the following lemma. In this statement, we use that the single-mode restricted derived circuit~$ U|_{B_\alpha}$ can be seen as an element of~$\langle \Uelemunbounded{1}{r}\rangle$ since it only acts on the mode~$B_\alpha$.
  \begin{lemma}[Multimode to single-mode reduction]\label{lem:multi-modesinglemodereduction}
  Let~$m,r\in\mathbb{N}$. Consider a circuit 
  \begin{align}
  U=U_T\cdots U_1\qquad\textrm{ where }\qquad U_t\in \Uelemunbounded{m}{r}\qquad\textrm{ for every }\qquad t\in \{1,\ldots,T\}\ .
  \end{align}    
  For every~$\alpha\in \{1,\ldots,m\}$, let~$(\Phi_\alpha,\widehat{\Phi}_\alpha):\mathbb{R}^2\rightarrow\mathbb{R}^2$ be a pair of entrywise affine-linear functions which is a fine-grained moment-limiting function for~$U|_{B_\alpha}$. 
        Then~$\left((\Phi_1,\ldots,\Phi_m),(\widehat{\Phi}_1,\ldots,\widehat{\Phi}_m)\right)$ is a fine-grained moment-limiting function 
    for~$U$. 
    \end{lemma}
    \begin{proof}
    By Lemma~\ref{lem:trivialqubitactmva}, we can assume without loss of generality that the set~$\{U_T,\ldots,U_1\}$ does not contain unitaries acting on qubits only. In other words, 
    every unitary is a (possibly qubit-controlled) displacement or single-mode squeezing operation.
    It is easy to check that~$U$ can be written as
    \begin{align}
U&=U|_{B_m}\cdots  U|_{B_1}\ .\label{eq:bmdeqva}
\end{align}
    In Eq.~\eqref{eq:bmdeqva} we made use of the fact that unitaries acting on different modes commute, and the same is true for qubit-controlled unitaries acting on different modes.
    
       Since for every~$\alpha\in \{1,\ldots,m\}$ the pair of functions~$(\Phi_\alpha,\widehat{\Phi}_\alpha)$ is moment-limiting for~$U|_{B_\alpha}$
       by assumption, we have the operator inequalities
       \begin{align}U_{B_\alpha}\Pi(\cJ_1,\ldots,\cJ_{\alpha-1},\cJ_\alpha,\cJ_{\alpha+1},\ldots,\cJ_m)       U_{B_\alpha}^\dagger &\leq \Pi(\cJ_1,\ldots,\cJ_{\alpha-1},\Phi_\alpha(\cJ_\alpha),\cJ_{\alpha+1},\ldots,\cJ_m)\ \label{eq:opineqabasdv}
       \end{align}
       for all $m$-tuples of intervals $(\cJ_1,\ldots,\cJ_m)$ and~$\alpha\in \{1,\ldots,m\}$. 
       Here we used that~$U_{B_\alpha}$ only acts on mode~$B_\alpha$. 
       The claim now follows inductively from Eq.~\eqref{eq:bmdeqva} and Eq.~\eqref{eq:opineqabasdv}.
    \end{proof}
    Our goal is to bound the amount of energy produced in the execution of a circuit.      For convenience, let us introduce the following quantities.
\begin{definition}[Squeezing and displacement parameters of multimode circuits]\label{def:squeezingdisplacementcircuitsmulti-mode}
Consider a product~$U=U_T \cdots U_1$ acting on~$B_1\cdots B_mQ_1\cdots Q_r\cong L^2(\mathbb{R})^{\otimes m}\otimes (\mathbb{C}^2)^{\otimes r}$with~$U_t\in\Uelemunbounded{m}{r}$ for~$t\in\{1,\ldots,T\}$.
Let~$g(x):=\max \{x,1/x\}$ for~$x\in\mathbb{R}\backslash \{0\}$.
For every~$\alpha\in \{1,\ldots,m\}$, define the functions
  \begin{align}
  \squeezingparam^{(\alpha)}(U)&:=\max_{t\in \{1,\ldots,T\}} \max_{p\in \{0,\ldots,T-t\}}g\left(\prod_{s=1}^t \eta^{(\alpha)}(U_{p+s})\right)\\\
  \overline{\xi}^{(\alpha)}(U)&:=  \sum_{j=1}^T \xi^{(\alpha)}(U_j)\ ,
  \end{align}
  where 
  \begin{align}
  \eta^{(\alpha)}(V):=\begin{cases}
  \eta(V)\qquad&\textrm{ if~$V$ acts non-trivially on the mode~$B_\alpha$}\\
  1 &\textrm{ otherwise }
  \end{cases}\\
  \xi^{(\alpha)}(V):=\begin{cases}
  \xi(V)\qquad&\textrm{ if~$V$ acts non-trivially on the mode~$B_\alpha$}\\
  0 &\textrm{ otherwise }
  \end{cases} 
  \end{align}
  and where~$(\eta(V),\xi(V))\in (0,\infty)\times [0,\infty)$ are the squeezing and displacement parameters introduced in Definition~\ref{def:squeezingdisplacementparameters} for every generator~$V\in\Uelemunbounded{m}{r}$. We then set
  \begin{align}
\squeezingparam(U)&:=\max_{\alpha\in \{1,\ldots,m\}}\squeezingparam^{(\alpha)}(U)\\
\overline{\xi}(U)&:=\max_{\alpha\in \{1,\ldots,m\}}\overline{\xi}^{(\alpha)}(U)
  \end{align}
 and  call~$(\squeezingparam(U),\overline{\xi}(U))$ squeezing and displacement parameters of the circuit~$U$. 
  \end{definition}

\begin{lemma}[Fine-grained moment-limiting function for (partially implemented) circuits and energy: multimode case]  \label{lem:PhiUpartiallyimmplementedmulti}
 Let~$m,r\in\mathbb{N}$.   Let~$U=U_T \cdots U_1$ with~$U_t\in\langle \Uelemunbounded{m}{r}\rangle$ for~$t\in\{1,\ldots,T\}$ be given. Define 
   \begin{align}
   U^{(t)}&=\begin{cases}
   I\qquad&\textrm{ for } t=0\\
      U_t\cdots U_1&\textrm{ otherwise}\ 
      \end{cases}
      \end{align}
and      \begin{align}
      \ket{\Psi^{(t)}}&:=(U^{(t)})^\dagger(\ket{\mathsf{vac}}^{\otimes m}\otimes\ket{0}^{\otimes r})\ 
      \end{align}
      for each~$t\in \{0,\ldots,T\}$.      
      let $\squeezingparam(U)$, $\overline{\xi}(U)$ 
      be the squeezing and displacement parameters of the circuit~$U$ introduced in Definition~\ref{def:squeezingdisplacementcircuitsmulti-mode}.
    Then
    \begin{align}
   \langle \Psi^{(t)},(Q_\alpha^2+P_\alpha^2)\Psi^{(t)}\rangle &\leq 168 \squeezingparam(U)^6\cdot \left(2+\overline{\xi}(U)^3\right)
   \end{align} for each  $t\in \{0,\ldots,T\}$ and $\alpha\in \{1,\ldots,m\}$.   
 \end{lemma}
 \begin{proof}
   For $\alpha\in \{1,\ldots,m\}$, let $(\squeezingparam^{(\alpha)}(U),\overline{\xi}^{(\alpha)}(U))$ be the squeezing and 
   displacement parameters from Definition~\ref{def:squeezingdisplacementcircuitsmulti-mode}.       Let~$\alpha\in \{1,\ldots,m\}$ and~$t\in \{0,\ldots,T\}$ be arbitrary. It is easy to check that the definition of~$\overline{g}^{(\alpha)}(U^{(t)})$ implies that 
        \begin{align}
        \begin{matrix}
        \overline{g}^{(\alpha)}(U^{(t)})&=&
        \overline{g}(U^{(t)}|_{B_{\alpha}})\\
        \overline{\xi}^{(\alpha)}(U^{(t)})
        &=&\overline{\xi}(U^{(t)}|_{B_\alpha})
        \end{matrix}\label{eq:galphautavad}
        \end{align}
        are equal to the squeezing and displacement parameters of the single-mode restricted derived circuit~$U^{(t)}|_{B_\alpha}$.     Using Eq.~\eqref{eq:galphautavad}
       and Lemma~\ref{lem:finegrainedenergy} we have
       \begin{align}
       \langle \Psi^{(t)}_\alpha,(Q^2+P^2) \Psi^{(t)}_\alpha\rangle &\leq  168\overline{g}^{(\alpha)}(U^{(t)})^6
       (2+\overline{\xi}^{(\alpha)}(U^{(t)})^3)\ .
       \end{align}
      for the state
      \begin{align}
      \Psi^{(t)}_\alpha&:=U^{(t)}|_{B_\alpha}(\ket{\mathsf{vac}}\otimes\ket{0}^{\otimes r})\in L^2(\mathbb{R})\otimes (\mathbb{C}^2)^{\otimes r}\ .
      \end{align}
      But it is easy to check from the structure of the circuit that 
      \begin{align}
             \langle \Psi^{(t)},(Q_\alpha^2+P_\alpha^2) \Psi^{(t)}\rangle &=
       \langle \Psi^{(t)}_\alpha,(Q^2+P^2) \Psi^{(t)}_\alpha\rangle\ . 
       \end{align}
      Hence 
      \begin{align}
             \langle \Psi^{(t)},(Q_\alpha^2+P_\alpha^2) \Psi^{(t)}\rangle &\leq 168\overline{g}^{(\alpha)}(U^{(t)})^6
       (2+\overline{\xi}^{(\alpha)}(U^{(t)})^3)\\
       &\leq 168\overline{g}^{(\alpha)}(U)^6
       (2+\overline{\xi}^{(\alpha)}(U)^3)
       \end{align} 
       where we used that
       \begin{align}
       \overline{g}^{(\alpha)}(U^{(t)})&\leq \overline{g}^{(\alpha)}(U)\\
              \overline{\xi}^{(\alpha)}(U^{(t)})&\leq \overline{\xi}^{(\alpha)}(U)       \end{align} 
        for every $t\in \{0,\ldots,T\}$ by definition.

  \end{proof}
        
 \subsection{Moment-limits of circuits obtained from bounded-strength substitutions\label{sec:boundedstrengthsubstitutions}}
In this section, we show moment-limits on  circuits obtained by 
bounded-strength substitutions. To define the latter, 
recall (see Eq.~\eqref{eq:mygroupdefinitionmultiqubitbounded})
that~$\Uelem{m}{n}(\alpha,\zeta)$ denotes the set of elementary unitary operations
with squeezing and displacement bounded by~$\alpha > 1$ and~$\zeta\geq 1$, respectively.

 We note that even for~$\alpha=2$ (i.e., constant-strength squeezing operations only), an element~$V\in\Uelem{1}{1}(2,\zeta)$ may not be constant-strength if~$\zeta$ is non-constant (e.g., grows  with the problem size).  However, the following substitution rule allows us to replace every such unitary~$V$ by a product~$V=V^{(N)}\cdots V^{(1)}$ of bounded-strength unitaries~$V^{(s)}\in \Uelem{1}{1}(2,1)$ for~$s\in \{1,\ldots,N\}$. Furthermore, the number~$N$ of such unitaries is of order~$O(\log |\zeta|)$. 

\begin{definition}[Bounded-strength substitution]\label{def:defconstantstrengthsubs}
We define a  procedure  called   bounded-strength (displacement) substitution, which
 takes as input a circuit~
\begin{align}
U=U_T\cdots U_1\qquad\textrm{ where }\qquad U_t\in  \Uelemunbounded{m}{r}\qquad\textrm{  for each }\qquad t\in \{1,\ldots,T\}
\end{align}
 and produces a circuit
 \begin{align}
 V_U&=V_{S}\cdots V_1\qquad\textrm{ where }\qquad V_s\in \Uelem{m}{r}(2,1)\qquad\textrm{ for each }\qquad s\in \{1,\ldots,S\}\label{eq:vnavdwelldefn}
 \end{align}
 such that~$S\geq T$ and~$V_U=U$ (i.e., the unitaries defined by these circuits have the same action). It proceeds as follows: 
\begin{enumerate}[(i)]
\item
 Every displacement~$e^{i\theta Q}$ with~$|\theta|>1$ in~$\{U_T,\ldots,U_1\}$ is decomposed (i.e., replaced by a product of~$2n+1$ bounded-strength unitaries) as
\begin{align}
e^{i \theta Q}&=\left(M_{\beta}^{\dagger}\right)^n e^{i \mathsf{sgn}(\theta)Q}\left(M_{\beta}\right)^n 
\end{align}
where 
\begin{align}
\beta=2^{\frac{\log_2 |\theta|}{\lceil\log_2 |\theta|\rceil}}\qquad\textrm{ and }\qquad  n=\lceil\log_2 |\theta|\rceil\ .\label{eq:betandefinition}
\end{align}
\item Similarly, we decompose every displacement~$e^{-i\theta P}$ with~$|\theta|>1$  in~$\{U_T,\ldots,U_1\}$  as 
\begin{align}
e^{-i \theta P}&=\left(M_{\beta}\right)^n e^{-i \mathsf{sgn}(\theta)P}\left(M_{\beta}^\dagger\right)^n\ . \label{eq:displacementdecompositionrulethetaP}
\end{align}
\item
Finally, every controlled displacement~$\ctrl e^{i\theta Q}$ and~$\ctrl_a e^{-i\theta P}$ with~$|\theta|>1$  in~$\{U_T,\ldots,U_1\}$  is decomposed as 
\begin{align}
\ctrl e^{i \theta Q}&=\left(M_{\beta}^{\dagger}\right)^n \ctrl e^{i \mathsf{sgn}(\theta)Q}\left(M_{\beta}\right)^n\\
\ctrl e^{-i \theta P}&=\left(M_{\beta}\right)^n \ctrl e^{-i \mathsf{sgn}(\theta)P}\left(M_{\beta}^\dagger\right)^n\ 
\end{align}
with~$(\beta,n)$ as defined in Eq.~\eqref{eq:betandefinition}.
\end{enumerate}
All other unitaries~$U_t$,~$t\in \{1,\ldots,T\}$ are kept. This completes the construction of the circuit~$V_U$.
\end{definition}
To see that this  is well-defined, observe that the circuit~$V_U$ constructed in this way clearly has the same action as~$U$.   We note that for any~$|\theta|>1$ we have~$\log_2 |\theta|/\lceil \log_2 |\theta|\rceil\in (0,1)$ and thus
\begin{align}
1<\beta<2\ 
\end{align}
by definition. In particular, each unitary appearing as a factor in these decompositions is an element of~$\Uelem{m}{r}(2,1)$ as claimed in Eq.~\eqref{eq:vnavdwelldefn}.

\begin{lemma}[Squeezing and displacement parameters of a circuit obtained from the substitution rule] \label{lem:substitutionsinglemodegen}
Let~$\zeta\geq 2$ be given. Let~$U=U_T\cdots U_1$ be a circuit composed of unitaries~$U_t\in \Uelem^{m,r}(2,\zeta)$ for every~$t\in \{1,\ldots,L\}$. 
For each~$\alpha\in \{1,\ldots,m\}$, 
let 
\begin{align}
\subsset^{(\alpha)}:=
\left\{t\in \{1,\ldots,L\}\ |\ 
U_t\in \left\{
e^{-i\theta P_\alpha},e^{i\theta Q_\alpha},\mathsf{ctrl}_ae^{-i\theta P_\alpha},
\mathsf{ctrl}_ae^{i\theta Q_\alpha}\ |\ a\in \{1,\ldots,r\}, |\theta|>1
\right\}
\right\}
\end{align}
be the circuit locations where a (possibly controlled) displacement of strength~$|\theta|>1$ is applied to mode~$B_\alpha$. 
 Let  
\begin{align}
\subsset=\bigcup_{\alpha=1}^m \subsset^{(\alpha)}\subset \{1,\ldots,T\}
\end{align}
be the list of indices such that for each~$t\in \subsset$, the unitary~$U_t$ is  a (possibly controlled) displacement
with strength~$\theta$ satisfying~$|\theta|\geq 1$. 
Let~$V_U=V_{S}\cdots V_1$, where~$V_s\in \Uelem{m}{r}(2,1)$ for each~$s\in \{1,\ldots,S\}$ be the circuit obtained by 
applying bounded-strength substitution to each~$U_t$,~$t\in \subsset$. 
Then the following holds for any~$\alpha\in \{1,\ldots,m\}$: We have
\begin{align}
\begin{aligned}
\overline{\xi}^{(\alpha)}(V_U)&=|\subsset^{(\alpha)}| +\sum_{t\in \subsset^c} \xi^{(\alpha)}(U_t) \ , \\
\squeezingparam^{(\alpha)}(V_U)&\leq \zeta^2  \prod_{a \in \subsset^c} g(\eta^{(\alpha)}(U_t))\ ,
\end{aligned}\label{eq:firstclaimupperboundvam}
\end{align}
where~$\subsset^c:=\{1,\ldots,T\}\backslash \subsset$. In particular,
\begin{align}
\begin{aligned}
\overline{\xi}(V_U)&\leq \max_{\alpha\in \{1,\ldots,m\}}T^{(\alpha)} \ , \\
\squeezingparam(V_U)&\leq \zeta^2 \max_{\alpha\in \{1,\ldots,m\}}2^{T^{(\alpha)}- |\subsset^{(\alpha)}|}\ ,
\end{aligned}\label{eq:secondclaimupperboundvam}
\end{align}
where~$T^{(\alpha)}\in \{0,\ldots,T\}$ is the number of  unitaries in~$\{U_t\}_{t=1}^T$ acting non-trivially on the mode~$B_\alpha$. 
\end{lemma}
\begin{proof}
Let~$\alpha\in \{1,\ldots,m\}$ be fixed. Similar to Eq.~\eqref{eq:galphautavad} we use that
\begin{align}
\begin{aligned}
\overline{\xi}^{(\alpha)}(V_U)&=\overline{\xi}^{(\alpha)}(V_U|_{B_\alpha}) \ , \\
\overline{\eta}^{(\alpha)}(V_U)&=\overline{\eta}^{(\alpha)}(V_U|_{B_\alpha})
\end{aligned}\label{eq:vdavd}
\end{align}
are equal to the squeezing and displacement parameters of the single-mode restricted derived circuit~$V_U|_{B_\alpha}$.

Now consider the circuit~$V_U|_{B_\alpha}$.
It is the result of applying bounded-strength substitution to the circuit~$U|_{B_\alpha}$, i.e., it suffices to consider the unitaries~$U_t$ acting non-trivially on the mode~$B_\alpha$. 
Define
\begin{align}
\cT^{(\alpha)}&:=\{t\in \{1,\ldots,T\}\ |\ U_t\textrm{ acts non-trivially on mode }B_\alpha\}\ .
\end{align}
Consider  a unitary~$U_t$ with~$t\in \cT^{(\alpha)}$, i.e., acting non-trivially on~$B_\alpha$. If~$t\in \subsset^{(\alpha)}$, we can consider~$U_t$ as a subcircuit (with
gate decomposition as prescribed by the bounded-strength substitution, see Definition~\ref{def:defconstantstrengthsubs}). If~$t\not\in \subsset^{(\alpha)}$  we consider~$U_t$ as a subcircuit (of size~$1$) in its own right. 

In more detail, consider a unitary~$U_t$ with~$t\in \subsset^{(\alpha)}$. Assume for simplicity that
$U_t=e^{-i\theta P_\alpha}$ with~$|\theta|>1$ (The other cases are treated similarly.)
According to Eq.~\eqref{eq:displacementdecompositionrulethetaP} we then have 
\begin{align}
U_t&=V_{U_t} = \left(M_{\beta}\right)^n e^{-i \mathsf{sgn}(\theta)P}\left(M_{\beta}^\dagger\right)^n \ 
\end{align}
acting on mode~$\alpha$, where $V_{U_t}$ is the result of applying the bounded-strength substitution to $U_t$. It is easy to check from this expression that
\begin{align}
\eta(V_{U_t}) =1\ \qquad\textrm{ and }\qquad 
\begin{aligned}
\overline{\xi}(V_{U_t})&= 1\\
\squeezingparam(V_{U_t})&=\beta^n \ .
\end{aligned}\label{eq:xietasqm}
\end{align}
Because~$\beta^n=|\theta|\leq \zeta$ by the assumption that~$U_t\in
\Uelem{m}{r} (2,\zeta)$,  we conclude that
\begin{align}
\begin{aligned}
\squeezingparam(V_{U_t})&\leq\zeta\\
\overline{\xi}(V_{U_t})&\leq 1\\ 
g(\eta(V_{U_t}))&=1 
\end{aligned}\qquad\textrm{ for every }\qquad t\in \subsset^{(\alpha)}\ .\label{eq:sqvdm}
\end{align}
On the other hand, if~$t\in \cT^{(\alpha)}\backslash \subsset^{(\alpha)}$, then~$U_t$ is  either a (possibly qubit-controlled) displacement on mode~$B_\alpha$ of strength~$|\theta|\leq 1$, or a single-mode squeezing operator on mode~$B_\alpha$ of strength~$\alpha\in (1/2,2)$ because of the assumption that~$U_t\in \Uelem{m}{r}(2,\zeta)$.  
In particular, it follows that~$\eta(U_t)\in (1/2,2)$ and thus
\begin{align}
\begin{aligned}
\overline{\xi}(U_t)&=\xi(U_t)\leq 1\\
\overline{g}(U_t)&=g(\eta(U_t))\leq 2
\end{aligned}
\qquad\textrm{ for every }\qquad t\in  \cT^{(\alpha)}\backslash \subsset^{(\alpha)}\ .\label{eq:getautavd}
\end{align}
With 
Lemma~\ref{lem:squeezingdisplacementsubcircuits}, and Eqs.~\eqref{eq:sqvdm},~\eqref{eq:getautavd} we obtain 
\begin{align}
\overline{\xi}(V_U|_{B_\alpha})&=\sum_{
t\in  \subsset^{(\alpha)}
}\overline{\xi}(V_{U_t})+\sum_{
t\in \cT^{(\alpha)}\backslash \subsset^{(\alpha)}
}\overline{\xi}(U_t)\\
&\le \left|\subsset^{(\alpha)}\right|+\sum_{
  t\in \cT^{(\alpha)}\backslash \subsset^{(\alpha)}
  } \xi (U_t) \label{eq:auxb17} \\
&\le \left|\subsset^{(\alpha)}\right|+\sum_{
  t\in \subsset^{c}
  }\xi^{(\alpha)}(U_t)
\end{align}
and
\begin{align}
\squeezingparam(V_U|_{B_\alpha})&\leq \left(\max\left\{\max_{t\in\subsset^{(\alpha)}}\squeezingparam(V_{U_t}),
\max_{t\in \cT^{(\alpha)}\backslash \subsset^{(\alpha)}}\squeezingparam(U_t)\right\}
\right)^2\cdot \left(\prod_{t\in \subsset^{(\alpha)}}g(\eta(V_{U_t}))\right) \cdot\left(\prod_{t\in \cT^{(\alpha)}\backslash \subsset^{(\alpha)}}g(\eta(U_t))
\right)  \\
&\leq \left(\max\{\zeta,2\}\right)^2\cdot \prod_{t\in \cT^{(\alpha)}\backslash \subsset^{(\alpha)}}g(\eta(U_t)) \label{eq:auxb17_1}
\\
&\leq \zeta^2\cdot \prod_{t\in\subsset^{c}} g(\eta^{(\alpha)}(U_t))
 \ .
\end{align}
Here we used the assumption that~$\zeta\geq 2$. 
Claim~\eqref{eq:firstclaimupperboundvam} follows from this because of  Eq.~\eqref{eq:vdavd}.

To show Claim~\eqref{eq:secondclaimupperboundvam}, observe that
\begin{align}
\overline{\xi}^{(\alpha)}(V_U)&\leq |\subsset^{(\alpha)}|+|\cT^{(\alpha)}\backslash \subsset^{(\alpha)}|\\
&=|\cT^{(\alpha)}| \label{eq:auxb17_2} 
\end{align}
because of Eqs.~\eqref{eq:getautavd} and~\eqref{eq:auxb17}. 
Furthermore, by Eqs.~\eqref{eq:auxb17_1} and~\eqref{eq:getautavd} we have 
\begin{align}
\squeezingparam^{(\alpha)}(V_U)&\leq \zeta^2 \cdot 2^{|\cT^{(\alpha)}\backslash \subsset^{(\alpha)}|}\ . \label{eq:auxb17_3} 
\end{align}
Claim \eqref{eq:secondclaimupperboundvam} is a direct consequence of Eqs.~\eqref{eq:auxb17_2} and~\eqref{eq:auxb17_3}.
\end{proof}

\section{Comb states, their preparation and approximate GKP codes~\label{sec:appendixshastates}}
In this section we give detailed statements about the approximate GKP codes used in the main text. 
In Section~\ref{sec:rectGKP} we define rectangular envelope GKP states and approximate GKP codes based on these states. In Section~\ref{sec:GKPprep} we derive results about how costly it is to (approximately) prepare these states in the hybrid qubit-oscillator model. 
We conclude this section by proving Theorem~\ref{thm:initialstateprep} in the main text. 
\subsection{Definition of rectangular-envelope approximate GKP states} \label{sec:rectGKP}
Central to our construction is the use of certain approximate GKP codes. These are most easily introduced using 
the compactly supported, integer-spaced comb state (or ``rectangular-envelope GKP state''). For an even integer~$L\in 2\mathbb{N}$, a squeezing parameter~$\Delta>0$ and a truncation parameter~$\varepsilon \in (0,1/2)$, the latter is  defined as 
\begin{align}
\left|\Sha_{L, \Delta}^{\varepsilon}\right\rangle=\frac{1}{\sqrt{L}} \sum_{z=-L / 2}^{L / 2-1}\left|\chi_{\Delta}^{\varepsilon}(z)\right\rangle\ .\label{eq:integerspacedcomb}
\end{align}
Here~$\chi_{\Delta}^\varepsilon(z)(\cdot)=\Psi_\Delta^\varepsilon(\cdot-z)$
is a translated truncated Gaussian obtained from the Gaussian
$\Psi_{\Delta}(x)=\frac{1}{\left(\pi \Delta^2\right)^{1 / 4}} e^{-x^2 /\left(2 \Delta^2\right)}$
by setting
\begin{align}
\Psi_{\Delta}^\varepsilon(x)&\propto \begin{cases}
\Psi_\Delta(x)\qquad &\textrm{ if }\qquad |x|\leq \varepsilon\\
0 &\textrm{ otherwise }
\end{cases}
\end{align}
and normalizing such that~$\|\Psi_\Delta^\varepsilon\|=1$. 

Let~$d\geq 2$ be an integer and~$\varepsilon\in (0,1/(2d))$. For~$j\in\mathbb{Z}_d$ we define the normalized state 
\begin{align}
\ket{\Sha_{L,\Delta}^\varepsilon(j)_d}&= e^{-i\sqrt{2\pi/d}jP} M_{\sqrt{2\pi d}}\ket{\Sha_{L,\Delta}^\varepsilon}\  \label{eq:symmetricallysqueezedstates}
\end{align}
using the single-mode squeezing operator~$M_{\sqrt{2\pi d}}$. It can be checked easily that for~$\varepsilon\leq 1/(2d)$, the states~$\{\Sha_{L,\Delta}^\varepsilon(j)_d\}_{j=\{0,\ldots,d-1\}}$ form an orthonormal family.

The associated (rectangular-envelope  truncated) GKP code~$\gkpcoderect{L,\Delta}{\varepsilon}{d}$
with parameters~$(L,\Delta,\varepsilon)$ is defined as the span
\begin{align}
\gkpcoderect{L,\Delta}{\varepsilon}{d}:=\mathsf{span}\{\Sha_{L,\Delta}^\varepsilon(j)_d\}_{j=\{0,\ldots,d-1\}}\ \label{eq:gkpcodedefinition}
\end{align}
of these vectors.  
We use the map~$\ket{j}\mapsto \ket{\Sha_{L,\Delta}^\varepsilon(j)_d}$  on the computational basis to isometrically embed~$\mathbb{C}^d$ into the~$d$-dimensional space~$\gkpcoderect{L,\Delta}{\varepsilon}{d}\subset L^2(\mathbb{R})$.

The following parameter choices will be particularly natural and convenient. First, we choose the truncation parameter as~$\varepsilon_d=1/(2d)$ and  write
\begin{align}
\gkpcoderect{L,\Delta}{\star}{d}:=\gkpcoderect{L,\Delta}{1/(2d)}{d}\  .
\end{align}
Second, we typically choose the integer~$L\in\mathbb{N}$ as a certain function of~$(\Delta,d)$, i.e., we set
\begin{align}
L_{\Delta,d}&=2^{2\left(\lceil \log_2 1/\Delta\rceil - \lfloor\log_2 d\rfloor\right)}\ .\label{eq:Lchoicegood}
\end{align}
With these choices, we end up with a one-parameter family of approximate GKP codes depending only on the parameter~$\Delta>0$. We write
\begin{align}
\gkpcoderect{\Delta}{\star}{d}:=\gkpcoderect{L_{\Delta,d},\Delta}{\varepsilon_d}{d}\  
\end{align}
for the code associated with~$\Delta>0$, and call this the approximate (rectangular-envelope truncated) GKP code with parameter~$\Delta$. Its basis elements will be denoted as
\begin{align}
  \label{eq:shastar}
\ket{\Sha^\star_{\Delta}(j)_d}&=\ket{\Sha^{\varepsilon_d}_{L_{\Delta,d},\Delta}(j)_d}\qquad\textrm{ for }\qquad j\in\mathbb{Z}_d\ .
\end{align}
We note that the parameter~$L_{\Delta,d}$ in Eq.~\eqref{eq:Lchoicegood} is always a power of two. We write~$L_{\Delta,d}=2^{n_{\Delta,d}}$ where 
\begin{align}
n_{\Delta,d}:=2\left(\lceil \log_2 1/\Delta\rceil - \lfloor\log_2 d\rfloor\right)\ .
\end{align}
Moreover, we define the auxiliary state 
\begin{align}
  \ket{\Sha_{\Delta,\ell}^{\mathsf{aux}}(0)_2} = \ket{\Sha_{\Delta, L_{\Delta, 2^\ell}}^{2^{-(\ell+1)}}(0)_{2}}\, .
\end{align}
Hence~$\ket{\Sha_{\Delta,\ell}^{\mathsf{aux}}(0)_2}$ is a code state of a two-dimensional approximate GKP code whose parameters are derived from a~$2^\ell$-dimensional approximate GKP code.

\subsection{Preparation of rectangular-envelope approximate GKP states} \label{sec:GKPprep}
In the following, we argue that there is an efficient circuit
preparing  multiple copies of the state~$\ket{\Sha^\star_{\Delta}(0)_{2^\ell}}$ as well as an instance of the auxiliary state~$\ket{\Sha_{\Delta,\ell}^{\mathsf{aux}}(0)_2}$ in the qubit-oscillator model, see 
Theorem~\ref{thm:initialstateprep} in the main text.

We borrow the following protocol 
from Ref.~\cite{brenner2024complexity}: It prepares a (normalized) state of the form
\begin{align}
\ket{\Sha_{2^n,\Delta}}&\propto \sum_{z=-2^{n-1}}^{2^{n-1}-1}\left|\chi_{\Delta}(z)\right\rangle\ 
\end{align}
similar to the  
integer-spaced GKP state defined by~Eq.~\eqref{eq:integerspacedcomb} with $\chi_{\Delta}(z)(\cdot) = \Psi_\Delta(\cdot - z)$ and~$L=2^n$ using one auxiliary qubit.

\begin{theorem}(\cite[Theorem 3.1]{brenner2024complexity}, paraphrased) \label{thm:preparationproceduregmvd} Let~$\Delta \in(0,1 / 4)$
and~$n\in\mathbb{N}$ be given.
Define
\begin{align}
z_\Delta:= \frac{\log_2 1/\Delta}{\lceil |\log_2 1/\Delta| \rceil}\ 
\end{align}
and the unitary
\begin{align}
U_{2^n,\Delta}&=HV^{n-1}e^{iP} VH  M_{1/2}^n M_{2^{-{z_\Delta}}}^{\lceil \log 1/\Delta\rceil}\ \label{eq:UunitaryVH}
\end{align}
where
\begin{align}
V&=(\ctrl e^{i\pi Q})H(\ctrl e^{-iP}) M_2\ \label{eq:vdefinitionprepcircuit}
\end{align}
on~$L^2(\mathbb{R})\otimes \mathbb{C}^2$. (Here we omit identities on the qubit and oscillator, respectively.) Consider the output
state \begin{align}
\ket{\Phi_{2^n,\Delta}}&:=U_{2^n,\Delta}(\ket{\vac}\otimes\ket{0})\ .
\end{align}
Then 
\begin{align}
\left\| \proj{\Phi_{2^n,\Delta}}-\left|\Sha_{2^n, \Delta}\right\rangle\left\langle \Sha_{2^n, \Delta}\right|\otimes \proj{0} \right\|_1 \leq 17 \sqrt{\Delta}\ .\label{eq:shatwndeltadefreached}
\end{align}
The circuit~$U_{2^n,\Delta}$ is a product~$U_{2^n,\Delta}=U_{\size(U_{2^n,\Delta})}\cdots U_1$ of
\begin{align}
\size(U_{2^n,\Delta}) &= 5n + \lceil \log_2  1 / \Delta\rceil + 3\ ,\label{eq:tprimebx}
\end{align}
elementary unitary operations~$U_1,\ldots,U_{\size(U_{2^n,\Delta})}\in \Uelemunbounded{1}{1}$.
Furthermore, we have 
\begin{align}
\overline{\xi}(U_{2^n,\Delta})&=n(\pi+1)+1\\
\squeezingparam(U_{2^n,\Delta})&=2^n/\Delta\ .
\end{align}
\end{theorem}
\begin{proof}
See the proof of \cite[Theorem 3.1]{brenner2024complexity} for details including, in particular, the proof of Eq.~\eqref{eq:shatwndeltadefreached}.
We note that, in contrast to the statement given in \cite[Theorem 3.1]{brenner2024complexity} which focuses on the reduced density operator~$\tr_{\mathsf{qubit}}\proj{\Phi_{2^n,\Delta}}$, we include
the qubit in Eq.~\eqref{eq:shatwndeltadefreached} (and, correspondingly, include an additional Hadamard gate in the definition of~$U_{2^n,\Delta}$. We note that the proof of \cite[Theorem 3.1]{brenner2024complexity} actually establishes this stronger inequality (and is obtained by then using the monotonicity of the trace norm under partial traces).

We note that by definition, we have~$|z_\Delta|\leq 1$. 
This shows that each factor in Eq.~\eqref{eq:vdefinitionprepcircuit} as  well as Eq.~\eqref{eq:UunitaryVH} is bounded strength, i.e., belongs to the set~$\Uelemunbounded{1}{1}$. This implies  Eq.~\eqref{eq:tprimebx}.

Because~$2^{-{z_\Delta}}\leq 1$
and~$V$ only contains~$M_2$
 it is easy to check
that the main contribution to squeezing is from the term~$M_{1/2}^n M_{2^{-{z_\Delta}}}^{\lceil \log_2 1/\Delta\rceil}$. It follows that
\begin{align}
\squeezingparam(U)&=2^n\cdot 2^{\lceil \log_2 1/\Delta\rceil\cdot z_\Delta}=\textfrac{2^n}{\Delta}\ .
\end{align}
On the other hand,
each factor~$V$ contains displacements~$e^{i\pi Q}$ and~$e^{-iP}$, and the unitary~$U$ additionally contains the factor~$e^{iP}$. It follows that
\begin{align}
\overline{\xi}(U)&=1+n (\pi+1)
\end{align}
\end{proof}

The difference between~$\ket{\Sha_{2^n,\Delta}}$ and the state state~$\ket{\Sha^\varepsilon_{2^n,\Delta}}$ is the lack of truncation in the former. In Lemma~A.6 of~\cite{brenner2024complexity}, it is shown that these states are close for suitable choices of parameters, that is, we have the following.
\begin{lemma}(\cite[Lemma A.6]{brenner2024complexity}, specialized to powers of~$2$)
 Let~$\varepsilon \in(0,1 / 2), \Delta \in(0,1 / 4)$, and~$n\in\mathbb{N}$. Then
\begin{align}
\left|\left\langle \Sha_{2^n, \Delta}, \Sha_{2^n, \Delta}^{\varepsilon}\right\rangle\right|^2 \geq 1-16 \Delta^2-2 e^{-(\varepsilon / \Delta)^2}\ .\label{eq:bdvacvz}
\end{align}
In particular, 
\begin{align}
\left\|\proj{\Sha_{2^n,\Delta}}-\proj{\Sha_{2^n,\Delta}^\varepsilon}\right\|_1
&\leq 8\Delta+ 6 (\Delta/\varepsilon)^2\ .\label{eq:shadeltadeltavarepsilon}
\end{align}
\end{lemma}
\begin{proof}
We refer to~\cite[Lemma A.6]{brenner2024complexity} for the proof of Eq.~\eqref{eq:bdvacvz}. 
Using that the trace distance and the overlap of two pure states~$\ket{\phi}$,~$\ket{\psi}$ are related by~$\|\proj{\psi}-\proj{\phi}\|_1=2\sqrt{1-|\langle \psi,\phi\rangle|^2}$, we obtain 
\begin{align}
\left\|\proj{\Sha_{2^n,\Delta}}-\proj{\Sha_{2^n,\Delta}^\varepsilon}\right\|_1
&\leq 2 \left(16 \Delta^2+2 e^{-(\varepsilon / \Delta)^2}\right)^{1/2}\\
&\leq 2 \left(4 \Delta+\sqrt{2} e^{-\frac{1}{2}(\varepsilon / \Delta)^2}\right)\\
&\leq 8\Delta+2\sqrt{2}e^{-\frac{1}{2}(\varepsilon/\Delta)^2}\\
&\leq 8\Delta+ 6 (\Delta/\varepsilon)^2\label{eq:shadeltadeltavarepsilon}
\end{align}
where we additionally used the inequality~$\sqrt{a^2+b^2}\leq a+b$ for~$a,b\geq 0$,~$2 \sqrt{2} \le 3$ and~$e^{-x} \le 1/x$ for~$x>0$. 
\end{proof}

We obtain a preparation circuit for the state~$\ket{\Sha_{2^n,\Delta}^\varepsilon(0)_d}$ as follows.
\begin{lemma}[Code state preparation for~$\gkpcoderect{2^n,\Delta}{\varepsilon}{d}$]
  \label{lem:prepcomb1x}
  Let~$\varepsilon\in (0,1/2)$,~$\Delta\in (0,1/4)$.  Let~$d\geq 2$ be an integer and~$n\in\mathbb{N}$. Define the circuit
\begin{align}
U_{2^n,\Delta}(0)_d&:=
M_{2^{z_d}}^{\lceil \log_2 \sqrt{2\pi d}\rceil}
U_{2^n,\Delta}\ 
\end{align}
where~$U_{2^n,\Delta}$ is the circuit introduced in Theorem~\ref{thm:preparationproceduregmvd} and where 
\begin{align}
  z_d&:= \frac{\log_2 \sqrt{2\pi d}}{\lceil \log_2 \sqrt{2\pi d}\rceil}\ .
  \end{align}  
Then the output state
 \begin{align}
 \ket{\Phi_{2^n,\Delta}(0)_d}:=U_{2^n,\Delta}(0)_d(\ket{\vac}\otimes \ket{0})
 \end{align}
satisfies 
\begin{align}
\left\| 
\proj{\Phi_{2^n,\Delta}(0)_d}-\proj{\Sha_{2^n,\Delta}^\varepsilon(0)_d} \otimes \proj{0}\right| \|_1 \leq 25\sqrt{\Delta}+6 (\Delta/\varepsilon)^2\ .\label{eq:distanceboundpreparationclaim}
\end{align}
The circuit~$U_{2^n,\Delta}(0)_d$ consists of 
\begin{align}
\size(U_{2^n,\Delta}(0)_d)&\leq 
 \log_2 \sqrt{2\pi d}+5n+\log_2 1/\Delta +4\ .\label{eq:sizeupperboundmzved}
\end{align}
gates belonging to~$\Uelemunbounded{1}{1}$. 
Furthermore, we have 
\begin{align}
\begin{matrix}
\overline{\xi}(U_{2^n,\Delta}(0)_d) &\leq &n(\pi+1)+1&\\
\squeezingparam(U_{2^n,\Delta}(0)_d) &\leq &2^n/\Delta \cdot \sqrt{2\pi d}&\ .
\end{matrix}\label{eq:squeezinboundvtwondelta}
\end{align}
\end{lemma}

\begin{proof}
  Since~$|z_d|\leq 1$ by definition, the circuit~$U_{2^n,\Delta}(0)_d$ has size
  \begin{align}
  \size(U_{2^n,\Delta}(0)_d)&=\lceil \log_2\sqrt{2\pi d}\rceil+\size(U_{2^n,\Delta})\\
  &\le \log_2\sqrt{2\pi d} + \size(U_{2^n,\Delta}) +1 \ .
  \end{align} 
  The Claim~\eqref{eq:sizeupperboundmzved} thus follows from 
  Eq.~\eqref{eq:shatwndeltadefreached} of Theorem~\ref{thm:preparationproceduregmvd}. 
  
  It follows from the unitary invariance of the~$1$-norm
  \begin{align}
  &\left\|\proj{\Phi_{2^n,\Delta}(0)_d}-\proj{\Sha_{2^n,\Delta}^\varepsilon(0)_d}\otimes \proj{0}\right\|_1\\
  &\qquad\qquad=\left\|\proj{\Phi_{2^n,\Delta}}-
  \proj{\Sha_{2^n,\Delta}^\varepsilon}\otimes\proj{0}\right\|_1\, .
  \end{align}
  With the triangle inequality,  Theorem~\ref{thm:preparationproceduregmvd} (i.e., Eq.~\eqref{eq:shatwndeltadefreached}) and Eq.~\eqref{eq:shadeltadeltavarepsilon}
  we obtain
  \begin{align}
  &\left\|\proj{\Phi_{2^n,\Delta}}-\proj{\Sha_{2^n,\Delta}^\varepsilon}\otimes\proj{0}\right\|_1\\
  &\quad\leq \left\|\proj{\Phi_{2^n,\Delta}}-
  \proj{\Sha_{2^n,\Delta}}\otimes \proj{0}\right\|_1 \\
  &\qquad \qquad +\left\|\proj{\Sha_{2^n,\Delta}}\otimes\proj{0}-
  \proj{\Sha_{2^n,\Delta}^\varepsilon}\otimes\proj{0}\right\|_1\\
  &\quad\leq 17\sqrt{\Delta}+8\Delta+6(\Delta/\varepsilon)^2\ .
  \end{align}
  As $\Delta\in(0,1/4)$, this implies the Claim~\eqref{eq:distanceboundpreparationclaim}.

  By definition, we have
  \begin{align}
  U_{2^n,\Delta}(j)_d &= M_{2^{z_d}}^{\lceil \log_2 \sqrt{2\pi d}\rceil} U_{2^n,\Delta}\, .
  \end{align}
  It follows that 
   \begin{align}
  \overline{\xi}(U_{2^n,\Delta}(j)_d) &\leq \overline{\xi}(M_{2^{z_d}}^{\lceil \log_2 \sqrt{2\pi d}\rceil}) +\overline{\xi}(U_{2^n,\Delta}(j)_d) \\
  \squeezingparam(U_{2^n,\Delta}(j)_d) &\leq  \squeezingparam(M_{2^{z_d}}^{\lceil \log_2 \sqrt{2\pi d}\rceil})\cdot \squeezingparam(U_{2^n,\Delta}(j)_d)\ .
  \end{align}
  By definition we have 
  \begin{align}
  \overline{\xi}(M_{2^{z_d}}^{\lceil \log_2 \sqrt{2\pi d}\rceil}) &= 0\\
  \squeezingparam(M_{2^{z_d}}^{\lceil \log_2 \sqrt{2\pi d}\rceil}) &= \sqrt{2\pi d}\, .
  \end{align}
  This together with Theorem~\ref{thm:preparationproceduregmvd} 
  implies the Claim~\eqref{eq:squeezinboundvtwondelta}.

  \end{proof}

We can specialize Lemma~\ref{lem:prepcomb1x} as follows:

\begin{lemma}[Code state preparation for~$\gkpcoderect{\Delta}{\star}{2^\ell}$]
  \label{lem:prepcombcombstate} 
Let~$\ell\in\mathbb{N}$ and~$\Delta\in (0,1/4)$ be such that~$\Delta<2^{-(\ell+1)}$.
 Then there is a circuit 
~$U^\star_{\Delta}(0)_{2^\ell}$ on~$L^2(\mathbb{R})\otimes \mathbb{C}^2$ with the following property:
 The output  state 
 \begin{align}
 \ket{\Phi^\star(0)_{2^\ell}}:=U^\star_{\Delta}(0)_{2^\ell}(\ket{\vac}\otimes \ket{0})
 \end{align} satisfies
 \begin{align}
\left\| \proj{\Phi^\star(0)_{2^\ell}}-
\proj{\Sha_{\Delta}^\star(0)_{2^\ell}}\otimes \proj{0}\right\|_1 \leq 25 \left(\sqrt{\Delta}+ 2^{2\ell}\Delta^2\right)\ .
\end{align}
The circuit consists of
\begin{align}
\size(U^\star_{\Delta}(0)_{2^\ell})&\leq 21\log 1/\Delta
\end{align}
elementary operations and 
\begin{align}
\overline{\xi}(U^\star_{\Delta}(0)_{2^\ell})&\leq 10\log 1/\Delta\\
\squeezingparam(U^\star_{\Delta}(0)_{2^\ell})&\leq 4/\Delta^3\ .
\end{align}
  \end{lemma}
  \begin{proof}
  Recall that the state~$\ket{\Sha_{\Delta}^\star(0)_{2^\ell}}$ is defined with the truncation parameter 
  \begin{align}
      \varepsilon&=2^{-(\ell+1)}\ .
  \end{align}
  and~$L=2^n$, where
    \begin{align}
  n&=2(\lceil\log_2 1/\Delta \rceil - \ell)\ .
  \end{align}
  We use Lemma~\ref{lem:prepcomb1x} with these parameters and~$d=2^\ell$, obtaining a circuit~$U^\star_{\Delta}(0)_{2^\ell}:=U_{2^n,\Delta}(0)_{2^\ell}$
  such that the output state~$\ket{\Phi^\star(0)_{2^\ell}}:=\ket{\Phi_{2^n,\Delta}(0)_{2^\ell}}$ satisfies
  \begin{align}
  \left\|\proj{\Phi_{\Delta}^\star(0)_{2^\ell}}-\proj{\Sha^{2^{-(\ell+1)}}_{2^n,\Delta}(0)_{2^\ell}}\otimes \proj{0}\right\|_1&\leq 25 \sqrt{\Delta}+6 \Delta^2 2^{2(\ell+1)} \\
  &\le 25 \left(\sqrt{\Delta} + 2^{2\ell} \Delta^2\right)\, , 
  \end{align}
  and
  \begin{align}
  \size(U^\star_\Delta(0)_{2^\ell})&\leq \log_2 \sqrt{2\pi \cdot 2^\ell}+5n+\log_2 1/\Delta+4\\
  &\leq 2\ell+5n+\log_2 1/\Delta+ 4\\
  &= 11\log_2 1/\Delta-8 \ell+14\\
  &\leq 11 \log_2 1/\Delta+6\\
  &\leq 21 \log 1/\Delta 
  \end{align}
for~$\Delta<1/4$ and~$\ell \in \mathbb{N}$. 
Finally, using that~$n\leq 2(\log_2 1/\Delta -1)$ and~$\Delta<1/4$ we have 
\begin{align}
\overline{\xi}(U^\star_\Delta(0)_{2^\ell})&\leq n(\pi+1)+1\\
&\leq 5 n+1\\
&\leq 10\log_2 1/\Delta - 10 \ell + 1\\
&\leq 10\log_2 1/\Delta -9 \\
&\leq 10\log 1/\Delta
\end{align}
and 
\begin{align}
\squeezingparam(U^\star_\Delta(0)_{2^\ell})&=2^n/\Delta \cdot \sqrt{2\pi\cdot 2^\ell}\\
&= \sqrt{2\pi} \cdot 2^{-3\ell/2} \cdot 2^{2\log_2 1/\Delta + 2} /\Delta  \\
&\le 4 /\Delta^3\, .
\end{align}
    \end{proof}

  Similarly, the circuit preparing the auxiliary state~$\ket{\Sha_{\Delta,\ell}^{\mathsf{aux}}(0)_2}$ satisfies the following.
\begin{lemma}[Preparation of auxiliary GKP states]\label{lem:auxprep}
  Let~$\ell\in\mathbb{N}$ and~$\Delta\in (0,1/4)$ be such that~$\Delta<2^{-(\ell+1)}$. 
   Then there is a circuit~$U^\mathsf{aux}_{\Delta,\ell}(0)_{2}$ on~$L^2(\mathbb{R})\otimes \mathbb{C}^2$ with the following property:
   The output  state 
   \begin{align}
   \ket{\Phi^{\mathsf{aux}}(0)_{2}}:=U^{\mathsf{aux}}_{\Delta,\ell}(0)_{2}(\ket{\vac}\otimes \ket{0})
   \end{align} satisfies
   \begin{align} \label{eq:closenessaux}
  \left\| \proj{\Phi^\mathsf{aux}(0)_{2}}-
  \proj{\Sha_{\Delta,\ell}^{\mathsf{aux}}(0)_2}\otimes \proj{0}\right\|_1 \leq 25\left(\sqrt{\Delta}+  2^{2\ell}\Delta^2\right)\ .
  \end{align}
  The circuit consists of
  \begin{align}
  \size(U^{\mathsf{aux}}_{\Delta, \ell}(0)_2)&\leq 21\log 1/\Delta
  \end{align}
  elementary operations and 
  \begin{align}
  \overline{\xi}(U^{\mathsf{aux}}_{\Delta, \ell}(0)_2)&\leq 10\log 1/\Delta\\
  \squeezingparam(U^{\mathsf{aux}}_{\Delta, \ell}(0)_2)&\leq 4/\Delta^3\ .
  \end{align}
\end{lemma}
\begin{proof} 
Using Lemma~\ref{lem:prepcomb1x} with parameter~$n=2(\lceil\log_2 1/\Delta \rceil - \ell)$,~$\varepsilon= 2^{-(\ell+1)}$ and~$d=2$, we obtain a circuit~$U^\mathsf{aux}_{\Delta,\ell}(0)_{2}:=U_{2^n,\Delta}(0)_{2}$ which satisfies 
the the bound in Eq.~\eqref{eq:closenessaux}. The remaining claims follow by the same reasoning as Lemma~\ref{lem:prepcombcombstate}.
\end{proof}

Using Lemma~\ref{lem:prepcombcombstate}, we can give the proof of Theorem~\ref{thm:initialstateprep} in the main text. For completeness, we provide the statement from the main text.

\begingroup
\renewcommand{\thetheorem}{\ref{thm:initialstateprep}}
\begin{theorem}[Restated]
  Let~$\ell\in\mathbb{N}$ and~$\Delta\in (0,1/4)$ be such that~$\Delta \leq 2^{-(\ell+1)}$.  Let~$m\in\mathbb{N}$. 
  Then there is a circuit
  \begin{align}
  W^{\mathsf{prep}}=W_{\size(W^{\mathsf{prep}})}\cdots W_{1}
  \end{align}
  on~$L^2(\mathbb{R})^{\otimes m+1}\otimes \mathbb{C}^2$ 
  composed of
  \begin{align}
  \size(W^{\mathsf{prep}})\leq 42m \log 1/\Delta\label{eq:sizeupperboundUprepapp}
  \end{align}
  elementary operations belonging to~$\Uelem{m}{1}$ such that the
  output state
  \begin{align}
  \ket{\Phi_{\mathsf{init}}}_{B_1\cdots B_mB_{\mathsf{aux}}Q_1Q_2Q_3}:=\left(W^{\mathsf{prep}}_{B_1\cdots B_mB_{\mathsf{aux}}Q_1}\otimes I_{Q_2Q_3}\right)\left(
  \ket{\vac}^{\otimes m+1}_{B_1\cdots B_mB_{\mathsf{aux}}}\otimes \ket{0}^{\otimes 3}_{Q_1Q_2Q_3}\right)
  \end{align}
  when applying~$W^{\mathsf{prep}}$ to the bosonic modes prepared in the vacuum state and the qubits in the~$\ket{0}$-state satisfies
  \begin{align}
  \varepsilon_{prep}&:=\left\|
  \proj{\Phi_{\mathsf{init}}}-
  \proj{\Phi^{\mathsf{ideal}}_{\mathsf{init}}}
  \right\|_1\leq 50m \left(\sqrt{\Delta}+ 2^{2\ell}\Delta^2\right)\ \label{eq:claimdistancemfoldapp}
  \end{align}
  where~$\ket{\Phi^{\mathsf{ideal}}_{\mathsf{init}}}\in B_1\cdots B_mB_{\mathsf{aux}}Q_1Q_2Q_3$
  is the ideal initial state defined in Eq.~\eqref{eq:idealinitialstatedef}.
\end{theorem}
\addtocounter{theorem}{-1} 
\endgroup
    \begin{proof}
Let us denote the oscillators by~$B_1\cdots B_m B_{\mathsf{aux}}\cong L^2(\mathbb{R})^{\otimes (m+1)}$ and the qubit by~$Q_1\cong \mathbb{C}^2$. 
Let~$U=U_\Delta^\star(0)_{2^\ell}$ be the circuit on~$L^2(\mathbb{R})\otimes \mathbb{C}^2$
from Lemma~\ref{lem:prepcombcombstate} and let~$V=U^\mathsf{aux}_{\Delta,\ell}(0)_{2}$ be the circuit from Lemma~\ref{lem:auxprep}.   Recall that
\begin{align}
\begin{matrix}
\max \left\{ \size(U), \size(V) \right\} &\leq & 21 \log 1/\Delta
\end{matrix}\ \label{eq:Upropertiesv}
\end{align}
and 
\begin{align}
\left\|
U(\proj{\vac}\otimes\proj{0})U^\dagger - \proj{\Sha^\star_\Delta(0)_{2^\ell}}\otimes\proj{0}\right\|_1
&\leq 25\left(\sqrt{\Delta}+ 2^{2\ell}\Delta^2\right) \\
\left\|
V(\proj{\vac}\otimes\proj{0})V^\dagger - \proj{\Sha^\mathsf{aux}_{\Delta,\ell}(0)_{2}}\otimes\proj{0}\right\|_1
&\leq 25\left(\sqrt{\Delta}+ 2^{2\ell}\Delta^2 \right)\ .\label{eq:propertUprojvac}
\end{align}
We define the circuit~$W^{\mathsf{prep}}$ as the concatenation
\begin{align}
W^{\mathsf{prep}}&=V_{B_{\mathsf{aux}}Q_1}U_{B_mQ_1}\cdots U_{B_1Q_1}=:W_{\size(W^{\mathsf{prep}})}\cdots W_1\ \label{eq:defWprep}
\end{align}
where the sequence~$\{W_j\}_{j=1}^{\size(W^{\mathsf{prep}})}$ is obtained by inserting the definition of~$U$ and~$V$ respectively. 

Then Eq.~\eqref{eq:sizeupperboundUprepapp} immediately follows from the  definition of~$W^{\mathsf{prep}}$ and Eq.~\eqref{eq:Upropertiesv} using that $m+1 \le 2m$.

Finally,  Eq.~\eqref{eq:claimdistancemfoldapp}
follows inductively by using the triangle inequality and Eq.~\eqref{eq:propertUprojvac}~$m+1\le 2m$~times,  
in addition to the stabilization property~$\|\proj{\vac}\otimes A\|_1=\|A\|_1$ of the trace norm. 
\end{proof}

\section{Moment-limits on implementations of logical unitaries } \label{sec: momlimitsimplementation}
In this section we derive moment-limits on the unitary circuits~$W^{\mathsf{prep}}$ and~$W_U$ introduced in Theorems~\ref{thm:initialstateprep} and~\ref{thm:implementationlogicalqubit}, respectively. Subsequently, we use them to prove Theorem~\ref{lem:totalsqueezingbound}.

\begin{lemma}[Moment-limits of the state preparation circuit] \label{lem:squeezingWprep}
  Let~$\ell \in \mathbb{N}$,~$\Delta \in (0,1.4)$ such that~$\Delta \le 2^{-(\ell+1)}$ and~$m\in\mathbb{N}$. Let~$W^\mathsf{prep}$ be the state preparation unitary on~$L^2(\mathbb{R})^{\otimes (m+1)}  \otimes \mathbb{C}^2$ introduced in Theorem~\ref{thm:initialstateprep} (cf. Eq.~\eqref{eq:defWprep}).
  Then 
  \begin{align}
    \begin{aligned}
      \overline{\xi}(W^{\mathsf{prep}})&\leq &10\log 1/\Delta \ ,\\ 
      \overline{g}(W^{\mathsf{prep}})&\leq &4/\Delta^3\ .
      \end{aligned}\label{eq:wprepcircuitprop}
  \end{align}
  In particular, the amount of energy in the preparation is bounded by 
  \begin{align}
    \energy(W^{\mathsf{prep}}) \le 4096/\Delta^{18}\left(2 + 1000\log^3 1/\Delta \right)\, . \label{eq:boundWprepsqueezing}
  \end{align}
\end{lemma}
\begin{proof}
Let~$U=U_\Delta^\star(0)_{2^\ell}$ be the circuit on~$L^2(\mathbb{R})\otimes \mathbb{C}^2$ from Lemma~\ref{lem:prepcombcombstate} and let~$V=U^\mathsf{aux}_{\Delta,\ell}(0)_{2}$ be the circuit from Lemma~\ref{lem:auxprep}. It follows that 
\begin{align}
  \begin{aligned}
  \max\left\{\overline{\xi}(U),  \overline{\xi}(V)\right\}&\leq 10\log 1/\Delta \ ,\\
\max\left\{\overline{g}(U),\overline{g}(V) \right\} &\leq 4/\Delta^3 \ .
  \end{aligned} \label{eq:boundsqueezing}
\end{align}
Claim~\eqref{eq:wprepcircuitprop} then follows from the definition of~$W^\mathsf{prep}$ (see Eq.~\eqref{eq:defWprep}) in combination with Eq.~\eqref{eq:boundsqueezing}.

Claim~\eqref{eq:boundWprepsqueezing} follows from Claim~\eqref{eq:wprepcircuitprop} in combination with Lemma~\ref{lem:PhiUpartiallyimmplementedmulti}.
\end{proof}

\begin{lemma}[Moment-limits of the implementation of logical circuits] \label{thm:momentlimitlogicalcircuit}
  Let~$U=U_s\cdots U_1$ be a unitary circuit on~$n'=m\ell$~qubits of size~$s$, i.e., composed of~$s$ one- and two-qubit gates~$U_1,\ldots,U_s$. Let~$W_U = (W_U)_{B_1 \dots B_m B_{\mathsf{aux}}Q_1 Q_2 Q_3}$ be the unitary circuit introduced Theorem~\ref{thm:implementationlogicalqubit} acting on the space~$L^2(\mathbb{R})^{\otimes (m+1)}\otimes (\mathbb{C}^2)^{\otimes 3}$.
  Then we have 
  \begin{align}
    \begin{aligned}
    \overline{\xi}(W_U)&\le 72s \cdot 2^{\ell}\\
      \squeezingparam(W_U)&\leq  256 \cdot 2^{148\ell}\, .
    \end{aligned}
    \end{align}
\end{lemma}
\begin{proof}
To avoid handling separate cases, in the following we treat single-qubit unitaries as two-qubit unitaries.
Let us write~$W_U= W_{U_s} \cdots W_{U_1}$ where each unitary~$W_{U_t}$ is a implementation of the (two-qubit) unitary~$U_t$ for all~$t\in \{1,\dots, s\}$.
For each~$t$ let~$j_t,k_t \in \{1,\dots,n'\}$ be two indices such that~$U_t$ acts trivially on all qubits excepts on the~$j_t$-th and~$k_t$-th. Let $\alpha_{j_t}, \alpha_{k_t} \in \{1, \ldots, m\}$ and  
$\beta_{j_t}, \beta_{k_t} \in \{0,\dots, \ell-1\}$ such that 
\begin{align}
  j_t -1  & = (\alpha_{j_t} - 1) \ell + \beta_{j_t} \\
  k_t - 1 & = (\alpha_{k_t} - 1) \ell + \beta_{k_t}\, .
\end{align}
It follows that~$B_{\alpha_{j_t}}$ and~$B_{\alpha_{k_t}}$ are the modes in which the~$j_t$-th and~$k_t$-th qubit are encoded. Moreover,~$\beta_{j_t}$ and~$\beta_{k_t}$ determine which qubit they correspond to within the mode~$B_{\alpha_{j_t}}$  and~$B_{\alpha_{k_t}}$, respectively.

Let~$W_{\bittransfer{r}{\ell}}$ be the unitary on~$L^2(\mathbb{R})^{\otimes 2} \otimes (\mathbb{C}^2)^{\otimes 2}$ which implements the bit-transfer unitary~$\bittransfer{r}{\ell}$ on~$\mathbb{C}^{2^\ell} \otimes \mathbb{C}^2$ for~$ r\in \{0,\dots \ell-1\}$ as introduced in~\cite[Section 2.2]{cliffordshybrid2025}. Then 
\begin{align}
  W_{U_t} = \left(W_{\bittransfer{\beta_{j_t}}{\ell}}\right)^\dagger_{B_{\alpha_{k_t}} C Q_3} \left(W_{\bittransfer{\beta_{k_t}}{\ell}}\right)^\dagger_{B_{\alpha_{j_t}} C Q_2} (U_t)_{Q_2Q_3} \left(W_{\bittransfer{\beta_{k_t}}{\ell}}\right)_{B_{\alpha_{k_t}} C Q_3} \left(W_{\bittransfer{\beta{j_t}}{\ell}}\right)_{B_{\alpha_{j_t}} C  Q_3}  \label{eq:W_Uimplement}
\end{align} where for better readability we introduced the system~$C = B_{\mathsf{aux}} Q_1$.   Moreover, by slight abuse of notation we identified the multiqubit unitary~$U_t$ with the two-qubit unitary obtained by removing all but the~$j_t$-th and~$k_t$-th qubit.
We refer to Fig.~\ref{fig:W_Uimplement} for a circuit representation of Eq.~\eqref{eq:W_Uimplement} for~$m=1$.

Due to~\cite[Lemma 3.2]{cliffordshybrid2025} we can write~$W_{\bittransfer{r}{\ell}} = W^{(L_r)} \cdots W^{(1)}$ for~$r \in \{0,\dots, \ell-1\}$ where~$W^{(a)} \in \Uelem{2}{2}(2,\zeta)$ for all~$a\in \{1, \dots, L_r\}$ with~$L_r \le 36\ell$ and 
\begin{align}
   \zeta = \sqrt{\pi}\cdot  2^{(\ell-1)/2} \le 2\cdot 2^{\ell/2} \label{eq:boundzeta}\, .
\end{align}
By combining Lemma~\ref{lem:substitutionsinglemodegen} (setting~$U = W_{\bittransfer{r}{\ell}} = W^{(L_r)} \cdots W^{(1)}$ and using~$L^{(\alpha)} \le L_r$ and~$|\subsset^{(\alpha)}|\le L_r$ for all~$r \in \{0,\dots, \ell-1\}$) and Eq.~\eqref{eq:boundzeta} we find 
\begin{align}
  \overline{\xi}((W_{\bittransfer{r}{\ell}})_{B_{\alpha_r}CQ_{b_r}}) = \overline{\xi}(W_{\bittransfer{r}{\ell}})&\leq L_r \le 36 \ell \le 18 \cdot 2^{\ell}\\
  \squeezingparam((W_{\bittransfer{r}{\ell}})_{B_{\alpha_r}CQ_{b_r}}) = \squeezingparam(W_{\bittransfer{r}{\ell}})&\leq \zeta^2 \cdot 2^{L_r} \le 4 \cdot 2^{37\ell}
  \end{align}
  for all~$r \in \{0,\dots, \ell-1\}$,~$\alpha_r \in \{1, \dots, m\}$ and~$b_r \in \{2,3\}$.
Here we used that~$36\ell   \le 18\cdot 2^\ell$ for all~$\ell \in \mathbb{N}$.
  Using Lemma~\ref{lem:squeezingdisplacementsubcircuits} it follows that 
  \begin{align}
    \overline{\xi}((W_{\bittransfer{\beta_{k_t}}{\ell}})_{B_{\alpha_{k_t}}CQ_{3}} (W_{\bittransfer{\beta{j_t}}{\ell}})_{B_{\alpha_{j_t}}CQ_{2}})&\leq  \overline{\xi}(W_{\bittransfer{\beta_{k_t}}{\ell}}) + \overline{\xi}( W_{\bittransfer{\beta_{j_t}}{\ell}}) \le 36 \cdot 2^{\ell}\\
    \squeezingparam((W_{\bittransfer{\beta{k_t}}{\ell}})_{B_{\alpha_{k_t}}CQ_{3}} (W_{\bittransfer{\beta{j_t}}{\ell}})_{B_{\alpha_{j_t}}CQ_{2}} )&\leq \squeezingparam(W_{\bittransfer{\beta_{k_t}}{\ell}}) \cdot \squeezingparam(W_{\bittransfer{\beta_{j_t}}{\ell}})\le 16 \cdot 2^{74\ell} \label{eq:momentlimitbittransfer}
    \end{align} for all~$t \in \{1, \dots, s\}$.
     We observe (see Eq.~\eqref{eq:W_Uimplement}) that~$W_U$ is a dressed circuit. Therefore we can apply Lemma~\ref{lem:momentlimitdressedcircuit} with~$U^{(t)} = (W_{\bittransfer{\beta_{k_t}}{\ell}})_{B_{\alpha_{k_t}}CQ_{3}} (W_{\bittransfer{\beta{j_t}}{\ell}})_{B_{\alpha_{j_t}}CQ_{2}}$ and~$V_t = U_t$ for~$t\in \{1, \dots, s\}$.  This implies in combination with Eq.~\eqref{eq:momentlimitbittransfer} that 
     \begin{align}
      \overline{\xi}(W_U)&  \le 72s \cdot 2^{\ell}\\
      \squeezingparam(W_U)&\leq  256 \cdot 2^{148\ell}\, .
     \end{align}
\end{proof} 

With these preparations we can give the proof of Theorem~\ref{lem:totalsqueezingbound}. For completeness, we restate the claim from the main text.
\begingroup
\renewcommand{\thetheorem}{\ref{lem:totalsqueezingbound}}
\begin{theorem}[Restated]
  Let~$U=U_s\cdots U_1$ be a unitary circuit on~$n'=m\ell$~qubits of size~$s$, i.e., composed of~$s$ one- and two-qubit gates~$U_1,\ldots,U_s$. 
  Let~$W^{\mathsf{prep}}$ be the preparation circuit acting on~$L^2(\mathbb{R})^{\otimes (m+1)}\otimes \mathbb{C}^2$ introduced in Theorem~\ref{thm:initialstateprep} and let~$W_U$ be the unitary acting on~$L^2(\mathbb{R})^{\otimes (m+1)}\otimes (\mathbb{C}^2)^{\otimes 3}$ which implements the circuit~$U$ as introduced in Theorem~\ref{thm:implementationlogicalqubit}.
  Then the circuit~$W^{\mathsf{tot}} = W_U (W^{\mathsf{prep}} \otimes I_{\mathbb{C}^2}^{\otimes 2})$ satisfies 
  \begin{align}
    \energy(W^\mathsf{tot}) \le s^3 \cdot 2^{891\ell + 62} /\Delta^{21} \, .
  \end{align}
\end{theorem}
\addtocounter{theorem}{-1} 
\endgroup
\begin{proof}
 We show the claim using the notion of moment-limiting functions
It is easy check that 
\begin{align}
 \begin{matrix}
   \overline{\xi}(W^{\mathsf{prep}} \otimes I_{\mathbb{C}^2}^{\otimes 2})& = &  \overline{\xi}(W^{\mathsf{prep}}) &\leq &10\log 1/\Delta\\
   \overline{g}(W^{\mathsf{prep}} \otimes I_{\mathbb{C}^2}^{\otimes 2}) & = & \overline{g}(W^{\mathsf{prep}}) & \leq &4/\Delta^3\, ,
 \end{matrix}\ \label{eq:sqtot1}
 \end{align}
 where we used Lemma~\ref{lem:squeezingWprep} to obtain the inequalities.
 Moreover, by Lemma~\ref{thm:momentlimitlogicalcircuit} we have 
   \begin{align}
     \begin{aligned}
     \overline{\xi}(W_U)&\le 72s \cdot 2^{\ell}\\
       \squeezingparam(W_U)&\leq  256 \cdot 2^{148\ell}\, .
     \end{aligned} \label{eq:sqtot2}
     \end{align} 
Combining Eqs.~\eqref{eq:sqtot1} and~\eqref{eq:sqtot2} with Lemma~\ref{lem:squeezingdisplacementsubcircuits} it follows that 
   \begin{align}
     \begin{matrix}
     \overline{\xi}(W^{\mathsf{tot}})& \le & \overline{\xi}(W_U ) + \overline{\xi}(W^{\mathsf{prep}} \otimes I_{\mathbb{C}^2}^{\otimes 2}) & \le&  72s \cdot 2^{\ell} + 10\log 1/\Delta\\
       \squeezingparam(W^{\mathsf{tot}})&\leq& \overline{g}(W_U ) \cdot \overline{g}( W^{\mathsf{prep}} \otimes I_{\mathbb{C}^2}^{\otimes 2})& \le &  1024 \cdot 2^{148\ell} /\Delta^3\, .
     \end{matrix} \label{eq:etagwtot}
     \end{align}
     Finally, Lemma~\ref{lem:PhiUpartiallyimmplementedmulti} in combination with Eqs.~\eqref{eq:etagwtot} implies that 
     \begin{align}
       \energy(W^{\mathsf{tot}}) &\le  \squeezingparam(W^{\mathsf{tot}})^6\left(2 + \overline{\xi}(W^{\mathsf{tot}})^3\right)\\
       &\le 2^{888\ell + 40}/\Delta^{18} \left(2+ (72s\cdot 2^\ell + 10\log1/\Delta)^3\right)\\
       &\le 2^{888\ell + 40}/\Delta^{18}\left(2 + 4( (72s\cdot 2^\ell)^3 + (10 \log1/\Delta)^3)\right)\\
       &\le  2^{891\ell + 61}/\Delta^{18}\left(s^3 + \log^3 1/\Delta\right)\\
       &\le 2^{891\ell + 61}/\Delta^{18}\left(s^3 + 1/\Delta^3\right)\\
       &\le 2^{891\ell + 61}\Delta^{21}\left(s^3 + 1\right) \\
       &\le s^3 \cdot 2^{891\ell + 62} /\Delta^{21}\, .
     \end{align}
     The third line follows from the bound $(x+y)^3 \le 4 (x^3 + y^3)$ for all $x,y \ge 0$. To obtain the forth line we used that $s \ge 1$ and $\max\{\log_2(4\cdot 10^3),\log_2(4\cdot72^3+2)\}\le 21$.
     In the fifth line we used~$\log(x) \le x$ for all~$x >0$. 
     The penultimate inequality follows from the fact that by assumption we have $\Delta < 1$, see Lemma~\ref{lem:squeezingWprep}. Finally, in the last inequality we used $s\ge1$ and thus $s^3 +1 \le 2s^3$.
\end{proof}

\section{Squeezing and energy\label{sec:squeezingenergy}}
In this section, give  relations between the amount of squeezing (suitably quantified)
and the amount of energy of a state. In more detail, we introduce 
a quantity we call the diameter of a state~$\rho_{B_1\cdots B_mQ_1\cdots Q_r}\in \cB(L^2(\mathbb{R})^{\otimes m}\otimes (\mathbb{C}^2)^{\otimes r})$.  The definition is motivated by considering the amount of squeezing of a state. We will then show that it gives a lower bound on the energy of a state.

We start with a few general remarks on the degree of localization of a probability distribution on~$\mathbb{R}$ in Section~\ref{sec:localization}. In Section~\ref{sec:squeezingonemodestate}, we 
translate the corresponding notions to quantum states and discuss the connection to squeezing. We first consider the one-mode case $(m,r)=(1,0)$.
In Section~\ref{sec:squeezingmultimodestate}, we then define the relevant quantities for the multimode case. 

\subsection{Diameter of a random variable and variance\label{sec:localization}}
In the following, let $X$ denote a random variable on~$\mathbb{R}$ with finite first and second moments. The 
variance  $\Var(X):=\ExpE\left[(X-\ExpE[X])^2\right]$ is often used to quantify how ``wide'', i.e., spread out the distribution of such a random variable~$X$ is.  Indeed, according to 
Chebyshev's inequality
$\Pr\left[|X-\ExpE[X]| \geq R\right] \leq \Var(X)/R^2$ for $R\geq 0$, the probability that $X$ can be observed 
in an interval of length $2\Varroot(X)\cdot \delta^{-1/2}$ is at least $1-\delta$ for any $\delta>0$.  Thus the quantity~$2\Varroot(X)$ can be seen as determining an ``effective diameter'' of~$X$.

Here we consider a simpler notion. A natural first attempt is to use the diameter
\begin{align}
\diam(X):=\diam(\supp(X))
\end{align}
of the support~$\supp(X)$ of $X$, where 
\begin{align}
\diam(A):=\sup\{|x-y|\ |\ x,y\in A\}=\inf \{R_2-R_1\ |\ R_1\leq R_2, A\subseteq [R_1,R_2]\}
\end{align}
denotes the diameter of a subset~$A\subseteq\mathbb{R}$. Of course, the quantity~$\diam(X)$ is generally unbounded. For this reason, we consider the following definition, which involves taking the infimum of the diameter of any high-probability set: For $\delta>0$, let 
\begin{align}
\diam^\delta(X):=\inf \{\diam(A^\delta)\ |\ A^\delta\subseteq \mathbb{R}\textrm{ measurable }, \Pr[X\in A^\delta]\geq 1-\delta\}\ .
\end{align}
We call $\diam^\delta(X)$ the $\delta$-diameter of the random variable~$X$. This quantity~$\diam^\delta(X)$ is sometimes also referred to as the  $\delta$-essential diameter of a distribution on a metric space. It can equivalently be defined as the width of a minimal interval containing at least $1-\delta$ of the probability mass of~$X$, i.e., 
\begin{align}
\diam^\delta(X)=\inf \{R_2-R_1\ |\ R_1\leq R_2, \Pr\left[X\in [R_1,R_2]\right]\geq 1-\delta\}\ .
\end{align}

We observe the following relations to the variance~$\Var(X)$.
We denote by 
\begin{align}
\delta(X,Y):=\sup_{A\subset\mathbb{R}\textrm{ measurable}} |\Pr[X\in A]-\Pr[Y\in A]|
\end{align} the total variation distance of the probability measures associated with two random variables~$X$ and $Y$. We note that if $X$, $Y$ have probability densities~$f_X$ and $f_Y$, then $\delta(X,Y)=\frac{1}{2}\|f_X-f_Y\|_1$.

\begin{lemma}[Diameter and variance]\label{lem:variancediametercomparison}
Let $X$ be a random variable on~$\mathbb{R}$. Let $\delta>0$. Then 
the following holds.
\begin{enumerate}[(i)]
\item\label{it:firstinequalitydiameter}
We have $\diam^\delta(X)\leq 2\Varroot(X)\cdot \delta^{-1/2}$. 
\item
There is a random variable~$\overline{X}$ satisfying $\delta(X,\overline{X})\leq \delta$ and 
$2\Varroot(\overline{X})\leq \diam^\delta(X)$.
\label{it:secondinequalitydisameter}
\end{enumerate}

\end{lemma}
\begin{proof}
As argued above, Claim~\eqref{it:firstinequalitydiameter} immediately follows from Chebyshev's inequality.

To prove Claim~\eqref{it:secondinequalitydisameter},
assume for simplicity that $R_1<R_2$ achieve the infimum in the definition of~$\diam^\delta(X)$, i.e., 
$\diam^\delta(X)=R_2-R_1$ and $\Pr\left[X\in [R_1,R_2]\right]\geq 1-\delta$. (The general case follows by the same arguments.) Let $\overline{X}$ denote the random variable defined by the conditional distribution of~$X$ given that $X\in [R_1,R_2]$, i.e., for $\Pr\left[\overline{X}\in A\right]=p^{-1}\cdot \Pr\left[X\in A\cap [R_1,R_2]\right]$ for a measurable subset $A\subset \mathbb{R}$, where $p:=\Pr\left[X\in [R_1,R_2]\right]$. Then it is straightforward to check that the total variation distance of the corresponding distributions satisfies~$\delta(X,\overline{X})\leq \delta$. By definition, the random variable~$\overline{X}$ has bounded support contained in~$[R_1,R_2]$. The claim thus follows from Popoviciu's inequality $\Var(\overline{X})\leq (R_2-R_1)^2/4$ satisfied by such a random variable. 
\end{proof}
Let us also define the symmetric $\delta$-radius of  $X$ as
\begin{align}
\symradius^\delta(X)&:=\inf\left\{R>0\ |\ 
\Pr\left[X\in [-R,R]\right] \geq 1-\delta\right\}\ .
\end{align}
Then we have 
\begin{align}\diam^\delta(X)\leq 
2\symradius^\delta(X)
\end{align}
by definition.  The symmetric radius gives a lower bound on the second moment of~$X$, as follows.
\begin{lemma}[Symmetric radius and second moment]\label{lem:secondmomentbound}
Let $X$ be an arbitrary random variable on~$\mathbb{R}$. Then 
\begin{align}
\delta\cdot \symradius^\delta(X)^2 \leq \ExpE[X^2]\ .
\end{align}
\end{lemma}
We generalize these notions as well as the bound of Lemma~\ref{lem:secondmomentbound} below to quantum states. 
\begin{proof}
Let $R>0$ be arbitrary. 
By Markov's inequality, we have 
\begin{align}
\Pr[|X|\geq R] &=
\Pr[X^2\geq R^2]\\
&\leq \frac{\ExpE[X^2]}{R^2}\ .
\end{align}
 It follows that 
\begin{align}
\Pr\left[X\in [-R,R]\right]&\geq 
\Pr\left[X\in (-R,R)\right]\\
&\geq 1-\delta 
\end{align}
whenever 
\begin{align}
R&\geq \ExpE[X^2]^{1/2}\cdot \delta^{-1/2}\ .
\end{align}
This implies that $\symradius^\delta(X)\leq \ExpE[X^2]^{1/2}\cdot \delta^{-1/2}$, which is  the claim.
\end{proof}

\subsection{Diameter of a one-mode state and squeezing\label{sec:squeezingonemodestate}}
Consider a state~$\rho\in \cB(L^2(\mathbb{R}))$ of a single mode with canonical position- and momentum operators~$Q$ and $P$, respectively.
For an observable $O$ on $L^2(\mathbb{R})$, let
\begin{align}
\Var_{\rho}(O):=\tr(\rho O^2)-\tr(\rho O)^2
\end{align}
be the variance of the measurement result when measuring the observable~$O$ in the state~$\rho$. Heisenberg's uncertainty relation
states that the corresponding standard deviations satisfy
\begin{align}
\Var_{\rho}(Q)\Var_{\rho}(P)& \geq 1/4\ .\label{eq:heisenbergrelation}
\end{align}
Eq.~\eqref{eq:heisenbergrelation} is saturated when~$\rho=\proj{\alpha}$ is a coherent state; in this case $\Var_\rho(Q)=\Var_\rho(P)=1/2$. 
A state~$\rho$ is called squeezed if 
there is some angle $\theta$ such that the rotated state 
\begin{align}
\rho(\theta):=e^{i\theta (Q^2+P^2)}\rho e^{-i\theta (Q^2+P^2)}
\end{align} 
satisfies
\begin{align}
\Var_{\rho(\theta)}(Q)&<1/2\ 
\end{align}
Correspondingly, a typically considered squeezing measure for $\rho$ is the minimal variance
\begin{align}
\xi(\rho):=\min_{\theta }\Var_{\rho(\theta)}(Q)\ .
\end{align}
Let us only consider $\theta=0$ for the following. 
Then
it is clear that the quantity
\begin{align}
\mathsf{d}(\rho):=\max\{\Var_\rho(Q),\Var_\rho(P)\}\label{eq:drhomotivation}
\end{align}
provides an upper bound on the amount of squeezing in both the $Q$- and $P$-direction: Indeed, if e.g., $\mathsf{d}(\rho)=\Var_\rho(P)$, then
we obtain $\Var_\rho(Q)\geq 1/(4\mathsf{d}(\rho))$
by Heisenberg's uncertainty relation, which means that the amount of squeezing in the $Q$-direction is limited.

Our considerations are motivated by the quantity~\eqref{eq:drhomotivation}, but we replace the notion of standard deviation by the $\delta$-diameter. We note that, as discussed in Section~\ref{sec:localization},
the corresponding quantities have similar behavior.
 Specifically, we define the following.
 \begin{definition}
 Let $\delta>0$ and let $\rho\in \cB(L^2(\mathbb{R}))$ be a state, and let $O$ be an observable on $L^2(\mathbb{R})$. 
 Then the $\delta$-diameter of~$\rho$ on the observable~$O$ is defined as
\begin{align}
\diam^\delta_\rho(O)&:=\diam^\delta(O_\rho)
\end{align}
where $O_\rho$ is the distribution of measurement outcomes when measuring~$\rho$.
Similarly, the symmetric $\delta$-radius of $\rho$ on the observable~$O$ is defined as
\begin{align}
\symradius^\delta_\rho(O)&:=\symradius^\delta(O_\rho)\ .
\end{align}
\end{definition}
Similar to Heisenberg's uncertainty relation, Dohono and Stark~\cite[Theorem 2]{dohonostark} showed that the $\delta$-diameter satisfies the following uncertainty relation
 \begin{align}
 \diam^{\delta_1}_\rho(Q)\cdot \diam^{\delta_2}_\rho(P)&\geq  (1-(\delta_1+\delta_2))^2\ 
 \end{align}
  for any $\delta_1,\delta_2>0$. This motivates the following definition.
\begin{definition}
Let $\rho\in \cB(L^2(\mathbb{R}))$ be a state. Let $\delta>0$. Then the $\delta$-diameter of $\rho$ is the quantity
\begin{align}
\diam^\delta(\rho):=\max \{\diam^\delta_\rho(Q),\diam^\delta_\rho(P)\}\ .
\end{align}
Similarly, we define the symmetric $\delta$-radius of~$\rho$
as
\begin{align}
\symradius^\delta(\rho)&:=\max\{\symradius^\delta_\rho(Q),\symradius^\delta_\rho(P)\}\ .
\end{align}
\end{definition}
By definition,  the scalar $\diam^\delta(\rho)$ tells us how concentrated~$\rho$ is in position- respectively momentum-space: The measurement outcomes~$Q_\rho$ and $P_\rho$ are given by distributions with at least~$1-\delta$ probability mass contained  in intervals of length upper bounded~$\diam^\delta(\rho)$.  Similarly, both distributions~$Q_\rho$ and $P_\rho$ have at least mass $1-\delta$ in the interval $[-\symradius^\delta(\rho),\symradius^\delta(\rho)]$ by definition. It also follows from these definitions that \begin{align}
\diam^\delta(\rho)&\leq 2\cdot \symradius^\delta(\rho)\ .
\end{align}
Similar to the quantity~$\mathsf{d}(\rho)$ introduced above, the quantity $\diam^\delta(\rho)$ can be seen as a measure of the amount of squeezing of the state~$\rho$. On the ther hand, the quantity $\symradius^\delta(\rho)$ is less suitable as a squeezing measure because it also takes into account displacements: For example, 
the symmetric $\delta$-radius~$\symradius^\delta(\proj{\alpha})$  of a coherent state~$\ket{\alpha}$ increases with $|\alpha|$.  In contrast, $\diam^\delta(\proj{\alpha})=\diam^\delta(\proj{\mathsf{vac}})$
is equal to the diameter of the vacuum state. (This matches the commonly used notion of squeezing where a displaced vacuum state is not squeezed.)

Let us briefly discuss some basic properties of this quantity.  Recall the spectral projections~$\Pi_{[R_1,R_2]}$ and~$\widehat{\Pi}_{[R_1,R_2]}$ of the position- and momentum operators~$Q$ and~$P$ (on~$L^2(\mathbb{R})$) introduced in Section~\ref{sec:momentlimitfunctionsonemode}
for~$R_1<R_2$. Clearly, an equivalent definition of the $\delta$-diameter of $\rho$ and its symmetric $\delta$-radius is  
\begin{align}
\diam^\delta(\rho)&:=
\max\left(
\inf\left\{
R_2-R_1
\ |\ 
\substack{R_1<R_2\\
\tr\left(\Pi_{[R_1,R_2]}\rho\right)\geq 1-\delta}
\right\}, \inf\left\{
R'_2-R'_1
\ |\ 
\substack{
R_1'<R_2'\\
\tr\left(\widehat{\Pi}_{[R'_1,R'_2]}\rho\right)\geq 1-\delta
}
\right\}
\right)\\ 
\symradius^\delta(\rho)&:=
\inf\left\{R\ |\ 
\tr\left(\Pi_{[-R,R]}\rho\right)\geq 1-\delta\textrm{ and }\tr\left(\widehat{\Pi}_{[-R,R]}\rho\right)\geq 1-\delta\right\}\ .
\end{align}
We   interested in the energy
\begin{align}
\energy(\rho)&:=\tr\left((Q^2+P^2)\rho\right)
\end{align} of a state~$\rho$ with respect to the harmonic oscillator Hamiltonian $Q^2+P^2$. 
We find the following:
\begin{lemma}[Energy lower bound in terms of symmetric radius]\label{lem:energylowerboundonemode}
 Let $\delta>0$ and let $\rho\in \cB(L^2(\mathbb{R}))$ be a state.
 Then
 \begin{align}
 \delta\cdot \symradius^\delta(\rho)^2 &\leq \energy(\rho)\ .
 \end{align}
\end{lemma}
\begin{proof}
Consider the random variables~$Q_\rho$ and $P_\rho$ obtained by measuring the $Q$- and the $P$-quadrature, respectively. By definition of the symmetric $\delta$-radius, we 
have  $\symradius^\delta(\rho)\in \{\symradius^\delta(Q_\rho),\symradius^\delta(P_\rho)\}$. 
Assume without loss of generality (the other case is treated similarly) that 
\begin{align}
\symradius^\delta(Q_\rho)&=\symradius^\delta(\rho)\ .\label{eq:equalityvsm}
\end{align} 
Then we obtain
\begin{align}
\delta\cdot \symradius^\delta(\rho)^2&\leq \ExpE\left[Q_\rho^2\right]\\
&\leq \ExpE\left[Q_\rho^2\right]+\ExpE\left[P_\rho^2\right]\\
&=\tr(Q^2\rho)+\tr(P^2\rho)\ ,
\end{align}
where the first inequality is obtained by combining Eq.~\eqref{eq:equalityvsm}
with Lemma~\ref{lem:secondmomentbound}, the second inequality trivially follows because $\ExpE[P_\rho^2]\geq 0$, and the last identity follows by definition of $Q_\rho$ and $P_\rho$. This is the claim. 
\end{proof}
We note that
since the energy
is invariant under passive phase space rotations, 
Lemma~\ref{lem:energylowerboundonemode}
also implies that 
\begin{align}
\sup_\theta \delta \cdot \symradius^\delta(e^{i\theta (Q^2+P^2)}\rho e^{-i\theta (Q^2+P^2)})^2&\leq \energy(\rho)\ .\label{eq:vdavm}
\end{align} 
Eq.~\eqref{eq:vdavm} means that the energy 
provides an upper bound on the amount of squeezing in any phase space direction. 

\subsection{Diameter a multimode state\label{sec:squeezingmultimodestate}}
Here we generalize the notion of the diameter and the symmetric radius of a state to the multimode setting (including additional qubits).

 Recall the spectral projections~$\Pi_{[R_1,R_2]}$ and~$\widehat{\Pi}_{[R_1,R_2]}$ of the position- and momentum operators~$Q$ and~$P$ (on~$L^2(\mathbb{R})$) introduced in Section~\ref{sec:momentlimitfunctionsonemode}
for~$R_1<R_2$.  We define the $\delta$-diameter and symmetric $\delta$-radius  of a state on  $m$~modes and~$r$~qubits as follows.
\begin{definition} \label{lem: def squeezing}
  Let~$m \in \mathbb{N}$ and~$r \in \mathbb{N}_0$.
Define the projections
  \begin{align}
    \begin{aligned}
    \Pi^{(m)}_{[-R,R]} &:= \Pi^{\otimes m}_{[-R,R]} \otimes I_{\mathbb{C}^2}^{\otimes r} \\
    \widehat{\Pi}^{(m)}_{[-R,R]} &:= \widehat{\Pi}^{\otimes m}_{[-R,R]} \otimes I_{\mathbb{C}^2}^{\otimes r}
    \end{aligned} \qquad \textrm{for} \qquad R>0\, .
  \end{align}
 Let~$\rho=\rho_{B_1\cdots B_mQ_1\cdots Q_r}\in \cB(L^2(\mathbb{R})^{\otimes m}\otimes (\mathbb{C}^2)^{\otimes r})$ be a state. 
  Let~$\delta>0$ and $d>0$. We call a pair  $\left(\cJ=(\cJ_\alpha)_{\alpha=1}^m, \cJ=(\cJ'_\alpha)_{\alpha=1}^m\right)$ of $m$-tuples of (closed) intervals $(d,\delta,\rho)$-valid if 
 \begin{align}
 \tr\left(\Pi(\cJ)\rho\right)&\geq 1-\delta\\
  \tr\left(\widehat{\Pi}(\cJ')\rho\right)&\geq 1-\delta
 \end{align}
 and
 \begin{align}
 |\cJ_\alpha|\leq d\textrm{ and } |\cJ'_\alpha|\leq d \textrm{ for all }\alpha\in \{1,\ldots,m\}\ . \end{align}
 Here $|[a,b]|=b-a$ denotes the length of an interval and 
 \begin{align}
 \Pi(\cJ)&=(\Pi_{\cJ_1}\otimes \cdots \otimes \Pi_{\cJ_m})\otimes I_{\mathbb{C}^2}^{\otimes r}\\
 \widehat{\Pi}(\cJ')&=(\widehat{\Pi}_{\cJ'_1}\otimes \cdots \otimes \widehat{\Pi}_{\cJ'_m})\otimes I_{\mathbb{C}^2}^{\otimes r}\ .\end{align}
 The $\delta$-radius of $\rho$ is then defined as 
    \begin{align}
  \diam^\delta(\rho):=
  \inf\left\{d>0\ |\ \exists (d,\delta,\rho)\textrm{-valid pair }(\cJ,\cJ')  \right\}\ .
  \end{align}  
  Furthermore, we call the quantity 
  \begin{align}
  \symradius^\delta\left(\rho\right)=\inf\left\{R>0\ |\ 
  \min\left\{\tr\left(\Pi_{[-R,R]}^{(m)}\rho\right),
  \tr\left(\widehat{\Pi}_{[-R,R]}^{(m)}\rho\right)\right\}
  \geq 1-\delta 
  \right\}\ \label{eq:squeezingdeltadefinitionbasic}
  \end{align}
  the symmetric $\delta$-radius of~$\rho$. 
  \end{definition}

  We note that the quantity~$\symradius^\delta(\rho)$ gives the linear size~$R$ of a cube~$[-R,R]^m\subset\mathbb{R}^{m}$ in position- respectively momentum space such that most of the support of~$\rho$ is contained within the cube. However, it only consideres cubes centered around the origin. In contrast, 
    $\diam^\delta(\rho)$ gives the sidelength~$2R$ of any cube which is a translate of~$[-R,R]^m$  and contains most of the support of~$\rho$ (in position- respectively momentum-space). We therefore again have the inequality
  \begin{align}
\diam^\delta(\rho)\leq 2\symradius^\delta(\rho)\ .\label{eq:rhocentereda}
\end{align}
It is easy to see that the reduced density operator~$\rho_{B_\alpha}$ on the $\alpha$-th mode satisfies
\begin{align}
\symradius^\delta(\rho_{B_\alpha})\leq
\symradius^\delta(\rho)\ .
 \end{align}

 The following is a quantum generalization of
 Lemma~\ref{lem:secondmomentbound}.
We define the total energy of an $m$-mode, $r$-qubit state~$\rho$ as
 \begin{align}
 \energytotal(\rho)&=\tr\left(\left(\sum_{\alpha=1}^m 
 (Q_\alpha^2+P_\alpha^2)\right)\rho\right)\ .
 \end{align}
 \begin{lemma}[Symmetric radius and total energy] \label{lem:totalenergysqueezing}
\label{lem:symradiusenergy} Let $\rho\in \cB(L^2(\mathbb{R})^{\otimes m}\otimes (\mathbb{C}^2)^{\otimes r})$ be a state. Then 
\begin{align}
\delta\cdot \symradius^\delta(\rho)^2 \leq \energytotal(\rho)\ . \label{eq:energy1}
\end{align}
In particular, we have 
\begin{align}
  \delta\cdot \symradius^\delta(\rho)^2/m \le  \energy(\rho)\, . \label{eq:energy2}
\end{align}
\end{lemma}
\begin{proof}
  Claim~\eqref{eq:energy2} follows immediately from Claim~\eqref{eq:energy1} because we have by definition $\energytotal(\rho) \le m \cdot \energy(\rho)$.

  Let $R>0$. 
Let $(X_1,\ldots,X_m)$ be the vector of outcomes when measuring the commuting observables~$Q_1,\ldots,Q_m$ in the state~$\rho$.
Then
\begin{align}
\Pr\left[\exists \alpha\in \{1,\ldots,m\}: X_\alpha\not\in [-R,R]\right]
\leq \sum_{\alpha=1}^m \Pr\left[X_\alpha\not\in [-R,R]\right]
\end{align}
by the union bound. 
By Markov's inequality we have 
\begin{align}
 \Pr\left[X_\alpha\not\in [-R,R]\right]&\leq 
 \Pr\left[|X_\alpha|\geq R\right]\leq \ExpE[X_\alpha^2]/R^2\ .
\end{align}
But $\ExpE[X_\alpha^2]=\tr(Q_\alpha^2\rho)$, hence it follows that
\begin{align}
\tr\left(
\left(\Pi_{[-R,R]}^{\otimes m}\otimes I_{\mathbb{C}^2}^{\otimes r}\right)\rho\right)
&=\Pr\left[(X_1,\ldots,X_m)\in [-R,R]^m\right]\\
&\geq 1-\tr\left(\left(\sum_{\alpha=1}^m Q_\alpha^2\right)\rho\right)/R^2\\
&\geq 1-\energytotal(\rho)/R^2\ .
\end{align}
Analogous reasoning gives
\begin{align}
\tr\left(
\left(\widehat{\Pi}_{[-R,R]}^{\otimes m}\otimes I_{\mathbb{C}^2}^{\otimes r}\right)\rho\right)
&\geq 1-\energytotal(\rho)/R^2\ .
\end{align}
In particular, both quantities are greater than or equal to~$1-\delta$
for any $R>0$ satisfying
\begin{align}
R\geq \frac{\energytotal(\rho)^{1/2}}{\sqrt{\delta}}\ .
\end{align}
This implies that
\begin{align}
\symradius^\delta(\rho)\leq \frac{\energytotal(\rho)^{1/2}}{\sqrt{\delta}}\ .
\end{align}

\end{proof}

\section{Lower bounds on the amount of energy required}\label{sec:conversebounds}
In this section we derive  a lower bound on the maximal amount of energy of any member of a family of orthogonal states. This lower bound only depends on the size of the family, the number of modes and the the number of qubits. It is formulated in terms of the symmetric radius.

\begin{theorem}[Radius-dimension bound for families of orthonormal states on~$L^2(\mathbb{R})^{\otimes m}\otimes (\mathbb{C}^2)^{\otimes r}$] \label{thm:squeezingdimmulti-mode}
  Let~$d,m\in \mathbb{N}$ and~$r\in\mathbb{N}_0$. Let~$\{\phi_j\}_{j=0}^{d-1} \subset L^2(\mathbb{R})^{\otimes m} \otimes (\mathbb{C}^2)^{\otimes r}$ be an orthonormal family. Let~$\delta\in (0,1/9)$.
Then 
\begin{align}
    \max_{j \in \{0,\dots,d-1\}} \symradius^\delta(\phi_j) \ge \sqrt{\frac{\pi}{4}} \cdot \left(\frac{d(1 - 3 \sqrt{\delta})}{2^{r}}\right)^{1/(2m)}\, .
\end{align}
\end{theorem}

\begin{proof} 
 We first show the claim for the special case~$m=1$ and~$r=0$.
  Set~$R = \max_{j \in \{0,\dots,d-1\}} \symradius^\delta(\phi_j)$. Define the operator~$K: L^2(\mathbb{R}) \rightarrow L^2(\mathbb{R}),\ f\mapsto K(f)$ with  
  \begin{align}
      K(f)(x) = \int k(x,y) f(y) dy\, ,
  \end{align} where we use the integral kernel~$k: \mathbb{R}^2 \rightarrow \mathbb{R}$ defined as 
  \begin{align}
      k(x,y) = \begin{dcases} 
          \frac{\sin(2R(x-y))}{\pi (x-y)} \cdot \chi_{[-R,R]}(x) \cdot \chi_{[-R,R]}(y) \quad\ \textrm{if}\quad x\neq y\\
          \frac{2R}{\pi} \qquad \qquad \textrm{else}
      \end{dcases}\, . 
  \end{align}
  Here~$\chi_{[-R,R]}$ denotes the characteristic function of the interval~$[-R,R]$.
  In particular the integral kernel~$k$ is compactly supported and its restriction to the set~$[-R,R]^2$ is symmetric, positive-definite (which can be seen using Bochner's theorem, see Ref.~\cite{reed1975methods}) and continuous. 
  By Mercer's theorem, see e.g. Ref.~\cite{riesz1990functional},~$K$ is trace class and its trace is  
  \begin{align} \label{eq:traceK}
      \tr\, K = \int k(x,x) dx = \frac{4R^2}{\pi}\, .
  \end{align}
  Define the operators~$\Pi = \Pi_{[-R,R]}$ and~$\widehat{\Pi} = \widehat{\Pi}_{[-R,R]}$.
  It is easy to check that (see e.g.~\cite[Eq.~(21)]{FourierUncertainty})
  \begin{align}
      \langle f, K f\rangle = \| \widehat{\Pi} \, \Pi f\|^2 \qquad \textrm{ for all} \qquad f \in L^2(\mathbb{R})\, , \label{eq:Kproj}
  \end{align} which is equivalent to the identity~$K = \Pi \, \widehat{\Pi} \,\Pi$.
  In particular,~$K$ is a positive semidefinite operator.
  Moreover, by definition of~$R$ we have~$\symradius^\delta(\phi_j) \le R~$ for all~$j \in \{0,\dots,d-1\}$ and thus
  \begin{align} \label{eq:boundproj}
      \|\Pi \phi_j\|^2 \ge 1 - \delta  \qquad \textrm{and} \qquad
      \|\widehat{\Pi} \phi_j\|^2 \ge 1 - \delta \qquad \textrm{for all} \qquad j \in \{0,\dots, d-1\}\, .
  \end{align}
  It follows that
  \begin{align}
      \| \widehat{\Pi} \, \Pi \phi_j \|^2 &=  \| \widehat{\Pi} (\phi_j - \Pi \phi_j) - \widehat{\Pi} \phi_j\|^2 \\
      &= \| \widehat{\Pi} (\phi_j - \Pi \phi_j)\|^2 -  2\mathsf{Re}\, \langle \widehat{\Pi} (\phi_j - \Pi \phi_j), \widehat{\Pi} \phi_j \rangle + \| \widehat{\Pi} \phi_j \|^2 \\
      &\ge \| \widehat{\Pi} \phi_j \|^2 - 2 |\langle \widehat{\Pi} (\phi_j - \Pi \phi_j), \widehat{\Pi} \phi_j\rangle|\\
      & \ge \| \widehat{\Pi} \phi_j \|^2 - 2 \|\widehat{\Pi} (\phi_j - \Pi \phi_j)\| \\
      &\ge \| \widehat{\Pi} \phi_j \|^2  - 2\|\phi_j - \Pi \phi_j\|\, , \label{eq:boundprojproj}
  \end{align}
  where we used the Cauchy-Schwarz inequality and~$\|\Pi \phi_j\|\le 1$ in the forth line. The fifth line follows from~$\|\widehat{\Pi}\| \le 1$. Note that Eq.~\eqref{eq:boundproj} together with the Pythagorean theorem implies that~$\|\phi_j - \Pi \phi_j\|^2\le \delta$. This together with~$\|\widehat{\Pi}\phi_j\|^2 \ge 1-\delta$ (see Eq.~\eqref{eq:boundproj}) and Eq.~\eqref{eq:boundprojproj} gives  
  \begin{align}
      \| \widehat{\Pi} \, \Pi \phi_j \|^2 &\ge 1 - \delta - 2 \sqrt{\delta} \ge 1 - 3\sqrt{\delta}\, . \label{eq:projproj2}
  \end{align} 
  Finally, using Eqs.~\eqref{eq:projproj2} and~\eqref{eq:Kproj} we find 
  \begin{align}
      d \cdot (1 - 3\sqrt{\delta}) &\le \sum_{j=0}^{d-1} \| \widehat{\Pi} \, \Pi \phi_j\|^2= \sum_{j=0}^{d-1} \langle \phi_j, K \phi_j \rangle \le  \tr\, K = \frac{4R^2}{\pi}\, .
  \end{align} The second inequality follows from the fact that~$K$ is positive semidefinite. The last identity follows from Eq.~\eqref{eq:traceK}.
  This implies the claim.

  Next, we prove the general case, i.e.,~$m\in \mathbb{N}$ and~$r \in \mathbb{N}_0$ arbitary.
  Again set~$R = \max_{j \in \{0,\dots,d-1\}} \symradius^\delta(\phi_j)$.
    
    Define the projections 
    \begin{align}
    \begin{matrix}
      \Pi^{(m)} &= & \Pi_{[-R,R]}^{\otimes m} \otimes I_{\mathbb{C}^2}^{\otimes r}\\
       \widehat{\Pi}^{(m)} & = & \widehat{\Pi}_{[-R,R]}^{\otimes m} \otimes I_{\mathbb{C}^2}^{\otimes r}
    \end{matrix}\, . \label{eq:defPim}
    \end{align}

It follows from the definition of~$R$ that  
 \begin{align}
    \|\Pi^{(m)} \phi_j\|^2 \ge 1 -\delta \qquad 
    \textrm{and} \qquad \|\widehat{\Pi}^{(m)} \phi_j\|^2\ge 1- \delta \quad \textrm{for all} \quad j\in \{0,\dots,d-1\},\, s \in \{1,\dots,m\}\, .
 \end{align} 
 Let~$K$ be the operator on~$L^2(\mathbb{R})$ used in the first part of the proof, i.e.,~$K = \Pi_{[-R,R]}\, \widehat{\Pi}_{[-R,R]} \, \Pi_{[-R,R]}$. Define the operator
 \begin{align}
    K^{(m)} = \Pi^{(m)} \, \widehat{\Pi}^{(m)} \, \Pi^{(m)} = K^{\otimes m} \otimes I_{\mathbb{C}^2}^{\otimes r}\, . \label{eq:defoverlineK}
 \end{align}
 Then also~$K^{(m)}$ is a trace-class and positive semidefinite operator and we have 
 \begin{align}
    \tr \,K^{(m)} &= \left(\tr\, K\right)^{ m} \cdot \left(\tr \,I_{\mathbb{C}^2}\right)^{r} \\
                        &=  \left( \frac{4R^2}{\pi}\right)^m \cdot 2^r \, . \label{eq:trL}
 \end{align}
 The claim then follows by the identical arguments as the special case~$m=1$ and~$r=0$.
\end{proof}

As a corollary to Theorem~\ref{thm:squeezingdimmulti-mode}, let us specialize to the case of a constant number of physical qubits, and $2^n$~elements in the orthonormal family of states (corresponding to $n$~encoded qubits). We then have the following.
\begin{corollary}[Maximum radius and energy in an~$n$-qubit encoding]\label{cor:constantnumberofmodessqueezinglowerbound}
Let~$\delta\leq 1/36$. Let~$\{\phi_j\}_{j=0}^{2^n-1} \subset L^2(\mathbb{R})^{\otimes m} \otimes (\mathbb{C}^2)^{\otimes r}$ be an orthonormal family consisting of~$2^n$~states. Set
\begin{align}
s(n)&:=\max_{j \in \{0,\dots,2^n-1\}} \symradius^\delta(\phi_j)\label{eq:snfirst}\\
E(n)&:=\max_{j \in \{0,\dots,2^n-1\}} \energy(\phi_j)\ .
\end{align}
Assume that~$r=O(1)$. 
 Then 
\begin{align}
s(n)& =\Omega(2^{n/(2m)})\ .\label{eq:snclaimmain}
\end{align}
The maximum energy is at least 
\begin{align}
E(n)&=\Omega(2^{n/m}/m)\ .\label{eq:enmbav}
\end{align}
In particular, we obtain
\begin{enumerate}[(i)]
\item 
$E(n)=\exp(\Omega(n))$ for~$m=\Theta(1)$.
\item
$E(n)=\exp(\Omega(n^{1-\alpha}))$ for~$m=\Theta(n^\alpha)$ with~$\alpha\in (0,1)$. 
 \item \label{it:energylower}
$E(n)=\Omega(1)$ if~$m=\Theta(n)$.
\end{enumerate}
\end{corollary}
\begin{proof}
 Defining 
\begin{align}
\nu:=(1-3\sqrt{\delta})/2^r 
\end{align}
we have 
\begin{align}
\nu\geq 2^{-(r+1)}\label{eq:nulowerboundv}
\end{align}
 for any~$\delta\leq 1/36$, the bound  of Theorem~\ref{thm:squeezingdimmulti-mode} for~$d=2^n$ 
implies that 
\begin{align}
    \max_{j \in \{0,\dots,2^n-1\}} \symradius^\delta(\phi_j) \ge \sqrt{\frac{\pi}{4}} \cdot  2^{((n-r)-(r+1))/(2m)}\ .
\end{align}
Because 
$2^{((n-r)-(r+1))/(2m)}=2^{n/(2m)}\cdot 2^{-(r+1/2)/m}\geq 2^{n/(2m)}\cdot 2^{-(r+1/2)}$ 
we obtain
\begin{align}
    \max_{j \in \{0,\dots,2^n-1\}} \symradius^\delta(\phi_j) 
    &\geq C\cdot 2^{n/(2m)}
\end{align}
where~$C=\sqrt{\frac{\pi}{4}}2^{-(r+1/2)}$. This implies Claim~\eqref{eq:snclaimmain} for~$r=O(1)$. Claim~\eqref{eq:enmbav} follows from Claim~\eqref{eq:snclaimmain} and Lemma~\ref{lem:symradiusenergy}.
We are left to show Claim~\eqref{it:energylower}. We note that in the regime $m = \Theta(n)$ Eq.~\eqref{eq:enmbav} implies that $E(n) = \Omega(1/n)$. 
A tighter bound can be derived from the fact that 
the state $\ket{\Phi^{(0)}}= \ket{\vac}^{\otimes m} \otimes \ket{0}^{\otimes r}$ is the unique ground state of the Hamiltonian $\sum_{\alpha=1}^m (Q_\alpha^2 + P_\alpha^2)$ with $\energytotal(\Phi^{(0)}) = m$. 
This implies the claim as 
\begin{align}
  E(n) \ge \max_{j \in \{0,\dots, 2^n -1\}} \energytotal(\phi_j)/m  \ge \energytotal(\Phi^{(0)})/m = 1\, .
\end{align}
\end{proof}

\bibliographystyle{unsrturl}
\bibliography{q}

\end{document}